\documentclass[acmsmall,nonacm]{acmart}
\usepackage[utf8]{inputenc}
\usepackage[T1]{fontenc}
\usepackage{microtype}
\usepackage{amsmath,makecell}
\usepackage{stmaryrd}

\usepackage{bm}
\usepackage{multirow, bigdelim} 
\usepackage{tikz}
 \usepackage{subcaption}
 \usepackage[capitalize]{cleveref}
\usepackage[
       disable 
]
{todonotes}

  \usepackage{relsize}
     \usetikzlibrary{quantikz2}

\usepackage{wrapfig}
\usepackage{bussproofs}
\EnableBpAbbreviations%
\usepackage{enumitem}
\usepackage{siunitx}

\usepackage{xspace}
\usepackage[english]{babel}
\newtheorem{lemma}{Lemma}
\newtheorem{runex}{Running example}

\newtheorem{illus-case}{Illustration}
\newtheorem{notation}{Notation}
\newtheorem{corollary}{Corollary}
\newtheorem{remark}{Remark}

\theoremstyle{definition}
\newtheorem{definition}{Definition}
  \tikzstyle{bl}=[ rectangle, fill =red, minimum width=.4cm, minimum height=.4cm]
  \tikzstyle{un}=[ rectangle, fill =red, minimum width=.4cm, minimum height=.4cm]
  \tikzstyle{mun}=[ rectangle, fill =blue, minimum width=.4cm, minimum height=.4cm]
  \tikzstyle{ghost}=[ rectangle, minimum width=.4cm, minimum height=.4cm]
\tikzset{mymeter/.append style={draw, fill=white,inner sep=6, rectangle,
        font=\vphantom{\tiny a},
        minimum width=20,
        line width=.8,
        path picture={
            \draw[black] ([shift={(.1,.15)}]path picture bounding box.south
            west) to[bend left=50] ([shift={(-.1,.15)}]path picture bounding
            box.south east);\draw[black,-latex] ([shift={(0,.1)}]path picture
            bounding box.south) -- ([shift={(.3,-.1)}]path picture bounding 
box.north);}}}

\newcommand{\myparagraph}[1]{\medskip \noindent \textbf{#1.}}

\definecolor{yell}{RGB}{246,187,0}
\definecolor{blu}{RGB}{6,39,148}
\definecolor{re}{RGB}{148,6,39}
\definecolor{gree}{RGB}{10,150,10}
\definecolor{blugree}{RGB}{8,95,120}
\definecolor{reyell}{RGB}{190,75,25}

\newcommand{\fv}[1]{
{#1}}
\newcommand{\finv}[1]{\fv{#1}}
\newcommand{\disc}[1]{}

\newcommand{\pasuminline}[3]{\ensuremath{\langle #1,#3\cdot#2\rangle}}

\definecolor{mygrey}{rgb}{0.9, 0.95, 0.95}
\definecolor {iron}{RGB}{72, 73, 75}
\definecolor{yell}{RGB}{246,187,0}
\definecolor{yell}{RGB}{246,187,0}
\definecolor{blu}{RGB}{0,20,100}
\definecolor{re}{RGB}{148,6,39}
\definecolor{gree}{RGB}{0,100,0}
\definecolor{rgree}{RGB}{74,53,20}
\definecolor{orange}{RGB}{255,178,102}
\definecolor{tgrey}{RGB}{211,211,211}
\newcommand{\tyes}{\cellcolor{gree!50}}
\newcommand{\tno}{\cellcolor{re!80}}
\newcommand{\tquite}{\cellcolor{rgree!72}}

\newcommand{\nodata}{---}

\newcommand{\DO}{\ensuremath{\mathcal{DO}}\xspace}
\renewcommand{\xi}{\ensuremath{\mathcal{V}}\xspace}

\newcommand{\mncom}[1]{\todo[inline, color=iron!40]{MN: #1\xspace}}

\newcommand{\oad}{\ensuremath{\texttt{Oad}}\xspace}

\newcommand{\C}{\ensuremath{\texttt{c}}\xspace}

\newcommand{\X}{\ensuremath{\texttt{X}}\xspace}
\newcommand{\Qu}{\ensuremath{\texttt{Qu}}\xspace}

\newcommand{\qu}{\ensuremath{\texttt{qu}}\xspace}

\newcommand{\Cl}{\ensuremath{\texttt{Cl}}\xspace}

\newcommand{\pvar}{\ensuremath{\texttt{Pvar}}\xspace}

\newcommand{\n}{\ensuremath{\texttt{s}}\xspace}
\newcommand{\initqb}{\huge{\bm{\shortmid}}\xspace}
\newcommand{\tinitqb}{\footnotesize{\bm{\shortmid}}\xspace}

\newcommand{\y}{\ensuremath{\texttt{y}}\xspace}

\newcommand{\An}{\ensuremath{\texttt{a}}\xspace}

\newcommand{\Reals}{\ensuremath{\texttt{r}}\xspace}

\newcommand{\bsupport}{\text{\ttfamily\fontseries{b}\selectfont su}}

\newcommand{\creg}{\ensuremath{\operatorname{\texttt{creg}}\xspace}}
\newcommand{\qreg}{\ensuremath{\operatorname{\texttt{qreg}}\xspace}}
\newcommand{\cqreg}{\ensuremath{\operatorname{\texttt{(c/q)reg}}\xspace}}

\newcommand{\steps}[1]{{\ensuremath{\color{gree}\bm #1\color{black}}}\xspace}

\newcommand{\N}{\ensuremath{\texttt{s}}\xspace}
\newcommand{\F}{\ensuremath{\texttt{o}}\xspace}
\newcommand{\I}{\ensuremath{\texttt{i}}\xspace}
\newcommand{\J}{\ensuremath{\texttt{j}}\xspace}

\newcommand{\B}{\ensuremath{\texttt{b}}\xspace}
\newcommand{\bkappa}{\bm{\texttt{b}}}

\makeatletter
\def\moverlay{\mathpalette\mov@rlay}
\def\mov@rlay#1#2{\leavevmode\vtop{%
   \baselineskip\z@skip \lineskiplimit-\maxdimen
   \ialign{\hfil$\m@th#1##$\hfil\cr#2\crcr}}}
\newcommand{\charfusion}[3][\mathord]{
    #1{\ifx#1\mathop\vphantom{#2}\fi
        \mathpalette\mov@rlay{#2\cr#3}
      }
    \ifx#1\mathop\expandafter\displaylimits\fi}
\makeatother

\newcommand{\bN}{\text{\ttfamily\fontseries{b}\selectfont s}}
\newcommand{\bF}{\text{\ttfamily\fontseries{b}\selectfont o}}

\newcommand{\bB}{\text{\ttfamily\fontseries{b}\selectfont b}}
\renewcommand{\cos}{\ensuremath{\operatorname{\texttt{cos}}\xspace}}
\renewcommand{\sin}{\ensuremath{\operatorname{\texttt{sin}}\xspace}}
\renewcommand{\arccos}{\ensuremath{\operatorname{\texttt{arccos}}}}
\renewcommand{\arcsin}{\ensuremath{\operatorname{\texttt{arcsin}}}}
\newcommand{\U}{\ensuremath{\mathcal{G}}}
\newcommand{\Q}{\ensuremath{\texttt{q}}\xspace}

\renewcommand{\O}{\ensuremath{\texttt{o}}\xspace}

\newcommand{\cmem}{\ensuremath{\operatorname{\texttt{cmem}}}}

\newcommand{\qmem}{\ensuremath{\operatorname{\texttt{qmem}}}}

\newcommand{\de}[2]{\ensuremath{\frac{#1}{\sqrt{2^{#2}}}}\xspace}

\newcommand{\cnot}{\ensuremath{\texttt{cnot}}\xspace}

\DeclareMathOperator{\psadd}{\ensuremath{\uplus}\xspace }

\DeclareMathOperator{\pstens}
{\ensuremath{\otimes}\xspace }
\newcommand\uminus{\mathbin{\ooalign{$\cup$\cr%
   \hfil\raise0.42ex\hbox{$\scriptstyle -$}\hfil\cr}}}

\newcommand{\condition}{\ensuremath{\textsf{cont}}\xspace}

\newcommand{\probadef}{\ensuremath{\myrule{\mathbb{P}}}\xspace}

\newcommand{\intounitm}[2]{\ensuremath{\omega_n^{-k}}}

\newcommand{\norm}[1]{\ensuremath{\left| #1\right|}}

\newcommand{\qbrickCORE}{\ensuremath{\textsc{Qbricks}}}
\newcommand{\hqbricks}{H\qbrick}
\newcommand{\qbrick}{\qbrickCORE\xspace}

\newcommand{\prequiv}{\ensuremath{\equiv^{\P}}\xspace}
\newcommand{\equivr}{\ensuremath{\equiv_{\texttt{r}}}\xspace}
\newcommand{\prinduce}{\ensuremath{\Rrightarrow^{\P}}\xspace}
\newcommand{\prinducer}{\ensuremath{\preccurlyeq}_{\texttt{r}}\xspace}

\newcommand{\lift}[1]{\ensuremath{\overline{#1}}\xspace}

\newcommand{\zgate}{\begin{tikzpicture}\node[draw, rectangle](h)at(0,0){Z};\end{tikzpicture}\ }

\newcommand{\semantics}[1]{\ensuremath{\llbracket #1 \rrbracket}}

\renewcommand{\vec}[1]{\Vec{#1} }
\newcommand{\toint}[1]{\inttype(#1) }

\newcommand{\sqnorm}[1]{\ensuremath{\norm{#1}^2}\ }

\newcommand{\init}{\ensuremath{\textbf{Init}}\ }

\newcounter{magicrownumbers}

\renewcommand{\ket}[2]{\ensuremath{\vert #1 \rangle}_{#2}}

 \renewcommand{\bra}[2]{\ensuremath{\ \langle #1 \vert}_{#2}}
\newcommand{\ketbra}[2]{\ket{#1}{}\!\!\bra{#2}{}}

\newcommand{\card}[1]{\ensuremath{\# (#1)}}

\newcommand{\sem}[2]{\ensuremath{\llbracket #1\rrbracket_{#2} }}

\newcommand{\mycomment}[1]{}

\newcommand{\inttype}{\texttt{int}}

\newcommand{\program}{\prog}

\newcommand{\stack}[1]{\ensuremath{\uparrow{#1}}}

\newcommand{\cond}[1]{\left\{{#1}\right\}}

\newcommand{\hoare}[3]{\cond{#1}{#2}\cond{#3}}

\newcommand{\pasum}[3]{\ensuremath{\left\langle #1,#3\cdot#2\right\rangle}\xspace}
\newcommand{\pasumnormzero}[2]{\ensuremath{\left\langle #1,#2\right\rangle}\xspace}

\newcommand{\Skip}{\ensuremath{\textbf{Skip}}\xspace}

\newcommand{\case}[3]{\ensuremath{#1 ?\left( #2 ; #3 \right) }\xspace}

 \newcommand{\get}[1]{\ensuremath{\lceil#1\rceil}\xspace}

\newcommand{\q}{\ensuremath{\texttt{Q}}\xspace }
\newcommand{\cre}{\ensuremath{\texttt{C}}\xspace }
\newcommand{\prog}{\ensuremath{\texttt{Pr}}\xspace }
\newcommand{\qprog}{\ensuremath{\texttt{qP}}\xspace }

\renewcommand{\qprog}{\ensuremath{\texttt{qP}}\xspace }
\newcommand{\cb}{\ensuremath{\texttt{cb}}\xspace }

\newcommand{\length}{\ensuremath{\sharp}}

\newcommand{\Init}{\ensuremath{\textbf{Init}}\xspace}

\newcommand{\Apply}{\ensuremath{\textbf{Apply}}\xspace}

 \renewcommand{\for}[4]{\ensuremath{\textbf{For }#1 = #2,#3 \textbf{ do } #4 \textbf{ done} }\xspace}
 \newcommand{\forprog}[3]{\ensuremath{\textbf{For }#1 = #2,#3 \textbf{ do }}\xspace}

\newcommand{\ifthene}[3]{\ensuremath{\textbf{if }#1\textbf{ then }#2\texttt{ else }#3\textbf{ end}}\xspace}

\newcommand{\Ifthene}[3]{\ensuremath{\textbf{If }#1\textbf{ then }#2\textbf{ else }#3\textbf{ end}}\xspace}

\usetikzlibrary{decorations.markings,arrows,decorations.pathmorphing}
\usetikzlibrary{shapes.misc,matrix,backgrounds,folding,positioning}
\usetikzlibrary{shapes.geometric}
\pgfdeclarelayer{edgelayer}
\pgfdeclarelayer{nodelayer}
\pgfsetlayers{background,edgelayer,nodelayer,main}
\tikzset{every path/.style={draw=black!80, line width=0.6pt}}
\tikzstyle{every picture}=[baseline=-0.25em]
\tikzstyle{none}=[inner sep=0mm]
\tikzstyle{zxnode}=[shape=circle, minimum width=.25cm, inner sep=0.5pt, font=\footnotesize, draw=black,thick]
\tikzstyle{gn}=[zxnode ,fill=green, draw=green!10!black]
\tikzstyle{rn}=[zxnode ,fill=red, draw=red!10!black]
\tikzstyle{H box}=[rectangle,fill=yellow, draw=yellow!10!black,thick,xscale=1,yscale=1,font=\footnotesize,inner sep=1.2pt,minimum width=0.15cm,minimum height=0.15cm]
\tikzstyle{ug}=[regular polygon, regular polygon sides=3, fill=red,draw=black,inner sep = 0pt,minimum width=0.8em]
\tikzstyle{black dot}=[inner sep=0.7mm,minimum width=0pt,minimum height=0pt,fill=black,draw=black,shape=circle]
\tikzstyle{dot}=[black dot]
\tikzstyle{white dot}=[dot,fill=white]
\tikzstyle{arrow}=[decoration={markings,mark=at position 1 with
    {\arrow[scale=1.2,>=stealth]{>}}},postaction={decorate}]

\tikzstyle{glabel}=[rounded corners=0.2em,fill=green!30,inner sep=0.1em,font=\scriptsize, anchor=west, xshift=-0.3em, yshift=0,opacity=1]
\tikzstyle{rlabel}=[rounded corners=0.2em,fill=red!30,inner sep=0.1em,font=\scriptsize, anchor=west, xshift=-0.3em, yshift=0,opacity=1]
\tikzstyle{box}=[rectangle, draw=black, fill=white, inner sep=2pt]
\tikzstyle{box-no-outline}=[rectangle, draw=white, fill=white, inner sep=2pt]
\tikzstyle{circle-no-outline}=[circle, draw=white, fill=white, inner sep=0pt]
\tikzstyle{squigglearrow}=[->, line join=round, decorate, decoration={zigzag, segment length=4, amplitude=0.8, post=lineto, post length=2pt}]
\tikzstyle{divide}=[regular polygon, regular polygon sides=3, draw=black, fill=gray!50, inner sep=1.6pt, rounded corners=0.8mm]
\tikzstyle{very thick}=[-, line width=1pt]
\tikzstyle{boxedge}=[draw=gray!50]
\tikzstyle{rule-box}=[dashed, orange]

\pgfdeclarelayer{edgelayer}
\pgfdeclarelayer{nodelayer}
\pgfsetlayers{background,edgelayer,nodelayer,main}
\tikzstyle{every loop}=[]

\usetikzlibrary{circuits.ee.IEC}

\renewcommand{\P}{\ensuremath{\texttt{p}}\xspace}
\newcommand{\support}{\ensuremath{\texttt{su}}\xspace}
\newcommand{\myrule}[1]{\ensuremath{\mathbf{\color{re}#1}\color{black}}\xspace}
\newcommand{\bP}{\text{\ttfamily\fontseries{b}\selectfont p}}

 \newcommand{\hpps}{\ensuremath{\text{\ttfamily\fontseries{b}\selectfont HPS}}\xspace}
 
\newcommand{\Hps}{\ensuremath{\texttt{HPS}}\xspace}
\newcommand{\hps}{\ensuremath{\texttt{HPS}}\xspace}

\newcommand{\dsat}[3]{\ensuremath{\mathbb{P}\left(#1\right)#2 #3}\xspace}

\newcommand{\h}{\ensuremath{\texttt{h}}\xspace}
\newcommand{\reg}{\ensuremath{\texttt{reg}}\xspace}
\newcommand{\bh}{\text{\ttfamily\fontseries{b}\selectfont h}}

\renewcommand{\Vec}[1]{#1 \xspace }

\newcommand{\QReg}{\text{QReg}}
\newcommand{\CReg}{\text{CReg}}
\newcommand{\Reg}{\text{Reg}}
\newcommand{\Bool}{\text{Bool}}
\newcommand{\CBool}{\text{CBool}}
\newcommand{\QBool}{\text{QBool}}
\newcommand{\QB}{\text{qb}}
\newcommand{\Int}{\text{Int}}
\newcommand{\Prog}{\text{Prog}}
\newcommand{\QProg}{\text{QProg}}
\newcommand{\CProg}{\text{CProg}}

\newcommand{\equivdef}{\stackrel{\text{def}}{\Leftrightarrow}}

\newcommand{\bigpsadd}{\mathlarger \psadd} 
\allowdisplaybreaks

\newcommand{\histconcat}{\mathbin{+\mkern-5mu+}}
\newcommand{\bighistconcat}{\mathbin{+\mkern-5mu+}} 
\renewcommand{\kappa}{\ensuremath{\bkappa}}
\renewcommand{\vDash}{\ensuremath{\models}}

\renewcommand{\bh}{\h}

\renewcommand{\bP}{\P}
\renewcommand{\bN}{\N}
\renewcommand{\bF}{\F}
\renewcommand{\bB}{\B}
\renewcommand{\bsupport}{\support}

\begin{document}

\title{Hybrid Path-Sums for Hybrid Quantum Programs}

\author{Christophe Chareton}
\orcid{0000-0001-7113-563X}
\affiliation{%
  \institution{CEA List - Université Paris-Saclay}
  \city{Palaiseau}
  \country{France}
}
\email{christophe.chareton@cea.fr}

\author{Jad Issa}
\orcid{0009-0000-8595-728X}
\affiliation{%
  \institution{CEA List - Université Paris-Saclay}
  \city{Palaiseau}
  \country{France}
}
\affiliation{%
  \institution{Université de Lorraine, CNRS, Inria, LORIA}
  \city{F-54000 Nancy}
  \country{France}
}
\email{jad.issa@cea.fr}

\author{Mathieu Nguyen}
\orcid{0009-0002-6069-7471}
\affiliation{%
  \institution{CEA List - Université Paris-Saclay}
  \city{Palaiseau}
  \country{France}
}
\email{mathieu.nguyen@cea.fr}

\author{Nicolas Blanco}
\orcid{0009-0007-3037-0204}
\affiliation{%
  \institution{CEA List - Université Paris-Saclay}
  \city{Palaiseau}
  \country{France}
}
\email{nicalblanco@gmail.com}

\author{Sébastien Bardin}
\orcid{0000-0002-6509-3506}
\affiliation{%
  \institution{CEA List - Université Paris-Saclay}
  \city{Palaiseau}
  \country{France}
}
\email{sebastien.bardin@cea.fr}

\begin{abstract}

As quantum computing becomes an emerging reality, designing efficient quantum
programming capabilities is becoming more and more important. Particularly, the
debugging and validation of quantum programs is of paramount importance, as these programs are by definition hard to test. 
Static analysis and formal verification methods for quantum programs  
started to emerge a few years now, yet they often miss
hybrid quantum/classical reasoning facilities with, e.g., generic quantum
control, classical control and classical computation instructions. 
In this
paper, we lay out the foundations of a framework for the automated formal
verification of (full) hybrid quantum programs featuring both classical and quantum control,
measurement and hybrid data structures.  In particular, we propose: (1) a novel
symbolic representation for describing and manipulating sets of hybrid
quantum/classical states called Hybrid Path-Sums (HPS); (2) a set of rewriting rules 
providing a rich mechanism for simplifying and reasoning on these symbolic hybrid
states, and (3) a core assertion language  to specify 
equivalence of hybrid quantum programs, the satisfaction of properties on (parts of) hybrid states, 
and the extraction of probabilistic statements about the program behavior.  We prove  the
correctness of the novel symbolic representation, of its rewriting system and of
the specification system.  
Finally, we  
propose a full implementation of this framework as a dedicated symbolic execution engine for hybrid programs. We present an evaluation of a set of representative hybrid case-studies from the literature, showcasing the advantage of our approach and its efficiency compared to  state-of-the-art solutions.

\end{abstract}

 \maketitle
  \section{Introduction}%
\label{intro}

\myparagraph{Context}
 While quantum acceleration may soon become an industrial reality~\cite{arute2019quantum, acharya2024quantum}, it remains unclear how to evaluate the correctness of quantum programs. On the one hand, the only way to extract useful information from the state of a
quantum system is through \textit{{quantum measurement}}  which collapses the
state and irremediably modifies it, making  runtime debugging and monitoring 
impossible.  On the other hand, as quantum computations are probabilistic by nature, end-to-end testing requires enormous numbers of
repetitions to establish confidence in the probability distribution of the
output.  Therefore, \textit{robust, expressive,
    efficient, and tractable formal verification tools are absolute
requirements for the industrialization of quantum computing.}

Formal verification relies on mathematically establishing the correctness of a program
for any possible input.  It is already preferred over testing in critical
applications such as the formalization of mathematics~\cite{gonthier2008formal},
software architectures~\cite{Klein2010,leroy2012compcert}, and industrial
development systems~\cite{behm1999meteor, cuoq2012frama}. Its application for
quantum programs, on the other hand, is still in its infancy. Yet, 
there are already promising results 
involving symbolic representations~\cite{bauer2023symqv, amy2023complete,
amy2018towards, sistla2023symbolic}, specification language
design~\cite{chen2023autoq}, deductive
systems~\cite{ying2012floyd,ying2017invariants} and the development of
practical
tools~\cite{coqq,hietala19:verif_optim_quant_circuit,chareton2021automated,cheng2025embedding}
-- see~\cite{chareton2023formalm} for a survey. But none of these works 
supports all the core features necessary for the automated formal verification
of \textit{hybrid quantum programs}.

\myparagraph{The Problem of Hybrid Programs Formalization} 
Hybrid computation is ubiquitous in the literature, with algorithms intertwining
classical and quantum computation such as QAOA/VQE~\cite{peruzzo2014variational, farhi2014quantum, shor1994algorithms}, or embedding quantum computations
within classical \emph{repeat until success} control patterns. Furthermore, many purely
quantum algorithms such as Quantum Phase Estimation gain in performance when
split into smaller components and
hybridized~\cite{svore2013faster,kimmel2015robust}, and
even seemingly purely quantum primitives end up being hybridized to support
error detection and correction~\cite{kitaev2003fault, gottesman1997stabilizer}, which is critical in taming the intrinsic high noise levels in qubit systems. 

In hybrid architectures, a quantum processor (QPU) communicates with the
classical CPU, where measurement translates quantum superpositions into classical
non-determinism.   The challenge is therefore to design a formalism
flexible enough to handle primitives of both models as well as their
communication through measurement and quantum-,  classical-, or hybrid- control. 

\myparagraph{Goals and Challenges} We aim to provide the foundations of a
 verification framework for \textit{full} hybrid quantum programming, providing
\textit{measurement} \textbf{(M)}, \textit{hybrid data representation}
\textbf{(HD)} and \textit{classical} \textbf{(CC)} \textit{and quantum control}
\textbf{(QC)}. In addition, the representation should be \textit{compact}
\textbf{(CT)}, avoiding representations of exponential size, and provide both  \textit{equational theories} \textbf{(ET)} and \textit{refinement theories}, for
reasoning on equivalences and entailment of representations. Requirement
\textbf{(ET)} is inherently computational and is strongly tied to the possibility of
automated verification.

\medskip

\finv{\it To this end, the core intellectual challenge  is to design a compact symbolic formalism for hybrid classical/quantum processes.  A particular obstacle is to design a tractable way of representing measurement, the induced set of potentially exponentially many alternative projections, and the evolution of the probability distributions of those exponentially many states in a way that avoids exponential calculations in practice. 
} 

\medskip

A discussion of prior attempts is given in
\cref{sec:sota}. In summary, the standard \emph{density
operator representation}~\cite{qhlprover, coqq, Selinger2004} ($2^n \times 2^n$
positive self-adjoint matrices) from physics does not scale well
with the number of qubits $n$ and is unnatural for computational features such as
quantum control.

Perhaps the most successful correction to density operators to
include classical data is the \textit{cq-states}~\cite{qhlcv} where a hybrid
memory is described by pairs $(\sigma, \rho)$ of a classical state $\sigma$ and
a quantum state $\rho$ (described as a (partial) density operator) occurring
with probability $\mathrm{tr}(\rho)$. However, cq-states still suffer from the
exponential blowup in the size of the density operators (\textbf{CT}), lack
efficient equational theories to equate two states without computing them
entirely (\textbf{ET}), and, by erasing global phases, make it impossible to
express quantum control compositionally (\textbf{QC}). This last feature is
essential to enable a modular analysis of programs, where language constructs
receive a functional interpretation. 

The \emph{ZX
calculus}~\cite{coecke2017picturing,Coecke2011interacting}, on the other hand,
is efficient for analysis and optimization but is not suitable for modeling program
constructs and hybrid specifications. Finally, (plain)
\emph{path-sums}~\cite{amy2018towards,amy2023complete} provide compact symbolic
representations with good equational theories, but are limited to unitary
operations and do not support hybrid features.

\myparagraph{Insights}
To address these challenges, we \finv{extend the path-sum~\cite{amy2018towards,amy2023complete}
symbolic representation of quantum superposition to more involved structures where quantum
superposition and probability distributions are interleaved. }
  In path-sums, the unitary evolution of a quantum
system from a state $A$ to a state $B$ is given by a symbolic weighted sum of
all possible paths from $A$ to $B$. For example, the Hadamard gate $H$ applied
to a basis state $\ket{x}{}$ is represented as
$\pasum{\frac{xy}{2}}{\ket{y}{}}{\frac 1 {\sqrt 2}} = \sum_{y \in \{0, 1\}} \frac 1
{\sqrt 2} {e}^{2\pi i \frac{xy}{2}} \ket{y}{}$, where $y$ indexes the paths that are being
summed. \finv{Notice how the phase term $\frac{xy}{2}$ encodes the matrix for $H, \left(\frac{1}{\sqrt{2}}\begin{psmallmatrix} 1&1\\1&-1
\end{psmallmatrix}\right)$ with indexes $x$ and $y$.} We  extend this formalism by interpreting measurement as 
muting the symbolic expression for qubits  $\ket{y}{}$ into a classical bit which we
denote $[y]$ to obtain $\pasum{\frac{xy}{2}}{[y]}{\frac 1 {\sqrt 2}}$.
This $[\cdot]$ decoration instructs us to partition the set of paths according
to the value of $y$. In this case, we obtain
two \emph{worlds}\footnote{The reference to alternative \emph{worlds} is  inspired by
    the many-world interpretation of quantum
mechanics~\cite{everett1957relative,dewitt1970quantum}}: $\frac 1 {\sqrt 2}
[0]$ and $\frac 1 {\sqrt 2} [1]$ no longer in superposition
and each happening with probability $\left(\frac{1}{\sqrt{2}}\right)^2 = \frac 1 2$. The term $[y]$ also serves as a natural representation of classical data
(\textbf{HD}) that is accessible to the program, something which is missing
from standard path-sums and density operators~\cite{qhlprover, coqq,
Selinger2004}. 

\disc{Our vision is that the most extensive description is instead a set of
\textit{worlds}, each with its own hybrid classical/quantum state, which we
describe symbolically and compactly with \textit{hybrid path sums} (\hps). }

In
order to handle information about the structure of world splittings, hybrid
path-sums do not only explicitly encode alternative measurement results, but
also preserve this information by storing the entire history of classical
values. As such, the branching structure of worlds is preserved by (non-unitary)
operations over classical data, such as resetting a register to a constant.

To overcome the potential induced accumulation of non-relevant information, the
formalism is also provided with abstractions and  mechanisms enabling
an automated or user-controlled filtering of the relevant information out of the
representation (\cref{sec:discarding}). 

\finv{Hence, our solution 
 (i) is expressive enough to represent classical \textbf{(CC)} and quantum \textbf{(QC)} control and to derive standard algorithm specifications  and safe information discard;  and (ii) preserves the compactness (\textbf{CT}), tractability (\textbf{ET}), and elegance of original path-sums for unitary programs [3].}

\subsection*{Contributions} 
Building on the previous ideas, we set the basis for a fully integrated hybrid
deductive verification framework. 
To anchor it, we introduce  \hqbricks, a {language for hybrid classical-quantum programs}
    (\cref{sec:hqbricks-syntax}) that  supports measurement (\textbf{M}),  
    quantum and classical instructions (\textbf{HD}) with bounded loops, as control a both  quantum (\textbf{QC}) and classical
    (\textbf{CC}) conditions.\@ 
    Based on this language, we present the following contributions:
     \begin{itemize}

\item \Hps, a \textbf{compact (\textbf{CT}) symbolic representation} for hybrid
    classical/quantum  program states (\cref{hps}) that  extends the path-sum
    formalism~\cite{amy2018towards} to support measurement, classical variables
    and hybrid control, and is provided with  \emph{concretizations} interpreting it in
    terms of probabilistic hybrid classical/quantum state descriptions
    (\cref{probadefs}).  \Hps{} also soundly (\cref{theorem:sound-rewrite-gen})
    extend the path-sum \textbf{equational theory} (\textbf{ET}), enabling rewriting and
    simplification in the hybrid case. This equational theory is  completed with \textbf{local reasoning rules}, enabling partial state assertions and information discarding. 
    Local reasoning widens the scope of patterns that can be handled by
    equational reasoning.  We show it in action in several case studies of interest in \cref{sec:cases}; 
    
\item A symbolic execution framework for hybrid quantum programs based on \hps{}. It contains   a full \textbf{forward propagation semantics} (\cref{sec:operational_semantics}) for \hqbricks{} programs that is sound  relative to standard~\cite{qhlcv} density operators based semantics (\cref{theo:language-soundness}), and a dedicated 
    \textbf{core assertion language over hybrid states} (\cref{sec:static-program-analysis}), with the corresponding rules to check them on \hps{}.  
    Notably,  our assertion language can express hybrid state equivalence, hybrid predicate satisfaction and the probability of satisfying a given Boolean condition;  
    
\item  
A \textbf{complete implementation} of this framework in a working prototype symbolic executor for \hqbricks programs, 
 and a comparative  experimental evaluation over $5$  involved  cases from the literature (\cref{sec:implementation}). These case studies illustrate the  main aspects of our framework. In addition, they showcase the flexibility of the method.
Furthermore,  it turns out that our approach scales significantly better than state-of-the-art tools on  examples where comparison is available.    
\end{itemize}

In short, we provide a
\finv{foundational  framework for the
automated formal
verification of hybrid quantum programs.
We propose a symbolic representation together with a  reasoning mechanism that
is sufficient for deriving textbook specifications for standard computation
patterns from the literature (e.g., effective transmission of quantum states
through teleportation, probability of success for QPE, Repeat Until Success
schemes), and show the potential of this formalism by showcasing its 
\begin{itemize}
    \item 
 scalability:
 this is the first solution to provide symbolic
execution for hybrid classical/quantum computations scaling to thousands of
qubits on different standard case studies (teleportation, phase estimation, error correction);  
\item  flexibility: the formalism can handle diverse features such as
high-level classically controlled operations (Bounded Repeat-Until-Success), quantum information processing (Quantum Teleportation), noise modeling and error correction, as well as gate to gate quantum circuit composition (Phase Estimation), whereas in the literature these features are investigated in
separated developments. 
\end{itemize}}

\finv{
In future works, we intend to build upon the present contribution to handle essential hybrid quantum features still out of scope,  such as  
   unbounded control loops,
approximative correctness, deepened error correction analysis, integration of complex operations over classical data and expected values/resource estimation. }

\section{Prior Attempts}%
\label{sec:sota}

\def\tquite{\tikz\draw[scale=0.4,fill=black](0,.35) -- (.25,0) -- (1,.7) -- (.25,.15) -- cycle (0.75,0.2) -- (0.77,0.2)  -- (0.6,0.7) -- cycle;}
\def\tyes{\tikz\fill[scale=0.4](0,.35) -- (.25,0) -- (1,.7) -- (.25,.15) -- cycle;} 
\def\tno{} 

 \begin{wraptable}{r}{7.2cm}
\footnotesize     \caption{\fv{Comparison with the state-of-the-art}}%
    \label{SRs}

    \vspace{-1em}

        \begin{tabular}{|c|c|c|c|c|c|c|}
\cline{2-7}       \multicolumn{1}{c|}{} &\textbf{M}&
\textbf{HD} & \textbf{CC} &\textbf{QC}&\textbf{CT}& \textbf{ET}\\ \cline{2-7}\hline
Path-sums~\cite{amy2018towards}         &\tno{}  &\tno{} &\tno{}  &\tno&\tyes&\tyes\\
PPS~\cite{chareton2021automated}        &\tno{}  &\tno{} &\tno{}   &\tno&\tyes&\tno\\
PS Vilmart~\cite{Vilmart2020SOP}        &\tyes{}  &\tyes&\tno  &\tno&\tyes&\tyes\\
DOs~\cite{coqq,ying2012floyd}           &\tyes{}   &\tno{} &\tquite&\tno{}   &\tno&\tno\\
cq-states~\cite{qhlcv}                  &\tyes{}   &\tyes{}   &\tyes&\tno&\tno&\tno\\
Dirac n.~\cite{coqq,xu2024automating}   &\tyes&\tno &\tquite&\tno&\tquite&\tquite\\
ZX~\cite{Coecke2011interacting}         &\tyes{}   &\tyes{} &\tno&\tquite &\tyes&\tno\\
Qafny~\cite{li2024qafny,cheng2025embedding}         &\tyes{}  &\tyes{}&\tyes   &\tyes&\tyes&\tno\\
\hline\hline\textbf{\Hps{}} (this article)                                  &
\tyes{} &\tyes&\tyes{}  &\tyes&\tyes&\tyes\\
    \hline\end{tabular}

    \begin{center}
        \begin{tabular}{r@{}lr@{}l}
            M:\,&Measurement      &ET:\,&Equational theory \\ 
            HD:\,&Hybrid Data      &CC:\,&Classical Control \\
            QC:\,&Quantum Control  &CT:\,&Compactness 
        \end{tabular}
    \end{center}

\end{wraptable}

Formal verification of quantum programs has grown in interest for the last 5
to 10 years, with the emergence of a sustainable body of solutions. In
\cref{SRs}, we present some notable symbolic representations of quantum
and/or hybrid processes. They are compared to one another and to our \hps{}
solution in their ability to address the different items mentioned in the `Goals and challenges' paragraph above.

As introduced in \cref{intro}, the \emph{density operator formalism},
together with the superoperator interpretation of quantum processes involving
measurement~\cite{Selinger2004}, is the standard representation for quantum
states, and as such it serves as a reference
formalism for many quantum information analysis tools
 such as purity, trace
distance, entropy, etc.\ (see Nielsen and Chuang~\cite{nielsen2002quantum}, part III).

 It also provides a robust reference for
operational semantics. Nevertheless, density operators provide no means to
simplify, abstract, or compactify in any way the \emph{exponential explosion}
inherent to quantum program execution. Therefore, density operators alone fail
to provide an automatable and tractable solution for analyzing programs. To fill this
gap, they need to be provided with an additional symbolic layer targeting their
representation.
\smallskip
To do so, a recently explored
solution~\cite{bordg2021certified,coqq,xu2024automating} relies on the
formalized usage of the labeled Dirac notation~\cite{dirac1939new}.  This extra
symbolism provides more compact and intuitive data representations and
specifications, as well as additional reasoning rules. In particular, some of
these rules reason over state descriptions instead of density operators. One
therefore encodes specifications about an $n$-qubit system in a $2^n$
dimensional object, a quadratic reduction with regard to the $2^{2n}$ density
operator based representation.

As a side effect, our current proposition also generalizes this use of
symbolic state description encoded in density operators/cq-states
(\cref{def:density_maps}).

\smallskip 

An additional limitation of density operators is the use of complex vector
spaces with unnatural encodings of classical information. Solutions either
encode booleans in tensor combinations~\cite{Selinger2004} or restrict the
quantum processes under consideration to those that do not require explicit reference
to classical data~\cite{ying2012floyd,coqq}. The cq-state~\cite{qhlcv}
extension of the density operator semantics overcomes this last limitation.
However, it still lacks compactness (uses exponentially sized density operators),
a functorial interpretation of quantum control, and theory for simplification. 

\smallskip

\emph{ZX calculus}~\cite{coecke2017picturing, Coecke2011interacting} is an
alternative 
representation for instances of quantum processes and
circuits equipped with a very powerful rewriting theory. The derived
PyZX~\cite{pyzx} tool, e.g., is capable of performing verification for circuits
that feature thousands of gates in seconds, making it a powerful tool for
equivalence checking and/or circuit optimization for the design of compilation tasks.  As a drawback with respect to our purpose, ZX calculus does not
represent process inputs nor outputs, be they hybrid, classical, or quantum
data. Therefore, it is not suitable for representing and evaluating specifications
in terms of the probability of a memory being in a state satisfying a given
condition. Also, the formalism is very low-level and misses programming
constructs such as iterations or quantum control. Raising ZX verification
techniques to the parametrized case was recently addressed through the
\emph{scalable ZX} extension~\cite{carette_et_al:LIPIcs.MFCS.2019.55,
carette2021quantum}, demonstrating standard algorithm modeling but
still suffering from the difficulty of connecting an actual program to a diagrammatic
representation. 

\smallskip
In the recent \emph{Qafny} proposition~\cite{li2024qafny,cheng2025embedding},
the user chooses between several  provided
state representations to carry out a correctness proof. All these representations are
state-based, with similarities to the path-sum formalism. The solution enables
to formalize unitary computations with terminal measurement.
However, it cannot handle feedback
from the measurement results to the rest of the execution (classical control).
Therefore, hybrid programs with  proper interactions between classical and quantum data
(e.g., teleportation, error correction, repeat until success) are not supported.
\smallskip

Our present development is deeply inspired by the \emph{path-sum
calculus}~\cite{amy2018towards}, which combines compactness with an equational theory with good  computational properties and which is complete for the unitary Clifford fragment. In recent years, it has also
served as a seminal basis for different further extensions such as
parameterization~\cite{chareton2021automated,deng2024case}, unbalanced sums
~\cite{amy2023complete,deng2024case}, or
measurement~\cite{Vilmart2020SOP}. 

However, while the unbalanced treatment from~\cite{amy2023complete} seems to
pave the way for the representation of quantum control, a possibility that has not yet been exploited.  Similarly, the
treatment of quantum measurement in~\cite{Vilmart2020SOP} (via ``discarded
qubits'') generates a collapse of the quantum state system, but measurement results
are never stored. Therefore, once again, this formalism cannot handle classical
control based on measurement outcomes. In
particular, as in the case of Qafny, it is not sufficient to handle teleportation, error correction, or classically controlled loops.
Finally, while these features partly appear in different solutions, there exists
no single proposition simultaneously offering all of them to meet the requirements from \cref{SRs}.
In this article, we provide such a unified solution based on path-sums.

  \section{Background: Quantum Algorithms and Programs}%
\label{background}

\myparagraph{Quantum Computing} We adopt the quantum coprocessor
model~\cite{knill1996conventions} where a classical computer performs classical
computation while delegating to a quantum processing unit (QPU) the quantum
parts of the computation. The combined flow of classical and quantum operations
is described by \emph{hybrid circuits}. On the QPU, $n$-bit vector are replaced
by $n$-\emph{qubit} registers which can be in a linear combination
(\emph{superposition}) of all $2^n$ bit strings of length $n$ (\emph{basis
states}). Each basis state $\vec b$ is weighted by its \emph{amplitude}
$\alpha_{\vec b} \in \C$, so that the state of the register becomes
$\ket{\psi}{} = \sum_{\vec b \in {\{0,1\}}^n} \alpha_{\vec b} \ket{\vec b}{}$.
The vector  $\ket{\psi}{}$ cannot be read in full; instead, it is
\emph{measured} resulting in an outcome $\vec b$ with probability 
$|\alpha_{\vec b}|^2$, implying the \emph{normalisation} of the state: $\|
\ket{\psi}{} \| = 1$. After measurement the state changes (\emph{collapses}) to
$\ket{\vec b}{}$.  A typical execution flow consists of \emph{initializing} a
qubit into a basis state, applying \emph{unitary operations} (linear maps
preserving the inner product), measuring the state, and using the measurement
result for further classical computing and control. For a more comprehensive
introduction to quantum information and quantum computing, we refer the desirous
reader to~\cite{nielsen2002quantum}. In this work, we focus largely on the
interplay of quantum and classical computation, a notorious blind spot in the
literature which has only recently gained some attention in the context of
low-level languages~\cite{AmazonBraket, cross2022openqasm,gellerintroducing}.

\myparagraph{The Path-Sum Representation}
Our proposition relies on the path-sum
representation~\cite{amy2018towards,amy2023complete,vilmart2023rewriting,Vilmart2020SOP,deng2024case},
a discrete version of  Feynman's path integrals~\cite{feynman}. In it,
the evolution of a quantum system from a state $A$ to a state $B$ is described
by a symbolic sum of possible paths from $A$ to $B$.  Simply put,
a path-sum is a triplet $\pasum{\P}{\F}{\n}$ of symbolic expressions in boolean
variables $\vec y$ called \emph{path variables} and interpreted as indexing the
set of superposed paths the system takes to evolve from an original state
$\ket{\vec x}{}$ to a final state $\ket{\vec z}{}$, such that, with $r$ being the
number of variables in 
$\pasum{\P}{\F}{\N}$\footnote{Path-sum notations vary significantly within the literature}
\[
    \ket{x}{} \mapsto \ket{\vec z}{} = \sum_{\vec y \in \{0,1\}}^r \N(\vec x,
    \vec y) e^{2i\pi \P(\vec x,
    \vec y)} \ket{\F(\vec x, \vec y)}{}
\]

As representations of unitaries, path-sums can be composed sequentially $\circ$
and in parallel $\otimes$:
\begin{align*}
    \pasum{\P'}{ \F'}{\n'} \circ \pasum{\P}{\F}{\n} &:= 
    \pasum{\P + \P'[\F/\vec x]}{\F'[\F/\vec x]}{\n \cdot \n'[\F/\vec x]} \\
        \pasum{\P'}{\F'}{\n'} \otimes \pasum{\P}{\F}{\n} &:= \pasum{\P + \P'}{\F \otimes \F'}{\n
        \cdot \n'}
\end{align*}

As such, we can describe complex unitaries from primitive ones (\emph{gates}),
such as the Hadamard $H$ which introduces superposition, or the control not
(\text{CNOT}) which entangles two qubits by flipping the target when the control
is 1:

\begin{align*}
    H &: \ket{x_0}{}\mapsto \pasum{\frac{x_0 y_0}{2}}{\ket{y_0}{}}{\frac{1}{\sqrt{2}}}\quad &
    \text{CNOT} &: |x_0 x_1\rangle \mapsto 
    \pasum{0}{\ket{x_0 (x_0 \oplus
    x_1)}{}}{1} &
\end{align*}

However, path-sum representations are not unique. For instance, as a linear
operator $H^2 = I$, but the path-sum representation of $H^2$ is
$\pasum{\frac{y_0y_1}{2} + \frac{y_0 \F}{2}}{|y_1\rangle}{\frac 1 2}$, manifestly
different from the identity $\pasum{0}{|\F \rangle}{1}$. As such, rewrite rules
(\cref{sec:rewrite}) such as the $HH$ below are required to detect these
equivalent \hps{} and simplify them.

\[
\pasum{\P + \frac{y_0y_1}{2} + \frac{y_0 \F}{2}}{|y_1\rangle}{\frac 1 2 \n}
\xrightarrow[]{HH} \pasum{\P}{|\F\rangle}{\n}
\]

\section{Hybrid Path-sums by Example: the Quantum Teleportation Protocol}%
\label{sec:walking_through_telep}
 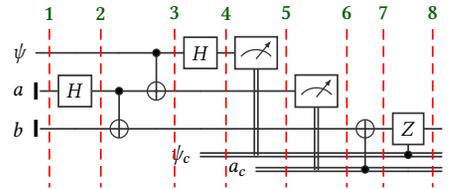
\begin{wrapfigure}{r}{6.5cm}
\centering
\scalebox{.8}{\begin{quantikz}
[row sep=0.1cm,column sep=.3cm,wire
types={q,q,q,n,n,n}]
    \lstick{$\psi$}
    &&&\ctrl{1}\slice{\steps{3}}&\gate{H}\slice{\steps{4}}&\meter{}\slice{\steps{5}}\\
     \lstick{$a$}\slice{\steps{1}}\initqb& \gate{H}\slice{\steps{2}}
     &\ctrl{1}&\targ{} &&&\meter{}\slice{\steps{6}}\\
     \initqb\lstick{$b$}&&\targ{} && &&&\targ{}&\gate{Z}&\\
 &&&&\lstick{$\psi_c$}&\setwiretype{c}\wire[u][3]{c}&&&\ctrl{-1}&\\
 &&&&&\lstick{$
 a_c$}&\setwiretype{c}\wire[u][3]{c}&\ctrl{-2}\slice{\steps{7}}&\slice{\steps{8}}&
\end{quantikz}}
    \caption{Quantum teleportation circuit}%
    \label{QTelcirc}
\Description[]{}
\end{wrapfigure}

Consider the standard quantum teleportation
protocol~\cite{bennett1993teleporting, nielsen2002quantum} 
which moves the state $\ket{x}{\psi}$ of a qubit $\psi$ into a qubit $b$.
Teleportation is an essential building block in quantum information
transmission~\cite{briegel1998quantum}, Measurement-based Quantum Computation
(MBQC)~\cite{raussendorf2003measurement}, and more.

\smallskip

The hybrid circuit for teleportation is given in \cref{QTelcirc} and we describe the induced hybrid state evolution in 
\cref{QTel}, using (i) a pseudo-formal  description of
the state, (ii) the path-sum abstraction as provided in~\cite{amy2018towards},
(iii) the cq-state~\cite{qhlcv} representation, and (iv) the \hps{} formalism,
subject of the current article.

The costs indicated in the top row of the table express the size of the representation
for the teleportation example parametrized with the number $n$ of teleported
qubits. Despite being seemingly specific to teleportation, this example is
actually representative of the sizes of the various representations in typical
applications.
Furthermore, 
while the \hps{} representation is extensible, by linearity, to any
        state of the Hilbert space, this is not the case for the representation
        in cq-states and the corresponding column is only valid for basis
        vectors.

\footnotesize
\begin{table}

    \caption{The example of quantum teleportation extended with a derivation in
    terms of cq-states. $*$ denotes the unique element of the empty product.
    The approximate sizes are given w.r.t.~the number of
    qubits $n$.}%
    \centering
 \scalebox{0.733}{
\begin{tabular}{|c|c|c|c||c|}
\hline
Labels& State $\sim 3\cdot2^{2n}$& Path-sums~\cite{amy2018towards} $\sim 3  \cdot 2^{2n}$ &cq-state~\cite{qhlcv} $\sim 2^{2\cdot(3n)}$
&\hps (this work)
$\sim 3n $
\\
\hline
\hline
\steps{1}&$\ket{x}{\psi}\ket{0}{a}\ket{0}{b}$&$\pasuminline{0}{\ket{x}{\psi}\ket{0}{a}\ket{0}{b}}{1}_{\emptyset}$&$    * \mapsto
 \begin{psmallmatrix}
     \overline{x} &0&0&0&0&0&0&0\\
     0 &0&0&0&0&0&0&0\\
     0 &0&0&0&0&0&0&0\\
     0 &0&0&0&0&0&0&0\\
     0 &0&0&0&x&0&0&0\\
     0 &0&0&0&0&0&0&0\\
     0 &0&0&0&0&0&0&0\\
     0 &0&0&0&0&0&0&0\\
 \end{psmallmatrix}$
&$\langle\ket{x}{\psi}\rangle = \pasumnormzero{ 0}{\ket{x}{\psi}\ket{0}{a}\ket{0}{b}}_{\emptyset}$
         \\
\hline

    \steps{2}&$\frac{1}{\sqrt{2}}\begin{psmallmatrix}\ket{x}{\psi}&\ket{0}{a}&\ket{0}{b}&+ \\ \ket{x}{\psi}&\ket{1}{a}&\ket{0}{b}&\end{psmallmatrix}$
             &$\pasum{0}{\ket{x}{\psi}\ket{y_0}{a}\ket{0}{b}}{\frac 1 {\sqrt 2}}_{\{\y_0\}}$
&
$ * \mapsto \frac{1}{2}\begin{psmallmatrix}
    \overline{x} &0&\overline{x}&0&0&0&0&0\\
    0 &0&0&0&0&0&0&0\\
    \overline{x} &0&\overline{x}&0&0&0&0&0\\
    0 &0&0&0&0&0&0&0\\
    0 &0&0&0&x&0&x&0\\
    0 &0&0&0&0&0&0&0\\
    0 &0&0&0&x&0& x&0\\
    0 &0&0&0&0&0&0&0\\
\end{psmallmatrix}$
&$\pasum{0}{\ket{x}{\psi}\ket{y_0}{a}\ket{0}{b}}{\frac{1}{\sqrt{2}}}_{\{\y_0\}}$\\
\hline
\steps{3}&
$\frac{1}{\sqrt{2}}\begin{psmallmatrix}\ket{x}{\psi}&\ket{x}{a}&\ket{0}{b}&+\\ \ket{x}{\psi}&\ket{x\oplus1}{a}&\ket{1}{b}&\end{psmallmatrix}$
&$\pasum{0}{\ket{x}{\psi}\ket{y_0\oplus x}{a}\ket{y_0}{b}}{\frac 1 {\sqrt 2}}_{\{\y_0\}}$
&
$* \mapsto \frac{1}{2}\begin{psmallmatrix}
    \overline{x} &0&0&\overline{x}&0&0&0&0\\
    0 &0&0&0&0&0&0&0\\
    0 &0&0&0&0&0&0&0\\
    \overline{x} &0&0&\overline{x}&0&0&0&0\\
    0 &0&0&0&0&0&0&0\\
    0 &0&0&0&0&x&x&0\\
    0 &0&0&0&0&x&x&0\\
    0 &0&0&0&0&0&0&0\\
\end{psmallmatrix}$
&$\pasum{0}{\ket{x}{\psi}\ket{y_0\oplus x}{a}\ket{y_0}{b}}{\frac{1}{\sqrt{2}}}_{\{\y_0\}}$\\
\hline
\steps{4}&
$\frac{1}{2}\begin{psmallmatrix}
&\ket{0}{\psi}&\ket{x}{a}&\ket{0}{b}&+\\ &\ket{0}{\psi}&\ket{x\oplus1}{a}&\ket{1}{b}
&+\\
&\ket{1}{\psi}&\ket{x}{a}&\ket{0}{b}&+\\ &\ket{1}{\psi}&\ket{x\oplus1}{a}&\ket{1}{b}&
\end{psmallmatrix}$

&$\pasum{\frac{y_1x}{2}}{\ket{y_1}{\psi}\ket{y_0\oplus x}{a}\ket{y_0}{b}}{\frac
1 2}_{\{\y_0,y_1\}}$
&
$* \mapsto \frac{1}{4}\begin{psmallmatrix}
    \overline{x}&0&0&\overline{x}&\overline{x}&0&0&\overline{x}\\
    0&x&x&0&0&-x&-x&0\\
    0&x&x&0&0&-x&-x&0\\
    \overline{x}&0&0&\overline{x}&\overline{x}&0&0&\overline{x}\\
    \overline{x}&0&0&\overline{x}&\overline{x}&0&0&\overline{x}\\
    0&-x&-x&0&0&x&x&0\\
    0&-x&-x&0&0&x&x&0\\
    \overline{x}&0&0&\overline{x}&\overline{x}&0&0&\overline{x}\end{psmallmatrix}$&$\pasum{\frac{y_1x}{2}}{\ket{y_1}{\psi}\ket{y_0\oplus x}{a}\ket{y_0}{b}}{\frac{1}{2}}_{\{\y_0,y_1\}}$\\
    
\hline

\steps{5}&
$\begin{smallmatrix}
&
\ldelim\{{.7}{*}\frac{1}{2}&;&&\begin{psmallmatrix}
\ket{x}{a}&\ket{0}{b}&+\\ \ket{x\oplus1}{a}&\ket{1}{b}&
\end{psmallmatrix}&[0]_{\psi_c}\rdelim\}{.7}{*}\\
 &
\ldelim\{{.7}{*}\frac{1}{2}&;&&\begin{psmallmatrix}
\ket{x}{a}&\ket{0}{b}&+\\ \ket{x\oplus1}{a}&\ket{1}{b}&
\end{psmallmatrix}&[1]_{\psi_c}\rdelim\}{.7}{*}\\
\end{smallmatrix}$

&
$\begin{smallmatrix}
\ldelim\{{.7}{*}
    \frac{1}{2}&;&\langle&0,&\frac 1 {\sqrt 2} &,&\ket{y_0\oplus x}{a}&\ket{y_0}{b}& \rangle_{\{\y_0\}}&[0]_{\psi_c}\rdelim\}{.7}{*}\\ \ldelim\{{.7}{*}
\frac{1}{2}&;&\langle&x,&\frac 1 {\sqrt 2} &,&\ket{y_0\oplus x}{a}&\ket{y_0}{b}& \rangle_{\{\y_0\}}&[1]_{\psi_c}\rdelim\}{.7}{*}\\
\end{smallmatrix}$

&
$\begin{array}{c}
0_b \mapsto \frac 1 4 \begin{psmallmatrix}
    \overline{x} &0&0&0\\
    0 &x&x&0\\
    0 &x&x&0\\
    0 &0&0&\overline{x}\\
\end{psmallmatrix}\\
1_b \mapsto  \frac 1 4
\begin{psmallmatrix}
    \overline{x} &0&0&0\\
    0 &x&x&0\\
    0 &x&x&0\\
    0 &0&0&\overline{x}\\
\end{psmallmatrix}\end{array}$
&$\pasum{\frac{y_1x}{2}}{\ket{y_0\oplus x}{a}\ket{y_0}{b}[y_1]_{\psi_c}}{\frac{1}{2}}_{\{\y_0,y_1\}}$
\\
\hline
\steps{6}&  
$
\begin{smallmatrix}
 \ldelim\{{.7}{*}\frac{1}{4}&;&&\ket{0}{b}&[0]_{\psi_c}&[x]_{a_c}&\rdelim\}{.7}{*}\\
 \ldelim\{{.7}{*}\frac{1}{4}&;&&\ket{1}{b}&[0]_{\psi_c}&[x\oplus 1]_{a_c}&\rdelim\}{.7}{*}\\
 \ldelim\{{.7}{*}\frac{1}{4}&;&&\ket{0}{b}&[1]_{\psi_c}&[x]_{a_c}&\rdelim\}{.7}{*}\\
 \ldelim\{{.7}{*}\frac{1}{4}&;&&\ket{1}{b}&[1]_{\psi_c}&[x\oplus 1]_{a_c}&\rdelim\}{.7}{*}\\
\end{smallmatrix}
$
&
$
\begin{smallmatrix}
    \ldelim\{{.7}{*}\frac{1}{4}&;&\langle 0,&1 \cdot &\ket{0}{b}\rangle_{\emptyset}&[0]_{\psi_c}&[x]_{a_c}&\rdelim\}{.7}{*}\\
 \ldelim\{{.7}{*}\frac{1}{4}&;&\langle 0,&1 \cdot &\ket{1}{b}\rangle_{\emptyset}&[0]_{\psi_c}&[x\oplus 1]_{a_c}&\rdelim\}{.7}{*}\\
 \ldelim\{{.7}{*}\frac{1}{4}&;&\langle x,&1 \cdot &\ket{0}{b}\rangle_{\emptyset}&[1]_{\psi_c}&[x]_{a_c}&\rdelim\}{.7}{*}\\
 \ldelim\{{.7}{*}\frac{1}{4}&;&\langle x,&1 \cdot &\ket{1}{b}\rangle_{\emptyset}&[1]_{\psi_c}&[x\oplus 1]_{a_c}&\rdelim\}{.7}{*}\\
\end{smallmatrix}
$
&
$
\begin{array}{c}
    (0_b, 0_{\psi_c}) \mapsto \frac 1 4 \begin{psmallmatrix}
    \overline{x} &0\\
    0 &x\\
\end{psmallmatrix}\\
(0_b, 1_{\psi_c}) \mapsto \frac 1 4 \begin{psmallmatrix}
    \overline{x} &0\\
    0 &x\\
\end{psmallmatrix}\\
(1_b, 0_{\psi_c}) \mapsto \frac 1 4 \begin{psmallmatrix}
    x&0\\
    0&\overline{x}\\
\end{psmallmatrix}\\
(1_b, 1_{\psi_c}) \mapsto \frac 1 4 \begin{psmallmatrix}
    x&0\\
    0&\overline{x}\\
\end{psmallmatrix}\\\end{array}
$
&$\pasum{\frac{ y_1x}{2}}{\ket{y_0}{b}[y_1]_{\psi_c}[y_0\oplus x]_{a_c}}{\frac{1}{2}}_{\{\y_0,y_1\}}$\\
\hline
\steps{7}&
$
\begin{smallmatrix}
 \ldelim\{{.7}{*}\frac{1}{4}&;&\ket{x}{b}&[0]_{\psi_c}&[x]_{a_c}&\rdelim\}{.7}{*}\\
 \ldelim\{{.7}{*}\frac{1}{4}&;&\ket{x}{b}&[0]_{\psi_c}&[x\oplus 1]_{a_c}&\rdelim\}{.7}{*}\\
 \ldelim\{{.7}{*}\frac{1}{4}&;&\ket{x}{b}&[1]_{\psi_c}&[x]_{a_c}&\rdelim\}{.7}{*}\\
 \ldelim\{{.7}{*}\frac{1}{4}&;&\ket{x}{b}&[1]_{\psi_c}&[x\oplus 1]_{a_c}&\rdelim\}{.7}{*}\\
\end{smallmatrix}
$
&
$
\begin{smallmatrix}
\ldelim\{{.7}{*}\frac{1}{4}&;&\langle 0,&1 \cdot &\ket{x}{b}\rangle_{\emptyset}&[0]_{\psi_c}&[x]_{a_c}&\rdelim\}{.7}{*}\\
 \ldelim\{{.7}{*}\frac{1}{4}&;&\langle 0,&1 \cdot &\ket{x}{b}\rangle_{\emptyset}&[0]_{\psi_c}&[x\oplus 1]_{a_c}&\rdelim\}{.7}{*}\\
 \ldelim\{{.7}{*}\frac{1}{4}&;&\langle x,&1 \cdot &\ket{x}{b}\rangle_{\emptyset}&[1]_{\psi_c}&[x]_{a_c}&\rdelim\}{.7}{*}\\
 \ldelim\{{.7}{*}\frac{1}{4}&;&\langle x,&1 \cdot &\ket{x}{b}\rangle_{\emptyset}&[1]_{\psi_c}&[x\oplus 1]_{a_c}&\rdelim\}{.7}{*}\\
\end{smallmatrix}
$
&
$
\begin{array}{c}
    (0_b, 0_{\psi_c}) \mapsto \frac 1 4 \begin{psmallmatrix}
    \overline{x} &0\\
    0 &x\\
\end{psmallmatrix}\\
(0_b, 1_{\psi_c}) \mapsto \frac 1 4 \begin{psmallmatrix}
    \overline{x} &0\\
    0 &x\\
\end{psmallmatrix}\\
(1_b, 0_{\psi_c}) \mapsto \frac 1 4 \begin{psmallmatrix}
    x&0\\
    0&\overline{x}\\
\end{psmallmatrix}\\
(1_b, 1_{\psi_c}) \mapsto \frac 1 4 \begin{psmallmatrix}
    x&0\\
    0&\overline{x}\\
\end{psmallmatrix}\\\end{array}
$&$\pasum{\frac{y_1 x}{2}}{\ket{ x}{b}[y_1]_{\psi_c}[y_0\oplus x]_{a_c}}{\frac{1}{2}}_{\{\y_0,y_1\}}$
\\
\hline
\steps{8}&
$
\begin{smallmatrix}
 \ldelim\{{.7}{*}\frac{1}{4}&;&\ket{x}{b}&[0]_{\psi_c}&[x]_{a_c}&\rdelim\}{.7}{*}\\
 \ldelim\{{.7}{*}\frac{1}{4}&;&\ket{x}{b}&[0]_{\psi_c}&[x\oplus 1]_{a_c}&\rdelim\}{.7}{*}\\
 \ldelim\{{.7}{*}\frac{1}{4}&;&\ket{x}{b}&[1]_{\psi_c}&[x]_{a_c}&\rdelim\}{.7}{*}\\
 \ldelim\{{.7}{*}\frac{1}{4}&;&\ket{x}{b}&[1]_{\psi_c}&[x\oplus 1]_{a_c}&\rdelim\}{.7}{*}\\
\end{smallmatrix}
$
&
$
\begin{smallmatrix}
 \ldelim\{{.7}{*}\frac{1}{4}&;&\langle0,&1 \cdot &\ket{x}{b}\rangle_{\emptyset}&[0]_{\psi_c}&[x]_{a_c}&\rdelim\}{.7}{*}\\
 \ldelim\{{.7}{*}\frac{1}{4}&;&\langle0,&1 \cdot &\ket{x}{b}\rangle_{\emptyset}&[0]_{\psi_c}&[x\oplus 1]_{a_c}&\rdelim\}{.7}{*}\\
 \ldelim\{{.7}{*}\frac{1}{4}&;&\langle0,&1 \cdot &\ket{x}{b}\rangle_{\emptyset}&[1]_{\psi_c}&[x]_{a_c}&\rdelim\}{.7}{*}\\
 \ldelim\{{.7}{*}\frac{1}{4}&;&\langle0,&1 \cdot &\ket{x}{b}\rangle_{\emptyset}&[1]_{\psi_c}&[x\oplus 1]_{a_c}&\rdelim\}{.7}{*}\\
\end{smallmatrix}
$
&
$
\begin{array}{c}
    (0_b, 0_{\psi_c}) \mapsto \frac 1 4 \begin{psmallmatrix}
    \overline{x} &0\\
    0 &x\\
\end{psmallmatrix}\\
(0_b, 1_{\psi_c}) \mapsto \frac 1 4 \begin{psmallmatrix}
    \overline{x} &0\\
    0 &x\\
\end{psmallmatrix}\\
(1_b, 0_{\psi_c}) \mapsto \frac 1 4 \begin{psmallmatrix}
    x&0\\
    0&\overline{x}\\
\end{psmallmatrix}\\
(1_b, 1_{\psi_c}) \mapsto \frac 1 4 \begin{psmallmatrix}
    x&0\\
    0&\overline{x}\\
\end{psmallmatrix}\\\end{array}
$&$\pasum{0}{\ket{ x}{b}[y_1]_{\psi_c}[y_0\oplus x]_{a_c}}{\frac{1}{2}}_{\{\y_0,y_1\}}$\\
\hline
\hline
\steps{9}&\multicolumn{2}{c|}{}
         &
$* \mapsto \left( \begin{array}{c c}
    \overline{x} &0\\
    0 &x\\
\end{array}\right)$
&
$\begin{array}{l}
    \langle\ket{x}{b}\rangle_{\emptyset} \otimes\pasum{0}{
[y_1]_{\psi_c}[y_0\oplus x]_{a_c}
}{\frac{1}{2}}_{\{\y_0,y_1\}}    
\\
\prinducer \langle\ket{x}{b} \rangle_{\emptyset}\\
\end{array}$
\\
\hline
\end{tabular}
 }
    \label{QTel}

\end{table}

\normalsize

\myparagraph{Unitary evolution}
At step $\steps{1}$, Alice and Bob initialize each a qubit in state $0$.  Then
(step $\steps{2}$), Alice applies the Hadamard gate on her qubit $a$, 
setting it in the superposition $\ket{+}{a} = \frac{1}{\sqrt{2}} (\ket{0}{} +
\ket{1}{})$ (\cref{background}). In the path-sum column, this (branching)
superposition manifests by the introduction of a fresh free boolean
\emph{path-variable} $y_0$, encoding the two new branches $\ket{0}{a}$ and
$\ket{1}{a}$ as $\ket{y_0}{a}$.  Step $\steps{3}$ then consists of a sequence of
\cnot{} gates, interpreted as an addition modulo 2 (written $\oplus$) in the
target wire (see \cref{background} again). 
At step $\steps{4}$, another Hadamard gate further
 splits each path,  encoded by the introduction of
another boolean path variable $y_1$. So far, the column \hps{} is identical to
the standard path-sum column, as the circuit is purely quantum.

\myparagraph{Measurements and Corrections}%
\label{sec:meas_and_corr}
Next, measurement is performed on the wire $\psi$ ({Step $\steps{5}$}). This
operation introduces non-determinism, according to the Born rule
(see \cref{background}): for each possible measurement result $m \in \{0,1\}$, the
initial state of the register is projected onto the subspaces characterized by
$\ket{m}{\psi}$.  Standard path-sums~\cite{amy2018towards} cannot support this
non-determinism, so we must use an extensive description of the probability
distribution as in the first two columns of line {\steps{5}}.

Importantly, however, while the set of paths is split into these two measurement scenarios, each path is individually  preserved by the
measurement.
This is
harnessed by \hps{} which encodes the measurement as a transfer of the measured
quantum value $\ket{y_1}{\psi}$ into a classical memory cell ${[y_1]}_{\psi_c}$
to denote that this expression $y_1$ must now be used to partition paths.

{Step $\steps{6}$} then corresponds to the second measurement operation
similar to Step $\steps{5}$, and {Step $\steps{7}$} (resp. $\steps{8}$) encodes
unitary correction $\oplus$ which performs a bit flip, i.e.\ transforming the state
$\ket{x}{}$ into $\ket{x\oplus 1}{}$ (resp.\ phase shift \zgate{} transforming
state $\ket{x}{}$ into ${(-1)}^x\ket{x}{}$) on the wire $b$, controlled by
measurement result  {$a_c$} (resp. ${\psi_c}$). At that stage (Step
$\steps{8}$), in the \hps{} representation, the wire $b$ clearly holds the value
$\ket{x}{}$ showing that the qubit has indeed been transferred to Bob,
regardless of the measurement results.

\myparagraph{Distribution and Discard}
Finally, step $\steps{9}$ illustrates the usage of
\hps{} as specifications. 
It consists of two steps: first, the \hps{} is
 separated into
$\pasuminline{0}{\ket{x}{b}}{1}_\emptyset$ and $\pasum{0}{ {[y_1]}_{
\psi_c}{[y_0\oplus x]}_{ a_c} }{\frac{1}{2}}_{\{\y_0,y_1\}}$ according to the
construction of the tensor product,
then, the irrelevant part of the state is formally discarded
(see \cref{sec:rewrite} ) leaving us with the original state
$\pasum{0}{\ket{x}{b}}{1}_\emptyset$ now on the qubit $b$.

\section{Hybrid Quantum Programs: The \hqbricks{} Language}%
\label{sec:hqbricks-syntax}

Consider the hybrid imperative language \hqbricks{}. It supports unitary gates
$G \in\ \mathcal G$ on qubits $\q$, measurement, quantum and classical conditioning,
bounded for loops, and oracles $f$ applied on classical data $\cre$. The
\hqbricks{} language is given by a simply typed $\lambda$-calculus with seven
base types: \QReg, \CReg, \QBool, \CBool, \Int, \QProg{}, and \CProg, all of
which are inductive\finv{, and with $\QProg$ being a subtype of $\CProg$ corresponding to unitary programs.}

\begin{definition}[\hqbricks(\U)]%
\label{def:hqbricks}
    The grammar of the seven inductive base types is as follows:
\[
\begin{array}{lcll}   

\QReg, \q & := & x_\Q \mid \q[\I: ] \mid \q[:\I] \mid \q[\I:\J] \mid \q[\I] & \\
    \CReg, \cre, \cre_1, \cre_2 &:= & x_\C \mid \cre[\I: ] \mid \cre[:\I] \mid
    \cre[\I:\J] \mid \cre[\I]& \\
    \QBool,   \QB,  \QB_1,  \QB_2 &:= &x_{\QB} \mid 0_\QB \mid 1_\QB   \mid\q \mid  \QB_1 \oplus  \QB_2 \mid  \QB_1 \land  \QB_2 \mid  \case{\QB}{\QB_1}{\QB_2}& \\
    \CBool, \cb, \cb_1, \cb_2 &:=&  x_{\cb} \mid 0_\cb \mid 1_\cb \mid \cre \mid \cb_1 \oplus \cb_2\mid \cb_1 \land \cb_2 \mid \case{\cb}{\cb_1}{\cb_2} &\\
    \Int, \I, \J &:= & x_\I \mid\texttt{const} \mid \I + \J \mid \I * \J \mid \I^\J\mid  \card{\cre}\mid  \card{\q}&\\
    \QProg, \qprog, \qprog' &:=& \textbf{Skip } \mid \qprog ; \qprog' \mid \for{x_n^I} i j {\qprog} \mid 
            \hspace{0cm} \Ifthene \QB {\qprog} {\qprog'} \\
                            &&\mid \Apply\ \texttt{G} (\q) & \\
   \Prog, \prog, \prog' &:= & \Init\ \q\mid\qprog \mid \prog ; \prog' \mid \for{x_{\texttt{c}}^{\texttt{I}}} \I \J \prog & \\
           &      & \mid \Ifthene{\cb}
    \prog {\prog'} \mid   \textbf{Measure}(\q, \cre) \mid \cre_1 := f(\cre_2)\ & 
\end{array}
\]
where, for any type identifier $\texttt{T}$, $x_{\texttt{T}}$ is a variable of
type $\texttt{T}$ and \texttt{const} is a constant integer.  
\end{definition}

In explicit terms, \QReg{} and \CReg{} are quantum and classical registers and
\QBool{} and \CBool{} are booleans that depend on values of quantum and
classical registers, respectively. This distinction is important when it comes to
control, also explaining the need to isolate a subtype \QProg{} of \Prog{},
which corresponds to purely quantum (unitary) programs, the only
programs that can be quantum controlled. The
ternary construct $\case{a}{b}{c}$ in both boolean types is equivalent to $((
a\oplus 1)\wedge b)\oplus(a\wedge c)$.  In this article, this ternary
notation is also used in \hps{} more broadly.

\subsection{Intuition of the Syntax}
The intuition is as follows: \textbf{Skip} does nothing, $\prog ; \prog'$
performs  $\prog$ followed by $\prog'$,
The bounds
of \textbf{For} loop are parameters of the program and cannot make reference to
values in the memory; i.e.\ loops can always be unraveled after instantiation of
program parameters, but before execution.

\textbf{If } $b$ \textbf{ then } $\prog$
\textbf{ else }$\prog'$\textbf{ end } performs quantum or hybrid conditioning
when $\prog$ and $\prog'$ are unitary, and classical conditioning otherwise,
$\textbf{Measure}(\q, \cre)$ measures $\q$ storing the result in $\cre$.
\hqbricks{} can embed purely classical computation via oracles $f$ in $\cre_1 :=
f(\cre_2)$.

As for unitaries, \finv{they are applied by command \Apply \texttt{G}, where $\texttt{G} \in \U$.} In the rest of this paper, we instantiate $\U$ as the set $ Z^*
= \{H , X\}\cup {\{Z_k\}}_{k\in\mathbb{N}}$, containing the Hadamard gate, the
\texttt{not} $X$ gate and rotation gates $Z_k$ on the Z-axis of angle
$\frac{2i\pi}{2^k}$ \finv{given by matrix $\begin{psmallmatrix}
    1&0\\0&e^{\frac{2i\pi}{2^k}}
\end{psmallmatrix}$}. 

Finally, note that the syntax allows using entire registers of more than one bit/qubit as
booleans. In this case, the boolean value represented is the conjunction of the values of the bits/qubits
in the register, reflecting the typical use of multiple control in quantum circuits. 

\subsection{Comparison with State-of-the-Art Formally Verified Languages}
\finv{Quantum formal verification
solutions~\cite{cheng2025embedding,li2024qafny,qhlprover}   are often built upon
a  \textit{while} language where the   measurement is never separated from the
control that is immediately performed on it. This restriction is highly limiting
on the possible implementations (e.g., teleportation or error correction are not
supported). Interestingly,  \textsc{qhl-cv}~\cite{qhlcv} extends
this approach with classical variables. We use \textsc{qhl-cv} as an anchor for the
soundness of our semantics (\cref{theo:language-soundness}).
By explicitly presenting  \hqbricks, we aim 
to highlight: (i) the benefits from embedding in a classical host language. It
enables us to integrate pure classical computation results  in \hqbricks{} using oracles; 
and (ii) the fact that, in \hqbricks, quantum data are
never addressed directly, but only through operational commands. This provides a
\textit{correct-by-construction} guarantee over  potential violations of quantum
mechanical postulates (non-cloning theorem and unitarity) without requiring a
linear type system such as  in,
e.g.,~\cite{10.1007/11417170_26,ross2015algebraic,green2013quipper}. }

\subsection{Discussion}
\finv{ In \hqbricks{},  arbitrary  classical computations are introduced as
black-box oracles stored in classical memory cells. We envision to handle these oracles with logical function names 
with adequate axioms ("logical oracles", in a sense) expressed in standard theories over classical data (e.g., integers or bit
vectors), and then to delegate the associated reasoning to the existing classical programming verification ecosystem.} 

\finv{Note also that, as opposed
to \textsc{qhl-cv}~\cite{qhlcv}, \hqbricks does not support unbounded 
loops. Handling them requires an extension of the path-sum semantics with infinite sums over paths 
 and some additional structure for the set of path-sums, so as to properly symbolically encode invariants. 
We consider these aspects as
orthogonal to the scope of this paper and leave them for future work.}

\subsection{Type System}\label{sub:type-system}

The language is a typed $\lambda$-calculus with constructors for hybrid circuits
as (closed) terms of type Prog. The typing enforces some well-formedness
criteria and raises unitarity constraints or proof obligations on the go.
In order to address questions of unitarity under quantum control, we introduce a
subtype QProg of Prog for purely unitary programs, the only ones allowed to be
quantum controlled. 

\finv{ Note that  the quantum control rule (\myrule{\textsc{Q-if}} below) raises a syntactic  precondition (ensuring non-intersection of control and target qubits and introduced with  $\vdash^*$ ). 
Additional  syntactic conditions of well-formedness include applying gates only to allocated (initialized)
registers, not re-initializing already initialized qubits, and having the same memory allocation
profile (sets of allocated/unallocated registers) in either branch of a conditional.  
In most standard cases (including in our case studies implementation, see \cref{sec:cases}), simple syntactic checks are sufficient for verifying these conditions.
  In future advanced usage,  corner cases might still appear that will  require discharging these verifications to external SMT-solvers.
}

The main language constructors rules are as follows 
 (completed with an exhaustive
rule exposition in Appendix~\ref{type-system}).

\begin{center}
\vskip 1em    
                \AxiomC{$\Gamma \vdash p : \QProg$}
                \RightLabel{\myrule{QProg}}
                \UnaryInfC{$\Gamma \vdash p : \Prog$}
            \DisplayProof\hskip 1em
            \AxiomC{$\Gamma \vdash q : \QReg$}
                \RightLabel{\myrule{Init}}
                \UnaryInfC{$\Gamma \vdash \textbf{Init}(q) : \Prog$}
            \DisplayProof\hskip 1em
                \AxiomC{$\Gamma \vdash q : \QReg$}
                \AxiomC{$\Gamma \vdash c : \CReg$}
                \RightLabel{\myrule{Measure}}
                \BinaryInfC{$\Gamma \vdash \textbf{Measure}(q, c) : \Prog$}
            \DisplayProof\vskip 1em
                \AxiomC{$U \in \mathcal G$}
                \AxiomC{$\Gamma \vdash q : \QReg$}
                \RightLabel{\textcolor{re}{\myrule{Unitary}}\color{black}}
                \BinaryInfC{$\Gamma \vdash \Apply\ U(q) : \QProg$}
            \DisplayProof\hskip 1em
                \AxiomC{$\Gamma \vdash c_1 : \CReg$}
                \AxiomC{$\Gamma \vdash c_2 : \CReg$}
                \RightLabel{\myrule{Classical}}
                \BinaryInfC{$\Gamma \vdash c_1 := f(\lceil c_2 \rceil) : \Prog$}
 \DisplayProof\vskip 1em
                \AxiomC{$\Gamma \vdash p : \Prog$}
                \AxiomC{$\Gamma \vdash q : \Prog$}
                \RightLabel{\myrule{Seq}}
                \BinaryInfC{$\Gamma \vdash p;q : \Prog$}
            \DisplayProof\hskip 1em
            \AxiomC{$\Gamma \vdash i : \Int$}
                \AxiomC{$\Gamma \vdash j : \Int$}
                \AxiomC{$\Gamma \vdash p : \Prog$}
                \RightLabel{\myrule{For}}
                \TrinaryInfC{$\Gamma \vdash \textbf{For } x_n^I = i,j \textbf{ do } \{ p\}: \Prog$}
            \DisplayProof\vskip 1em
                \AxiomC{$\Gamma \vdash b : \CBool$}
                \AxiomC{$\Gamma \vdash p : \Prog$}
                \AxiomC{$\Gamma \vdash q : \Prog$}
                \RightLabel{\myrule{C-if}}
                \TrinaryInfC{$\Gamma \vdash \Ifthene{b}{p}{q}: \Prog$}
\DisplayProof\vskip 1em
                \AxiomC{$\Gamma \vdash^* \qreg (p) \cap\qreg(b) = \emptyset$}
                \AxiomC{$\Gamma \vdash b : \Bool$}
                \AxiomC{$\Gamma \vdash p : \QProg$\; $\Gamma \vdash q : \QProg$}
                \RightLabel{\myrule{Q-if}}
                \TrinaryInfC{$\Gamma \vdash \Ifthene{b}{p}{q} :\QProg$}
\DisplayProof\vskip 1em
\end{center}

\smallskip
  
The usual $\lambda$-calculus typing rules are also included. Note also that,
since we do not address quantum data directly, but only via operational commands, we do not need linear types to ensure unitarity and our type
system admits weakening and contraction.
\small
  \begin{center}      
    \AxiomC{}
    \RightLabel{\myrule{Ax}}
    \UnaryInfC{$\Gamma, x : T \vdash x : T$}
\DisplayProof\hskip 1em
    \AxiomC{$\Gamma \vdash \Delta$}
    \RightLabel{\myrule{Weak}}
    \UnaryInfC{$\Gamma, \Gamma' \vdash \Delta, \Delta'$}
\DisplayProof\hskip 1em
    \AxiomC{$\Gamma, x : T \vdash t : U$}
    \RightLabel{\textcolor{re}{$\lambda $}\color{black}}
    \UnaryInfC{$\Gamma \vdash \lambda x. t : T \to U$}
\DisplayProof\hskip 1em
    \AxiomC{$\Gamma \vdash f : T \to U$}
    \AxiomC{$\Gamma \vdash t : T$}
    \RightLabel{\myrule{@}}
    \BinaryInfC{$\Gamma \vdash f t : U$}
\DisplayProof%
\end{center}
\normalsize
\subsection{Standard Semantics}
$\hqbricks$ has well-established semantics in terms of density
operators~\cite{qhlprover, coqq, Selinger2004} (without classical data) or
cq-states~\cite{qhlcv} (with classical data).  In the cq-state semantics, a
hybrid memory is described as a probability distribution over  pairs $(\sigma,
\rho)$ of a classical state $\sigma$ and a quantum state $\rho$ (described as a
density operator).\footnote{Technically, it is a map $\sigma \mapsto \mathbb
    P(\sigma, \rho) \rho$ into \emph{partial} density operators with norm $\leq
1$} A quantum input $\ket{x}{\reg}$ is represented as the cq-state $(-, \ketbra{x}{x})$,
where $-$ is the empty classical state. A denotational semantics can then be
defined as such:
\begin{align*}
    \llbracket U \rrbracket \langle \sigma, \rho \rangle &= \langle \sigma, U
    \rho U^\dagger \rangle \\
    \llbracket \textbf{Measure}(\q, \cre) \rrbracket \langle \sigma, \rho
    \rangle &= \langle \sigma[\cre := 0], |0\rangle_\q \langle0|_\q \rho
    |0\rangle_\q \langle0|_\q \rangle + \langle \sigma[\cre := 1],
    |1\rangle_\q \langle1|_\q \rho |1\rangle_\q \langle1|_\q \rangle \\
    \llbracket \cre_1 := f(\cre_2) \rrbracket \langle \sigma, \rho \rangle &=
    \langle \sigma[\cre_1 := f(\sigma(\cre_2))], \rho \rangle \\
    \llbracket \Ifthene{\cb}{\prog}{\prog'} \rrbracket \langle \sigma, \rho
    \rangle &= \left\{\begin{array}{ll}
        \llbracket \prog \rrbracket \langle \sigma, \rho \rangle & \text{if }
        \sigma \models \cb \\
        \llbracket \prog' \rrbracket \langle \sigma, \rho \rangle & \text{otherwise}
    \end{array}\right.
    \end{align*}

This semantics is lacking compactness and symbolism (\cref{intro}), so we
instead provide alternative semantics based on path-sums that surmounts these
limitations, and we show that it is sound with respect to these usual
cq-state/density operator semantics.

 \section{Hybrid Path Sums (HPS)}%
\label{hps}
\subsection{Formal Definition}
Hybrid path-sums extend unbalanced path-sums to the hybrid case. They manipulate
hybrid quantum-classical memories, that we first introduce:

\begin{definition}[Hybrid memory]
    A \emph{quantum memory} is a finite map from a set of \emph{quantum addresses}   to  symbolic Boolean
    expressions, while a \emph{classical memory} is a list of finite maps from a set of \emph{classical addresses} to  symbolic Boolean
    expressions, storing the history of classical value computation along the execution. We often refer to  a hybrid memory (a
    pair of quantum and classical memories) simply as a \emph{memory}. 
    \end{definition}

    Only the head of the list in a classical memory is accessible by programs while the
    remainder merely contains past values necessary to track the structure
    of world splittings. 
    One can visualize a classical memory as a two-dimensional table featuring
    addresses (space) on one axis and time on the other. Due to initializations
    and Discard operations (\cref{sec:discarding}), some cells of the table are
    `holes' denoted $-$.  If a pair of classical memories $C_1$ and
    $C_2$, seen as partial maps from pairs $(\text{address}, \text{time})$ to
    boolean expressions, have disjoint domains, they can be combined in parallel
    as $C_1 \otimes C_2$. 

\vspace{1em}\noindent
   \begin{minipage}{.51\textwidth}
   \begin{runex} Aside is the evolution of classical memory  in    
    teleportation  (\cref{QTel}).
   This memory stores the result of measurements at steps \steps{5} (time $t_1$)
   and \steps{6} (time $t_2$) in addresses $\bm \psi$ and $\bm a$ respectively.
   We note that pasts (non-head entries) of classical memories were omitted from
  \cref{QTel} since no overwriting of classical addresses occurs.
\end{runex}
   \end{minipage}
\begin{minipage}{.45\textwidth}
\begin{table}[H]
    \caption{Evolution of the classical memory.}
     $\begin{array}{ccc|c}
    \hline
    \textbf{Addresses} & \bm{\psi} & \bm{a}&\textbf{Time step}\\
    \hline
    \hline
    &y_1& y_0 \oplus x&t_2: 6,7\\
    &y_1& \_&t_1:5\\
    &\_& \_&t_0: 1,2,3,4
\end{array}$
\vspace{0.5em}
\end{table}
\end{minipage}

\begin{definition}[HPS]%
\label{def:hps}
    A hybrid path-sum (\hps) $\I \mapsto \pasum{\P}{\O}{\N}_{\support}$ is given
    by:
    \begin{itemize}
        \item An \emph{input signature}: a hybrid memory containing boolean
            \emph{input variables} $\{x_i\}$ or constants.
        \item A \emph{path support} $\support$: a set of boolean \emph{path
            variables} $\{y_j\}$.
        \item An \emph{output signature} $\O$: a hybrid memory
            containing boolean polynomials in $x_i$ and $y_j$.
        \item A \emph{phase polynomial} $\P \in \mathbb D[x_i, y_j]$ in $x_i$
            and $y_j$ over the ring of dyadics $\mathbb D=\{\frac {n}{2^m} : n,m
            \in \mathbb Z\}$.
        \item A \emph{scalar function} $\N$ over $x_i$ and 
            $y_j$ given by the following \emph{constructible
            numbers}~\cite{descartes1987discours} grammar:
        $\begin{array}{rclc}
         \textsc{Scalar},\N&:=& \frac{\I_1}{\I_2} \mid \sqrt{\N}\ \mid \cos\ \An
         \mid
        \sin\ \An \mid \N_1 +_r \N_2 \mid - r \mid \N_1 *_r \N_2
        \mid \frac{1}{\N} \mid x_i \mid y_j &
        \\
         \textsc{Angle},\An&:=&  \frac{2\pi\I_1}{2^{\I_2}} \mid \arccos\ \N\mid \arcsin\ \N \mid \I *_a \An\mid \B *_a \An \mid \An_1 +_a \An_2 \mid - \An&
        \\
        \end{array}
        $
    \end{itemize}
\end{definition}

\begin{notation}
A \emph{Register} $\texttt{r}$ is an array of indexed memory addresses
$\texttt{r}[i]$. $|x\rangle_{\qu,i}$ (resp.\ ${[x]}_{\texttt{c},i}$) is a quantum
(resp.\ classical) data cell with \emph{address} $(\qu[i])$ (resp.\
$(\texttt{c}[i])$), \emph{index} $i$, and \emph{containing} $x_0$. To address a vector of
several variables, we use a capital font: $X$ designates a vector of input
variables while $Y$ designates a vector of path variables. $\get \reg_i$ is
the expression contained at $\reg[i]$.. In the case of classical register, only
the topmost value is used.
The quantum and classical parts of the output $\O$ are
denoted by $\O_{\Qu}$ and $\O_{\Cl}$, respectively.  The height of $\O_{\Cl}$ is
its \emph{age} and is also the \emph{age} of the \hps{}.  A fully
instantiated $\O_{\Cl}$ is called a \emph{history} and denoted $\eta$. When
$\reg$ is of size 1, we write $\reg = \reg[0]$. Finally, $i$ and $\support$ are
almost always implicit.
\end{notation}

\subsection{Probability Derivation and Concretizations}%
\label{probadefs}
In a hybrid path sum $\h = \pasum{\P}{\F}{\N}_{\support}$,  each possible (history) $\eta$ for $\F_{\text{cl}}$ 
determines a world where $\eta$ describes the past evolution of the classical memory. Each such
world is associated to a state 
of the quantum registers. We can therefore define the vector-map of an
\hps{} as the following. 
    
\begin{definition}[Concretization I:\ Vector map of an \hps] Let $\h =
    \pasum{\P}{\F}{\N}_{\support}$, then:
    \[
    \xi : \eta \mapsto \sum_{
 \vec{y}\in2^{\support},\F_{\text{cl}}(\vec{y}) = \eta        
    } \N(\vec{y}) \cdot e^{2 \pi i \P(\vec{y})} \ket{\F_{\text{qu}}}{}
    \]
\end{definition}

\begin{runex}
At step 6 of \cref{QTel}, we have $PS_6 =\pasum{\frac{
    y_1x}{2}}{\ket{y_0}{b}{[y_1]}_{\bm \psi_c}{[y_0\oplus x]}_{\bm
a\finv{_c}}}{\frac{1}{2}}_{\{\y_0,y_1\}}$ 
\[  
    \xi(PS_6): \eta \in {\{0,1\}}^{\bm{\psi\finv{_c},a_c}} \mapsto \sum_{\begin{smallmatrix}
  y_0 = \eta(\bm a_c) \oplus 1\\
  y_1 = \eta(\bm \psi_c)
  \end{smallmatrix}}\frac{1}{2}e^{\pi i (y_1 x)}\ket{y_0}{b} = 
\frac{{(-1)}^{x\eta(\bm \psi_c)}}{2}\ket{\eta(\bm a_c) \oplus 1}{b}   \]  

\end{runex}

This map anchors our \hps{} representation with concrete vectors of the
Hilbert space and is essential in discussing and proving the soundness of
both the equational theory (Section~\ref{sec:rewrite}) and symbolic
semantics (Section~\ref{sec:operational_semantics}), and in defining
semantics assertions over \hqbricks{} programs
(Section~\ref{sec:reasoning-about-programs}). However, in practice, worlds
that differ only by a global phase are indistinguishable, so we also
introduce $ \tilde \xi(\h) : \eta \mapsto \overline{\xi(\h)(\eta)}$, the
equivalence class of $\xi(\h)(\eta)$ modulo global phase. 

Each history $\eta$ is associated to a  present value $\get{\eta}$ of the
classical memory (top row) which is reached with probability $
\label{proba_def_eq} \texttt{proba}(\h)(c) = \sum_{\eta \mid \get{\eta} =c}
|\xi(\h)(\eta)|^2 $. More generally, the probability of a predicate
$\varphi$ over the classical registers is $\texttt{proba}(\h)(\varphi) =
\sum_{\eta \mid \get\eta\vdash\varphi} |\xi(\h)(\eta)|^2$. The \emph{norm} of
$\h$ can then be defined as $|\h| = \sqrt{\sum_{\eta} |\xi(\h)(\eta)|^2}$.

To discuss the soundness of our symbolic execution with respect to the
cq-state~\cite{qhlcv} semantics, we also provide another concretization in terms
of cq-states.

\begin{definition}[Concretization II:\ Density operator map of an \hps]%
\label{def:density_maps} 
Let
    $\h = \pasum{\P}{\F}{\N}_{\support}$ Then:
    \[
    \DO(\h) : c \mapsto
\sum_{\begin{smallmatrix}
 \vec{y_1},\vec{y_2}\in2^\support, \F_{\text{cl}}(\vec{y_1}) = \F_{\text{cl}}(\vec{y_2}), 
  \get{\F_{\text{cl}}(\vec{y_1})} = c
    \end{smallmatrix}} \N(\vec{y_1})\N(\vec{y_2}) \cdot e^{2 \pi i
(\P(\vec{y_1})-\P(\vec{y_2}))} \ketbra{\F_{\text{qu}}(\vec{y_1})}{\F_{\text{qu}}(\vec{y_2})}
\]
\end{definition}
This definition is equivalent to $\DO(\h)(c) = \sum_{\get{\eta} = c} \ketbra
{\xi(\h)(\eta)} {\xi(\h)(\eta)}$, so $\DO(\h)$ is definable from $\xi(\h)$ by a
non-invertible loss of past classical information. This equivalence also
highlights the following equality, holding for an \hps{} \h and history $\eta$:
$\texttt{proba}(\h)(c) = \mathrm{tr}(\DO(\h)(c))$.

\subsection{Additional Features of \hps{} and Unbalancement Supplies.}%
\label{sec:hps_main_features}
In~\cite{amy2023complete}, M. Amy extends the original path-sum
framework~\cite{amy2018towards} to \emph{unbalanced} path-sums, with norms
varying along paths, and describes a complete equational theory for them, at the
expense of a potentially exponential size  of  terms. This is not our interest
here; instead, we focus on keeping \hps{} compact and practical for symbolic
execution and specification.

Hybrid path-sums eventually represent vectors in a vector space, so we would
wish to have a vector space structure over them. This is made possible by the
unbalancement feature, and grants us scalar multiplication and addition of \hps,
but also tensor products defined as follows.
\begin{align*}
    r \h_1 \psadd \h_2 &:=
    \pasum{\case{y_f}{\P_1}{\P_2}}{\case{y_f}{\F_1}{\F_2}}{\case{y_f}{\Pi(\support_2
        \setminus \support_1)r n_1}{\Pi(\support_1 \setminus
    \support_2)n_2}}_{\support_1 \cup \support_2 \cup \{y_f\}}\footnote{\text{Where
    $y_f$ is a  fresh path variable in  $\h_1$ and $\h_2$ and $r$ a
    constructible real number}.}\label{eq:def-psadd} \\
    \h_1 \pstens \h_2 &:= \pasum{\P_1+\P_2}{\F_1\cup\F_2}{\N_1*\N_2}_{\support_1\cup\support_2}
\end{align*}

\begin{theorem}
    Given $\h_1$ and $\h_2$ two \hps{} and $a$ a constructible real number, we
    have:
    \[
        \xi(a\h_1 \psadd \h_1) = a\xi(\h_1) + \xi(\h_2)
        \quad \text{and} \quad
        \xi(\h_1 \otimes \h_2) = \xi(\h_1) \otimes \xi(\h_2)
    \]
\end{theorem}

Decomposition over $\psadd$ and $\otimes$ corresponds to case analysis and
subsystem analysis respectively, and as such allow local reasoning both in
state space and in physical space. This is essential in practice, especially
under complex control flows.

\begin{runex}
    The path-sum at step 7 of \cref{QTel} can be decomposed over $\otimes$, with
    this decomposition constituting step 8 of the same table.
    \[
        \pasum{0}{\ket{ x}{b}{[y_1]}_{\bm \psi}{[y_0\oplus
x]}_{\bm a}}{\frac{1}{2}}_{\{\y_0,y_1\}} =  \langle\ket{x}{b}\rangle_{\emptyset} \otimes\pasum{0}{
    {[y_1]}_{\bm \psi}{[y_0\oplus x]}_{\bm a}
}{\frac{1}{2}}_{\{\y_0,y_1\}}
    \]
\end{runex}

The ability to perform case analysis over paths is critical for quantum control,
a notorious blind spot of most prior attempts (\cref{sec:sota}). In fact, once
an \hps{} $\h$ is split into $\h_b \psadd \h_{\lnot b}$ with $\h_b$ (resp.
$\h_{\lnot b}$) satisfying (resp.\ not satisfying) a quantum or classical boolean
condition $b$, an $\ifthene{b}{\prog_1}{\program_2}$ can be simply executed
by executing $\prog_1$ over $\h_b$ and $\program_2$ over $\h_{\lnot b}$
and recombining as $\sem{\prog_1}{}(\h_b) \psadd \sem{\program_2}{}(\h_{\lnot
b})$ (\cref{sec:operational_semantics}).

 \section{Rewriting Rules for \hps{}}%
\label{sec:rewrite}

We equip \hps{} with a rewrite system enabling reasoning about equivalence,
simplification, and specification of \hps{}. Concretely, we provide, an
equational theory $\prequiv$ for simplifying quantum interference patterns,
$\equiv$ which also identifies \hps{} up to undetectable global phase, and a
preorder $\prinduce$, compatible with $\equiv$ to capture lossy transformations
such as discarding subsystems.

\myparagraph{Path-Variable Reductions}%
\label{sec:path-variable-reductions}
A first goal of simplification is to reduce  the number of path-variables.
This is achieved by the rules (\myrule{HH}, \myrule{PB}, \myrule{Filter}, and
\myrule{CV}) below, which extend those of~\cite{amy2018towards, Vilmart2020SOP}.
In these rules, $f$ is a boolean expression over the input and path
variables, $\lift f$ is its casting as an integer expression where, in
particular, $\lift{f_1 \oplus f_2} = \lift{f_1} + \lift{f_2} - 2
\lift{f_1}\,\lift{f_2}$, $\Vec y$ is a vector of boolean variables of size $\#
\Vec y$, $\Pi \Vec y$ is the product $y_1 \cdots y_{\# {\Vec y}}$, and finally,
$\mathfrak S({\{0,1\}}^{n})$ is the set of bijections over $n$-bit strings,
expressed as boolean formulae over the bits.  The Phase-Bisector \myrule{PB}
rule computes the superposition of paths varying only by a phase, using Euler's
formula $1 + e^{2\pi i\bP_2} = 2 (\cos \pi \bP_2) e^{2 \pi i ({\bP_2} / 2)}$,
\myrule{Filter} eliminates null-amplitude paths, and \myrule{CV} allows changing
path-variables, i.e.\ the indexing of the paths.  These three rules generalize
[HH], [Elim] and [$\omega$] from~\cite{amy2018towards}.  Although it is
derivable from the other rules, we also  explicitly introduce the \myrule{HH}
rule from~\cite{amy2018towards} below, since it is used and referred to in
\cref{sec:cases}.

\smallskip
\begin{prooftree}
    \AxiomC{$\Vec{y_0}\cap\bsupport= \emptyset$}
    \AxiomC{$\Vec{y_1}\cap\pvar(f,\bF_{\Cl})=\emptyset$}
    \RightLabel{(\myrule{HH})}
    \BinaryInfC{$\pasum{\bP+\frac{\Vec{y_0}\cdot (\Vec{y_1} +
        \lift{f})}{2} }{\bF}{\bN}_{\bsupport
        \cup\Vec{y_0}\cup\Vec{y_1}}\prequiv
        \pasum{\bP}{\bF}{{2^{\#\Vec{y_0}}}\bN}_{\bsupport}[\Vec{y_1}\leftarrow
        f]$}
    \DisplayProof\vskip 1em
    \AxiomC{$y \notin \bsupport$}
    \AxiomC{$y \notin \bP_1, \bP_2, \bN, \bF$}
    \RightLabel{(\myrule{PhaseBisector -PB})}
    \BinaryInfC{$\pasum{\bP_1+y \bP_2}{\bF}{\bN}_{\bsupport\cup\{y\}}\prequiv
    \pasum{\bP+\frac {\bP_2} {2}}{\bF}{(2\cos \pi \bP_2)\bN}_{\bsupport}$}
 
 \DisplayProof\vskip 1em
\AxiomC{}
\RightLabel{(\myrule{Filter})}
\UnaryInfC{$\pasum{\bP}{\bF}{(\Pi \Vec{y})\bN}_{\bsupport\cup \Vec{y}}\prequiv \pasum{\bP}{\bF}{\bN}_{\bsupport}[\Vec{y} \leftarrow 1]$}
\DisplayProof\hskip 1em
\AxiomC{$\sigma \in \mathfrak S({\{0,1\}}^{\#(\Vec y)})$ bijective}
\RightLabel{(\myrule{CV})}
 \UnaryInfC{$\bh\prequiv \bh[\Vec{y}\leftarrow \sigma(\Vec{y})]$}
\end{prooftree}

\myparagraph{World-Relative Phase Equivalence}%
\label{sec:world-relative-phase}
The generalized global phase elimination rule introduces a modulo-phase $\equiv$
between \hps{} equivalence:  a phase factor that is uniform across each given
world of a \hps{} is physically undetectable. We write $\bP_2 \in \mathbb D[c
\mid c\in \get{\Cl}]$ to indicate that $\bP_2$ is a polynomial over the dyadics
in the boolean expressions appearing in the classical memory.
\begin{prooftree}
    \AxiomC{$\bh\prequiv \pasum{\bP_1 + \bP_2}{\bF}{\bN}_{\bsupport}$}
    \AxiomC{ $\bP_2 \in \mathbb D[c \mid c\in \get{\Cl}]$}
    \RightLabel{(\myrule{Phase-Elimination - PE})}
    \BinaryInfC{$\bh\equiv \pasum{\bP_1 }{\bF}{\bN}_{\bsupport}$}
\end{prooftree}

\myparagraph{Local Reasoning}%
\label{sec:local-rewriting}
We may also perform rewriting locally after a case analysis or a subspace
analysis. For this, we introduce the case analysis rule \myrule{Split} and the
subspace factorization rule \myrule{FD}, as well as compatibility rules for
these decompositions (\myrule{Plus} and \myrule{Times}). To these is added a
whole set of algebraic properties of $\psadd$ and $\otimes$ replicating the
structure of the underlying vector space (e.g.\ commutativity, distributivity,
etc.)

\begin{prooftree}
    \AxiomC{$\pvar(\bP_1,\bF_1,\bN_1, \bsupport_1) \cap\pvar(\bP_2,\bF_2,\bN_2,
    \bsupport_2) = \emptyset$}
    \AxiomC{$\oad_1\cap \oad_2 = \emptyset$}
    \RightLabel{(\myrule{Factor./Distr.-FD})}
    \BinaryInfC{$\pasum{\bP_1}{\bF_1}{\bN_1}_{\bsupport_1} \pstens
        \pasum{\bP_2}{\bF_2}{\bN_2}_{\bsupport_2} \prequiv
        \pasum{\bP_1+\bP_2}{\bF_1\cup\bF_2}{\bN_1*\bN_2}_{\bsupport_1\cup\bsupport_2}$}
    \DisplayProof\vskip 1em
    \AxiomC{$y \in \pvar(\h)$}
    \RightLabel{(\myrule{Split})}
    \UnaryInfC{$\h \prequiv \h[y:=0] \psadd \h[y:=1]$}
    \DisplayProof\hskip 1em
    \AxiomC{$\bh_1 \prequiv \bh_2$}
    \RightLabel{(\myrule{Plus})}
    \UnaryInfC{$\bh \psadd \bh_1 \prequiv \bh \psadd \bh_2$}
    \DisplayProof\hskip 1em
    \AxiomC{$\bh_1 \prequiv \bh_2$}
    \RightLabel{(\myrule{Times})}
    \UnaryInfC{$\bh \otimes \bh_1 \prequiv \bh \otimes \bh_2$}
\end{prooftree}
\smallskip

\myparagraph{Discarding and \hps{} Refinements}%
\label{sec:discarding}
When the state of a composite hybrid system $AB$ is separable (not entangled), the
\hps{} representing it can be factored as $\h_{AB} = \h_A \otimes \h_B$ using
(\myrule{FD}). Under such conditions, it is meaningful to \emph{discard} a
subsystem (\myrule{Disc.}) and obtain a partial description of the state
instead. We obtain a preorder relation which we call \emph{\hps{} refinement}
and denote $\prinduce$. 

\begin{prooftree}
\AxiomC{$\norm{\bh_1}= 1$}
\AxiomC{$\bh_1 \cap X = \emptyset$}
\RightLabel{(\myrule{Disc.})}
 \BinaryInfC{$\bh_1\otimes\bh_2 \prinduce \bh_2$}
   \DisplayProof\vskip 1em
\AxiomC{$\bh_1 \prinduce \bh_2$}
\AxiomC{$\bh_2 \prinduce \bh_3$}
\RightLabel{(\myrule{ \prinduce{}-Trans.})}
\BinaryInfC{$\bh_1 \prinduce \bh_3$}
  \DisplayProof\hskip 1em
\AxiomC{$\bh \prequiv \bh'$}
 \RightLabel{(\myrule{\prequiv\prinduce-Comp.})}
\UnaryInfC{  $\bh \prinduce \bh'$}
  \DisplayProof\vskip 1em
\AxiomC{$\bh_1 \prinduce\bh_2$}
 \RightLabel{(\myrule{\prinduce\otimes-Mon.})}
\UnaryInfC{  $\bh_1 \otimes \bh_3 \prinduce \bh_2 \otimes \bh_3$}
  \DisplayProof\hskip 1em
  \AxiomC{$\bh_1 \prinduce\bh_2$}
 \RightLabel{$\left(\myrule{\prinduce \sem{\cdot}{} - Mon.}\right)$}
\UnaryInfC{  $\sem{P}{}(\bh_1) \prinduce
\sem{P}{}(\bh_2)$}
\end{prooftree}

This subsystem discard rule has important implications on formal verification as
it provides a satisfaction relation $\prinducer$. 

\begin{theorem}[Soundness of \hps{} reductions]%
\label{theorem:sound-rewrite-gen}%
Let $\h$ and $\h'$ be two \hps{}, then,
\begin{itemize}
    \item If $\bh \prequiv \bh'$ then $\xi(\bh) =
        \xi(\bh')$
    \item If $\bh \equiv \bh'$ then $\tilde{\xi}(\bh)
        = \tilde{\xi}(\bh')$
    \item If $\bh \prinduce \bh'$ then there is $\bh''$ such that $\bh \prequiv
        \bh'\otimes \bh''$ with $|\bh''| = 1$ and $\bh$ containing no input
        variables $x$.
    \end{itemize}
\end{theorem}

  \section{Reasoning about Hybrid Quantum Programs}%
\label{sec:reasoning-about-programs}

\subsection{Forward Symbolic Semantics}%
\label{sec:operational_semantics}

A valid $\lambda$-term for obtaining a parametrized hybrid circuit is a term of
type Prog. The typing judgement $\Gamma, \Delta \vdash t : \text{Prog}$ may have
validity constraints in $\Delta$. After instantiation of the parameters and
normalization in the $\lambda$-calculus, this term has the form of one of the
constructors of the type \Prog{}. We can therefore inductively define the
interpretation of $t$ as a transformation of \hps{}; i.e.\ a function mapping an
\hps{} to another. The transformations are given in~\cref{fig:denot-sem}
with the example of the $Z^* = \{H ,
X\}\cup {\{Z_k\}}_{k\in\mathbb{N}}$  set of gates from Section~\ref{sec:hqbricks-syntax}. For sake of
readability, we suppose that $\Qu$ is a single qubit register with $\qu = \Qu[0]$. The
general case is easily obtained by iterated parallelism. 

In figure~\ref{fig:denot-sem}, we use notations briefly summarized as: $\bh +_\bP
\P$ adds $\P$ to the phase of $\bh$ (for instance, $\pasum{\P_1}{\bF}{\bN} +_{\P} \P_2 =
\pasum{\P_1 + \P_2}{\bF}{\bN}$); $\bh \cup_{\bsupport} \{y\}$ adds a new
path-variable $y$ to the support of $\bh$; $\bh[\qu \leftarrow \bB]$ replaces
the value of $\qu$ in $\bh$ by $\bB$; $\bh[\Cl \leftarrow \bB]$ appends $\bB$ to
the history stack in the register $\Cl$; and $\get{\qu}$ returns the value of $\qu$ in
$\bh$.

\small
\begin{table}[htpb]
\caption{Program symbolic semantics in terms of transformations of \hps{}}%
\label{fig:denot-sem}
\begin{subfigure}{0.5\textwidth}
\begin{align*}
    \llbracket \textbf{Skip} \rrbracket(\bh) &= \bh \\
    \llbracket \init{\Qu} \rrbracket(\bh) &= \bh[\Qu\leftarrow 0] \\
    \llbracket \textbf{Measure}(\Qu,\Cl) \rrbracket(\bh) &= \bh[\Cl \leftarrow \lceil \Qu \rceil] \\
    \llbracket \Cl_1 = f(\lceil \Cl_2 \rceil) \rrbracket(\bh) &= \bh[\Cl_1 \leftarrow f(\lceil \Cl_2 \rceil)] \\
    \llbracket \Ifthene \bkappa p q \rrbracket(\bh) &= \bkappa{\semantics{p}\bh}
    + (1- \bkappa){\semantics{q}\bh}\\
    \llbracket p ; q \rrbracket &= \llbracket q \rrbracket \circ \llbracket p \rrbracket\\
    \llbracket \for{x} j k p \rrbracket &=\\\multicolumn{3}{r}{$
    \semantics{p[x:=j]}{} \circ \semantics{p[x:=j+1]}{} \circ \cdots \circ \semantics{p[x:=k]}{}$} 
\end{align*}
\caption{Language constructs}
\end{subfigure}
\hfill
\begin{subfigure}{0.38\textwidth}
\begin{align*}
          \llbracket H(\qu) \rrbracket(\bh) &= \\ 
     \multicolumn{3}{r}{\qquad     $\left(\frac{1}{\sqrt{2}} \cdot \bh 
     \cup_{\bsupport} \{y\} +_\bP{} \frac{y \lceil\qu\rceil} 2 \right)[\qu\leftarrow y]$}  \\
    \llbracket X(\qu) \rrbracket(\bh) &= \bh[\qu \leftarrow 1 \oplus \lceil \qu\rceil] \\
    \llbracket Z_k(\qu) \rrbracket (\bh)&= \bh +_\bP \frac {\lceil \qu \rceil} {2^k}
\end{align*}
\caption{Primitive gates: the $Z^*$ example}

\end{subfigure}

    \Description[]{}
\end{table}

\normalsize

\myparagraph{Soundness of the \hps{} Semantics} 
We establish the soundness of our symbolic semantics  with respect to the
 cq-state formalism and the \textsc{qhl-cv} language~\cite{qhlcv}, \fv{as it is the closest to a standard among existing propositions for the formalization of programming languages over hybrid data.} Technically, modulo some rewriting, the object languages only
differ in that the loops of \hqbricks{} are bounded
\textbf{for} loops, while those of \textsc{qhl-cv} are potentially
unbounded \textbf{while} loops.\footnote{At the constructor level, an
    additional apparent differences comes from the fact $S$ language also
    provides oracle probabilistic classical assignment $c :=_\$ p$. This is
    covered by \hqbricks{} through the possibility of introducing non-unitary
    elementary operations
    through \hpps{} such as $ |\vec 0\rangle \mapsto \sqrt{p_i}|i\rangle$. See \cref{sec:cases-corrections} for an application to error modelling.}
In the following theorem, we write
$\semantics{\cdot}_{cq}$ for the cq-state semantics
and we identify density operator maps from
\cref{def:density_maps} to cq-states (provided they are obtained from an \hps{}
of norm $\leq 1$).

\begin{theorem}[Soundness of the semantics]%
\label{theo:language-soundness}
\hqbricks{} is sound with respect to \textsc{qhl-cv}: there is a
    (non-surjective) mapping $\phi$ from \hqbricks{} to language $S$ such that for
    any closed \hqbricks{} program $\Pr$ (i.e.\ a program where all parameters
    have been instantiated) and an \hps{} $\h$ with norm $\leq 1$,
    \[\DO(\semantics{\Pr}(\h)) = \semantics{\phi(\Pr)}_{cq}(\DO(\h))\]
    
\end{theorem}
\begin{proof}[Proof sketch:]
    After defining a mostly cosmetic translation of the programs of \hqbricks{}
    in the language of \textsc{qhl-cv}, the proof proceeds by induction on the
    constructors of \prog{}. For each of the constructors, we show that the
    symbolic semantics $\semantics{\prog}(\h)$ matches the denotational
    semantics of \textsc{qhl-cv}. The complete induction is developed in
    Appendix~\ref{full-proof-soundness-language}.
\end{proof}
\subsection{A Minimal \Hps{} Based Assertion Language}%
\label{sec:static-program-analysis}

\hps{} provide a sound symbolic representation for hybrid programs semantics, paving the way for the static analysis of these programs. In the present section we introduce a core assertion language on top of \hpps{} derivation. As illustrated in Section~\ref{sec:cases}, assertions introduced below are sufficient for reaching the standard literature specifications for all application cases we thought about.

\myparagraph{Full Hybrid State  Description Assertions}
The most straightforward usage of \hps{} is to derive the result $\bh'$ of
executing a hybrid program $\P$ on a symbolic hybrid state \bh. The induced
assertion is written $ \equivr \bh'$  and it is tested by the following condition, with soundness established through Theorem~\ref{theorem:sound-rewrite-gen}.
\[\hoare{\equivr\bh}{\Prog}{\equivr\bh'} \textit{ iff } \semantics{\Prog}(\bh)
\equiv \bh'\]

\myparagraph{Satisfaction Assertions}
Our refinement relation $\prinduce$, which formalizes the Discard operation,
directly induces a satisfaction relation $\prinducer$ as its logical
counterpart, referring to \hps{} as predicates over hybrid classical/quantum
states.  This allows to specify properties of a state without the need to
describe the full state. For example, such properties include specifying values
of classical registers.  The satisfaction of this assertion is given by:
\[\hoare{\equivr\bh}{\Prog}{\prinducer\bh'} \textit{ iff there is $\bh''$ such
    that $\norm{\bh''} = 1$ and } \semantics{\Prog}(\bh) \equiv \bh'\otimes
\bh''\] with soundness established, again, through
Theorem~\ref{theorem:sound-rewrite-gen}. Note that the specification requires an
existential witness. Concretely, the test is built-in and does not require
writing \hps{} variables.

\myparagraph{Probabilistic Satisfaction Assertions}
There is an even more general way to specify properties of a state, by
quantifying the probability that a certain quantum or classical property is
true. More formally, given an \hps{} $\h$, we can quantify the probability of being
in a world of that \hps{} in which a certain property \texttt{cond} holds.
The induced assertion is written $\dsat{\texttt{cond}}{\sim}{r}$  and it is
defined as:
\[\hoare{\equivr\bh}{\Prog}{\dsat{\texttt{cond}}{\sim}{r}} \textit{ iff }
    \sum_{\begin{smallmatrix}\eta\in \semantics{\Prog}(\bh)\\ \eta \models
    \texttt{cond}\end{smallmatrix}} \norm{\xi(\semantics{\Prog}\h)(\eta)} \sim r
\]
where $\eta, \h \models \texttt{cond}$ is evaluated as a standard boolean
formula, quantum atoms 
are evaluated under state
$(\bh) (\eta)\prinducer  \texttt{At}$, and with $\sim$ being
one of $\leq, \geq, =, < , >$.

\fv{\myparagraph{Discussion}
As illustrated in \cref{sec:cases}, this core assertion language is sufficient
for expressing textbook specifications for a broad range of hybrid computing
features. Nevertheless, it does not allow expressing  some advanced features such
as  program approximations or relational assertions (approximate correctness,
quantitative bounds)~\cite{10.1145/3290346,barthe2025duality,Yu:2025tod}. We do not see any fundamental limitations that would impede extending our
approach  to such cases; nonetheless, we leave this for future work, as we first
focus on fundamental basic assertions.}
    \section{Illustrative Cases}%
\label{sec:cases}
\def\checkmark{\tikz\fill[scale=0.3](0,.35) -- (.25,0) -- (1,.7) -- (.25,.15) -- cycle;} 
 \begin{wraptable}{r}{7.5cm}
\tiny
\centering
\caption{\fv{Features demonstrated by our  case studies}}%
\label{tab:implementation}
\begin{tabular}{|c|c|c|c|c|c|}
\cline{2-6}
\multicolumn{1}{c|}{}& Telep. & \begin{tabular}{c}
    QPE \\(exact)
\end{tabular}    
    & \begin{tabular}{c}
    QPE \\(appr)
\end{tabular}    
 & BRUS & QEC \\
\cline{2-6}
\hline
 \multicolumn{6}{|c|}{\textbf{Equational Theory \fv{(\textbf{ET})}}}\\
\hline
$\myrule{Change Variable(CV)}$&\checkmark&&&& \\
\hline
$\myrule{HH}$&&\checkmark&&&\checkmark \\
\hline
$\myrule{Phase Bisector (PB)}$&&&\checkmark&& \\
\hline
\hline
\multicolumn{6}{|c|}{\textbf{Partial state verification \fv{--- Hybrid data (HD}) }}\\
\hline
$\myrule{Factorize - Distr (FD)}$&\checkmark&\checkmark&\checkmark&\checkmark& \\
\hline
$\myrule{Discard}$&\checkmark&\checkmark&\checkmark&\checkmark& \\
\hline
\hline
  \multicolumn{6}{|c|}{\textbf{Probabilistic specification checking \fv{--- Hybrid data (\textbf{HD})}}}  \\
\hline
$\myrule{Filter}$&&\checkmark&\checkmark&\checkmark&\checkmark \\
\hline
\end{tabular}

\end{wraptable}
We show now the advantages of \hps{} and their different features on a representative set of hybrid quantum programs: the Quantum Phase Estimation algorithm --- both in the exact and in the approximate case 
(\cref{sec:cases-qpe}), a bounded version of the Repeat-Until-Success pattern
(\cref{sec:cases-brus}) and a basic quantum error correction scheme
(\cref{sec:cases-corrections}).  Together with the teleportation running example, these cases illustrate the main rewriting, refinement, and assertion features presented so far.  
\fv{This is summarized  in \cref{tab:implementation},  refining  classification  from \cref{SRs} over the 
 Hybrid data (\textbf{HD}) and 
Equational Theory (\textbf{ET}) criteria.}

Let us first introduce some
\hqbricks{} syntactic sugar:
\[
    \centering
\begin{array}{rclcrcl}
    \q\triangleright \U&:=&\Apply\ \U(\q)&    \hspace{1cm}&\kappa\Rightarrow \prog &:=& \Ifthene \kappa p {\textbf{Skip}}    \\      
\q \multimap \cre&:=&\textbf{Measure}(\q,\cre)&   &\textbf{Let } f(\Vec x)  = t \textbf{ in } u   &:=& (\lambda f. u) (\lambda
   \Vec x. t) \\
   \multicolumn{7}{c}{\q\rightrightarrows p(k):=\for{k}{0}{\card{\q}-1}{\get{Q[k]}\Rightarrow p(k)}}\\
\\
\end{array}
 \]

\subsection{Quantum Phase Estimation (QPE)}%
\label{sec:cases-qpe}
Developed by Cleve et al.~\cite{cleve1998quantum} and
Kitaev~\cite{kitaev1995quantum}, Quantum Phase Estimation (QPE) is a central
piece in many quantum algorithms, such as quantum
simulation~\cite{georgescu2014quantum} and the HHL
algorithm~\cite{harrow2009quantum} which solves linear systems of equations in
PolyLog time. QPE takes as input a unitary operator $U$ and an Eigenvector
$\ket{v}{\texttt{E}}$ of $U$ represented as a \hpps{} $\bh =
\pasum{\bP}{\bF_{\texttt{E}}}{\bN}$ of norm $1$, and finds the eigenvalue $e^{2
i\pi\theta}$ associated with $\ket{v}{\texttt{E}}$. In other words: assuming the
existence of a $\theta$ satisfying $\hoare{\bh}{O}{\bh+_{\bP}\theta}$, the goal
is to find it. More precisely, we search for the value 
$\tilde \theta \in [0;2^n)$ 
that minimizes the distance $\tilde \theta - \frac{\theta}{2^n}$ (modulo 1). The
target specification is that the QPE output is $\tilde \theta$ with
probability at least $\frac{4}{\pi^2}$. Detailed introduction to the case is
given in~\cite{nielsen2002quantum} Chapter 5, among others.

Figure~\ref{fig:qpe-exec} presents the symbolic circuit for QPE and its
\hqbricks{} code together with its specification derivation. In the
specifications, the function $\inttype(Y)$ refers to the integer interpretation
of bit vectors: $\inttype(Y):= \sum_{j=0}^{2n-1}Y(j)* 2^j$.  QPE goes through
three main steps:
\begin{enumerate}
    \item The state $\ket{0}{\texttt{C}}\ket{v}{\texttt{E}}$ is raised to the
        superposition $\left(\sum_{0\leq k < 2^{n}}\ket{k}{\texttt{c}}\right)
        \otimes \ket{v}{\texttt{E}}$ by a parallel application of the Hadamard
        gate in register \texttt{C}. This is encoded by the \hps{} $
        \pasum{0}{\ket{\Vec{Y_0}}{\texttt{C}}}{\frac{1}{\sqrt{2^n}}}_{\Vec{Y_0}}\otimes
        \ket{v}{\texttt{E}}$.  \item A series of $U^{2^{i}}$ gates is applied to
        register $\texttt{E}$ each controlled by the value in register
        $\texttt{c}$ at position
        $i^{\text{th}}$, so as to produce the state 
        $\pasum{\frac{\theta\cdot \toint{Y_0}}{2^n}}{\ket{\Vec{Y_0}}{\texttt{C}}}{\frac{1}{\sqrt{2^n}}}_{\Vec{Y_0}}\otimes
\ket{v}{\texttt{E}}$,
    \item This last  is a close approximation of 
        $\ket{\text{QFT}(\tilde \theta)}{\texttt{c}}\ket{v}{\texttt{E}}$, 
        the result of applying the Quantum Fourier Transform, a standard unitary
        primitive, to  $\ket{\tilde \theta} {\texttt{c}} \ket{v}
        {\texttt{E}}$.  Hence, 
        applying the inverse of QFT,
        written QFT$^{\dagger}$, 
to $\pasum{\frac{\theta\cdot \toint{Y_0}}{2^n}}{\ket{\Vec{Y_0}}{\texttt{C}}}{\frac{1}{\sqrt{2^n}}}_{\Vec{Y_0}}\otimes
\ket{v}{\texttt{E}}$
        results in an  approximation of state
        $\ket{\tilde \theta} {\texttt{c}} \ket{v} {\texttt{E}}$ that is good
        enough to output value $\get{\texttt{C}} = \tilde\theta$ with probability
        $\geq \frac{4}{\pi^2}$ after measurement.
\end{enumerate}

\medskip

\disc{,
giving  state ${[\begin{smallmatrix}
\Vec{Y_1}\\X \end{smallmatrix}]}_\cre$ for  register \cre{} of. This \emph{memory of the past} is then
separated from the rest of the specification (\myrule{FD}) and
discarded (\myrule{Disc}), so that \cre{} holds a single present value.}

With regard to specification, the QPE case is twofold, depending on whether the
sought eigenvalue has (exact case, when $\theta = \tilde \theta$) or does not
have (approximate) an exact $n$ bits decomposition. In the exact case the
algorithm guarantee is to output the exact eigenvalue with certainty.
The probabilistic specification is obtained by rule \myrule{Filter}. In the approximate case, the resulting geometric series is computed by use of  \myrule{Phase Bisector}.

\begin{figure}[htbp]
    \centering
    \scriptsize
\[
\begin{array}{cc}
\begin{array}{ll}
    \textbf{Let } QPE(U,\texttt{c},\texttt{E},\cre)=\{ &\equivr \ket{v}{\texttt{E}}
\\
\quad \init(\texttt{c}); \myrule{FD}&\equivr
\pasum{0}{\ket{0}{\texttt{c}}}{1}{1}\otimes
\ket{v}{\texttt{E}}

\\
\quad \texttt{c} \triangleright H;&\equivr
\pasum{0}{\ket{\Vec{Y_0}}{\texttt{c}}}{\frac{1}{\sqrt{2^n}}}_{\Vec{Y_0}}\otimes
\ket{v}{\texttt{E}}
\\
\quad \texttt{c} \rightrightarrows (\texttt{E} \triangleright U^{k+1});&\equivr

\pasum{\frac{\theta\cdot
\toint{Y_0}}{2^n}}{\ket{\Vec{Y_0}}{\texttt{c}}}{\frac{1}{\sqrt{2^n}}}_{\Vec{Y_0}}\otimes
\ket{v}{\texttt{E}}

\\\quad
c \triangleright QFT^{\dagger};&\equivr
\pasum{\frac{\toint{Y_0}(\theta -
\toint{Y_1})}{2^{n}}}{\ket{\Vec{Y_1}}{\texttt{c}}}{\frac{1}{2^n}}_{\Vec{Y_0},\Vec{Y_1}}\otimes
\ket{v}{\texttt{E}}

\\\quad
c \multimap \cre\}&\equivr
\pasum{\frac{\toint{Y_0}(\theta- \toint{Y_1})}{2^{n}}}{{[
\Vec{Y_1}]}_{\texttt{C}}}{\frac{1}{2^n}}_{\Vec{Y_0},\Vec{Y_1}}\otimes
\ket{v}{\texttt{E}}
  \\
\multicolumn{1}{c}{\myrule{FD/Disc/Filter/PB} }& \dsat{\get{\texttt{C}} = \tilde \theta}{=}{\sqnorm{\frac{1}{2^{2n}}\sum_{\Vec{y_0}}e^{2\pi i\Vec{Y_0}\frac{\theta- \tilde\theta}{2^{n}}}}} \geq \frac{4}{\pi^2}
\\
\end{array}& 
\begin{array}{c}\hspace{-.5cm}\scalebox{.8}{   \input{qpeq}}
\end{array}
\end{array} 
\] 
\normalsize
    \caption{\fv{Quantum Phase Estimation, code, specification derivation and
    circuit}}%
    \label{fig:qpe-exec}
\end{figure}

\subsection{Repeat Until Success (RUS)}%
\label{sec:cases-brus}

Repeat-Until-Success (RUS) is a standard routine in quantum programming: given a
stochastic quantum process $U$ achieving a success condition $\condition$ with
probability $\delta$, it consists in repeating iterations of $U$ until reaching this
condition. Also referred to as \emph{post-selection}, the scheme essentially enables one to select a desired result from a
non-deterministic process, shifting the indeterminacy from the computation
output to its execution time. It is used at various levels in the quantum
development process, be it as an elementary brick in involved microcode
implementation protocols (magic state distillation~\cite{knill2004fault}) or as
the highest-level layer for most quantum algorithms with non-deterministic
outcome (Grover, Shor, Phase estimation, etc.).

We develop a parameter-bounded version of the scheme, Bounded Repeat Until
Success (BRUS), where the number of loop calls is bounded by a user-fixed
parameter.

\begin{table}[htbp]
\footnotesize
    \caption{\fv{Bounded Repeat Until Success: code and specifications }}%
    \label{fig:brus-exec}
\[
\begin{array}{c|c}
\textit{Code + spec.}&\textit{Invariant definition}\\\hline
\begin{array}{lc}
    \textbf{Let } \texttt{BRUS}(\{U,\texttt{Q}\},\{\texttt{C}\},\texttt{cont},n)=\{&\Gamma_{0}\\\quad
\forprog{i}{1}{n}&\Gamma_{i}
\\\qquad
 (\neg \texttt{cont}(\get{\texttt{C}_c})) \Rightarrow
\ U(\{\texttt{Q}\},\{\texttt{C}\})& \Gamma_{i+1}
\\
\quad \textbf{done}\}& \Gamma_{n}
\\
\end{array}
&\begin{array}{rl}\Gamma{i}(\delta, \texttt{cont},\cre_c):=& \prinducer \bh, \\
&\dsat{ \texttt{cont}(\get{\cre_c})}{\geq}{1-{(1-\delta)}^i},\\
&\dsat{ \neg\texttt{cont}(\get{\cre_c})}{\leq}{{(1-\delta)}^i},\\
     
\end{array}\end{array}\]
\end{table}

The code and specifications are given abstractly in~\cref{fig:brus-exec}.
They rely  on a non-instantiated invariant condition $\Gamma_i$ (formalized in
~\cref{fig:brus-exec}, right side), such that:

\begin{itemize}
    \item For all integers $i$, $\Gamma_i$ entails
        $\dsat{\condition(\cre_c)}{\geq}{1-{(1-\delta)}^i}$;
    \item Any conditional application  of $U(\{\texttt{Q}\},\{\texttt{C}\})$
        from a  $\Gamma_i$ state ends in a 
        $\Gamma_{i+1}$ state. 
\end{itemize}
Then, any iterative execution of $n$ application leads from a $\Gamma_0$ state
to a $\Gamma_n$ satisfying $\dsat{\condition(\cre_c)}{\geq}{1-{(1-\delta)}^i}$.
The proof of correctness requires the invariant $\Gamma_i$ to be of the form $\bh\otimes(\bh_i\psadd\bh'_i)$, that is a tensored composition
of:
\begin{itemize}
    \item an invariant component $\bh$. In the QPE example from
        above, it corresponds to the description of the
        initial state of quantum register \texttt{E}:
        $\ket{v}{\texttt{E}}$, and
    \item a variant component, itself structured as a path-sum addition
        $\bh_i\psadd\bh'_i$ whose terms are separated by the classically read
        condition $\texttt{cont}(\cre_c)$. 
        \disc{In Figure~\ref{fig:brus-exec}, $\bh\models
        \varphi$ is used as a shortcut for: $\bh$ satisfies
        $\dsat{\varphi}{=}{\norm{\bh}}$.}

\end{itemize} 
Then we analyse $\bh_i$ and $\bh'_i$  separately. The
instruction  $(\neg \texttt{cont}(\get{\texttt{C}_c})) \Rightarrow \
U(\{\texttt{Q}\},\{\texttt{C}\})$ triggers the execution of $U$ in the $\bh_i$
component but not in $\bh'_i$ and it transforms $\bh_i$ into a path-sum addition
with a $(\norm{\bh_i}\cdot\delta)$ probable component satisfying
$\texttt{cont}(\get{\texttt{C}_c})$. Adequate information discarding then
enables reshaping the induced three-term
path-sum addition into the $\Gamma_{i+1}$ form.

We applied this generic loop scheme to both a parametrized iteration of Hadamard gates and our instances of the Phase Estimation circuit from \cref{sec:cases-qpe}. Experimental details are provided in \cref{table:xps}. We think this development as a foundational step toward the formalization of Hoare style
reasoning and probability propagation upon \hps{} analysis tools.  

 \subsection{Quantum Errors and Quantum Error Correction}%
\label{sec:cases-corrections}
\begin{figure}[h]
    \tiny
    \[
\begin{array}{cc}
\begin{tabular}{ll}
\multicolumn{1}{l}{\text{\textbf{Let}  QEC3(\texttt{Q,\cre})=\{}}&$\equivr\pasum{0}{\ket{x}{\psi}}{1}_{\emptyset}$\\\quad
$\q_0,\q_1:=0$;&$\equivr\pasum{ 0}{\ket{x}{\psi}\ket{0}{\q_0}\ket{0}{\q_1}}{1}_{\emptyset}$\\\quad
$\psi\Rightarrow(\q_0 \triangleright X)$;&$\equivr\pasum{0}{\ket{x}{\psi}\ket{x}{\q_0}\ket{0}{\q_1}}{1}_{\emptyset}$\\\quad
$\q_0\Rightarrow(\q_1 \triangleright X)$;&$\equivr\pasum{0}{\ket{x}{\psi}\ket{x}{\q_0}\ket{x}{\q_1}}{1}_{\emptyset}$\\\quad
 $\psi,\q_0,\q_1\triangleright\tilde{I}$;&$\equivr\pasum{0}{\ket{x\oplus
 w_0}{\psi}\ket{x\oplus w_1}{\q_0}\ket{x\oplus w_2}{\q_1}{[\Vec{w}]}_{\bot}}{\sqrt{f(\Vec w)}}_{\{\Vec w\}}$\\\quad
$c_0,c_1:=0$;&$\equivr\pasum{0}{\ket{x\oplus w_0}{\psi}\ket{x\oplus
w_1}{\q_0}\ket{x\oplus w_2}{\q_1}\ket{0}{c_0}\ket{0}{c_1}{[\Vec w]}_{\bot}}{\sqrt{f(\Vec w)}}_{\{\Vec w\}}$\\\quad
$c_0,c_1\triangleright H$;&$\equivr\pasum{0}{\ket{x\oplus w_0}{\psi}\ket{x\oplus
w_1}{\q_0}\ket{x\oplus w_2}{\q_1}\ket{y_0}{c_0}\ket{y_1}{c_1}{[\Vec w]}_{\bot}}{\frac{1}{2^n}\sqrt{f(\Vec w)}}_{\{\Vec w,y_0,y_1\}}$\\\quad

$c_0 \Rightarrow(\psi \triangleright Z;\q_0 \triangleright Z)$;&$\equivr\pasum{ \frac{y_0(w_0+w_1)}{2}}{\begin{array}{ll}
    \ket{x\oplus w_0}{\psi}&\ket{x\oplus w_1}{\q_0}\\\ket{x\oplus w_2}{\q_1}&\ket{y_0}{c_0}\ket{y_1}{c_1}\end{array}}{\de{1}{2}\sqrt{f(\Vec w)}}_{\{\Vec w,y_0,y_1\}}$\\\quad

$c_1\Rightarrow(\q_0 \triangleright Z;\q_1 \triangleright Z)$;&$\equivr\pasum{ \frac{
    y_0(w_0+w_1)+y_1(w_1+w_2)}{2}}{ 
    \begin{array}{ll}
    \ket{x\oplus w_0}{\psi}&\ket{x\oplus w_1}{\q_0}\\\ket{x\oplus
w_2}{\q_1}&\ket{y_0}{c_0}\ket{y_1}{c_1}\end{array}{[\Vec w]}_{\bot}}{\de{1}{2}\sqrt{f(\Vec w)}}_{\{\Vec w,y_0,y_1\}}$\\\quad

$c_0,c_1\triangleright
H$;&$\equivr\left\langle\frac{y_0(w_0+w_1)+y_1(w_1+w_2)+y_2y_0 + y_3y_2}{2},
    {\de{1}{2}\sqrt{f(\Vec w)}}\right. $\\
   &\qquad $\left.
{\begin{array}{ll}
    \ket{x\oplus w_0}{\psi}&\ket{x\oplus w_1}{\q_0}\\\ket{x\oplus
w_2}{\q_1}&\ket{y_2}{c_0}\ket{y_3}{c_1}\end{array}{[\Vec
w]}_{\bot}}\right\rangle_{\{\Vec w,y_0,y_1,y_2,y_3\}}$\\

\multicolumn{1}{c}{$\myrule{HH^2}$}
&$\pasum{ 0}{\ket{x\oplus w_0}{\psi}\ket{x\oplus w_1}{\q_0}\ket{x\oplus
w_2}{\q_1}\ket{w_0+w_1}{c_0}\ket{w_1+w_2}{c_1}{[\Vec w]}_{\bot}}{\sqrt{f(\Vec w)}}_{\emptyset}$\\\quad

$c\multimap$\cre\}&$\equivr\pasum{0}{\ket{x\oplus w_0}{\psi}\ket{x\oplus
w_1}{\q_0}\ket{x\oplus w_2}{\q_1}{[w_0+w_1]}_{\cre_0}{[w_1+w_2]}_{\cre_1}{[\Vec
w]}_{\bot}}{\sqrt{f(\Vec w)}}_{\emptyset}$\\\quad

$\begin{array}{l}
\cre_0\cre_1\Rightarrow(\q_0 \triangleright X);\\
\cre_0\overline{\cre_1}\Rightarrow(\psi \triangleright X);\\
\overline{\cre_0}\cre_1\Rightarrow(\q_1 \triangleright X);
\end{array}$
&$\equivr\pasum{0}{\ket{x\oplus (\sum_{\Vec{w}} \geq
2)}{\psi}\ket{x\oplus (\sum_{\Vec{w}} \geq 2)}{\q_0}\ket{x\oplus
(\sum_{\Vec{w}} \geq
2)}{\q_1}{[w_0+w_1]}_{\cre_0}{[w_1+w_2]}_{\cre_1}{[\Vec w]}_{\bot}}{\sqrt{f(\Vec w)}}_{\emptyset}$\\

\multicolumn{1}{c}{\probadef\myrule{/Disc.}}&$\dsat{\get{\psi \q_0 \q_1} =
xxx}{=}{{(1-p)}^3+3p{(1-p)}^2}$\\
\end{tabular}
&
\hspace{-5.2cm}
\begin{tabular}{c}
\scalebox{.76}{
\begin{quantikz}[row sep={0.5cm,between origins},column sep=0.05cm,wire
types={q,q,q,n,n,n,n}]
\lstick{$\ket{x}{\psi}$}\slice[style = blue!0]{\steps{1}}&\ctrl{1}& & \gate[3,disable auto height]{\tilde{I}^{\otimes 3}
    }\slice[style = blue!0]{\steps{3}}& &&\gate[2]{Z}& && && &\targ{}&& \\
    \lstick{$\ket{0}{\q_0}$}&\targ{}&\ctrl{1}\slice[style = blue!0]{\steps{2}}& & && &\gate[2]{Z}& && &\targ{} &&&\\
    \lstick{$\ket{0}{\q_1}$}& &\targ{}&    & && && && && &\targ{}&\\
& & &\lstick{$\ket{0}{c_0}$}&\setwiretype{q}&\gate{H}&\ctrl{-3}& &\gate{H}\slice{\steps{4}}&\meter{}\\    
& & &\lstick{$\ket{0}{c_1}$}&\setwiretype{q}&\gate{H}& &\ctrl{-3}&\gate{H}& &\meter{}\slice{\steps{5}}\\    
& && && && &\lstick{$\cre_0$}&\setwiretype{c}\wire[u][2]{c}& &\ctrl{-4}&\ctrl{-5}&\octrl{-3}&\\
& && && && &\lstick{$\cre_1$}&\setwiretype{c}&\wire[u][2]{c}&\ctrl{-1}&\octrl{-1}&\ctrl{-1}&
\end{quantikz}
  }
\\\vspace{1cm}\\\vspace{1.3cm}\\\vspace{1.3cm}  \end{tabular}

  \end{array}
\]
\normalsize
\caption{\fv{Circuit and \hps{} derivation for the 3-qubit error correction code}}%
\label{fig:freriv-trhree-qbits}
\end{figure}
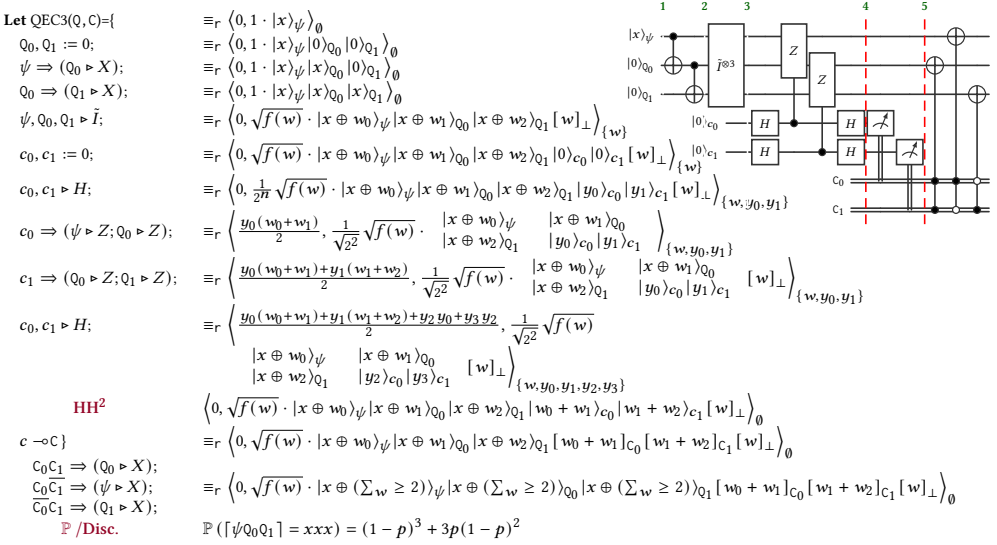
This example illustrates the possibility to encode non-deterministic processes
(Kraus operators) and error models and correction in \hps{}. Errors can be
modeled as mixed gates $(f, {\{U_{\Vec w}\}}_{\Vec{w} \in \mathbb Z_2^n})$ which
act as $U_{\Vec w}$ with probability $f(\Vec w)$.  Formally,
\[ \semantics{(f, \{U_{\Vec w}\})}_{\Vec w \in \mathbb Z_2^n}(\h) =
\bigpsadd_{w\in \mathbb Z_2^n} \sqrt{f(\Vec w)} \semantics{U_{\Vec w}(\h)}
\otimes {[w]}_{\bot} \]
where the $\bot$ register is to be interpreted as an inaccessible `dump site'
address, allowing world splitting without access to the values.

\fv{For instance, consider the three-qubit code described in details by
Roffe~\cite{roffe2019quantum} with its circuit and the corresponding specification derivation in \hps{} given in \cref{fig:freriv-trhree-qbits}. Here,
an erroneous identity operator $\tilde I$,  behaves as the identity with
probability $p$ and otherwise applies a bit flip. In \cref{fig:freriv-trhree-qbits}, the state
$\ket{x}{\psi}$ is first broadcast to the register $\psi \q_0 \q_1$ (Step~\steps{1}).  Then, each qubit in $\psi \q_0 \q_1$ undergoes  the
non-deterministic  erroneous identity $\tilde I$ (Step~\steps{2}) modeling the
occurrence of errors. Next,  a \emph{syndrome} is measured by applying
$Z$-stabilizers controlled by ancilla (Step~\steps{3}) and these ancillas are
measured (Step~\steps{4}). The results are then processed by classical control to
deduce a correction to apply (Step~\steps{5}). 
}

\fv{Note that Step~\steps{5} performs successive control on the same measurement results.
Therefore, storing that classical data in actual memory cells is unavoidable for this
case, a requirement that is missed by most pre-existing solutions (see \cref{SRs})}.

\section{Implementation and Experiments}%
\label{sec:implementation}

We have implemented \hqbricks{} and the \hps{} framework described so far in OCaml. 
This development provides the full material introduced so far
and enables symbolic execution over hybrid quantum programs for our core assertion language.  
It consists of 14567 lines of code, including 
\hqbricks{} and a \hps{} library, tests that cover both the
implementation and analysis of the illustrative cases from \cref{sec:cases}, an
intermediate representation with translations to and from OpenQASM2, and a
Jupyter notebook tutorial providing an accessible and interactive introduction
to the tool.

\disc{\medskip 

\textit{Our implementation as well as all data and scripts for experiments will be submitted for artifact evaluation and made open source upon acceptance.} 
}

Our  experiments cover  the standard hybrid quantum use
cases from the literature from \cref{tab:implementation}.  Different implementation metrics   are reported in \cref{table:xps}:
number of (quantum + classical) \emph{wires}, number of quantum
\emph{gates}, number of alternative measurement scenarios (\emph{worlds}) \emph{computation time} in seconds \fv{and number of \textit{Rewriting rules} automatically triggered. 
In addition to the final specifications to verify (assertions from the language
in \cref{sec:static-program-analysis}), our experiments require the user to provide a generic invariant  for the BRUS scheme case. Note that this only
concerns the generic scheme: no additional input is required for applying BRUS
to a specific unitary function. Other cases are fully automated. The last column in   \cref{table:xps} provides data about automatically triggered rewriting rules. For example, the
quantum teleportation with $50\ 000$ wires automatically triggers $10\ 000$
occurrences of the ChangeVariable (\myrule{CV}) rule and one occurrence of both the
Factorize-Distr (\myrule{FD}) and the Discard (\myrule{Disc.}) rules}.
\fv{\cref{table:xps} also provides, when available, computation time data for similar tests from the literature, along with the related references.}

\begin{table}[htbp]
\caption{\fv{Experimental evaluation result for \hqbricks{} symbolic execution}}%
\scriptsize
\begin{tabular}{|l|ccccc||c|}
\hline
 Case&\#(wires)&\#(gates)&\#(worlds)&Comp. time & Rewritings&Sota Comp. time*  \\
\hline
\hline
&5&6&4&0.000005&3&
\begin{tabular}{cc}
21.6  &symQV~\cite{bauer2023symqv}\\0.02&TDD~\cite{hong2022equivalence}\end{tabular}\\
\cline{2-7}
Quantum teleportation&500&600&$2^{200}$&0.001434&102&\nodata\\
&50000&60000&$2^{20000}$&31.043056&10002&\nodata\\
\hline
QPE: exact (symbolic entries)&210&5075&$2^{70}$&0.071082&72&\nodata\\
QPE: exact (symbolic entries)&2100&491750&$2^{700}$&98.428297&702&\nodata\\
QPE: exact- RSA 260 (concrete entry)
&2586&745199&$2^{862}$&15.213478&2&\nodata\\ \hline
&7&13&8&0.000021&5&\begin{tabular}{cc}
8.7 & symQV~\cite{bauer2023symqv}\end{tabular}\\
&11&27&32&0.000047&7&\begin{tabular}{cc}
234 &symQV~\cite{bauer2023symqv}\end{tabular}\\
QPE: approximate (concrete entry)
&101&1400&$2^{50}$&0.003662&52&\begin{tabular}{cc}
0.017 &QCEC\cite{burgholzer2022handling}\end{tabular}\\
&300&11025&$2^{90}$&0.020378&92&\nodata\\
&1036&6159&64&0.003416&8&\begin{tabular}{cc}
288,196&\cite{govindankutty2026formally}\end{tabular}\\
&3000&1082250&$2^{900}$&17.322091&902&\nodata\\
\hline
\phantom{BRUS($H^{\otimes n}$)}$\quad$k =100&20&1000&$>2^{999}$&0.001214&200&\nodata\\ 
BRUS($H^{\otimes n}$)&200&10000&$>2^{9999}$&0.042743&200&\nodata\\ 

\cline{2-7}
\phantom{BRUS($H^{\otimes n}$)}$\quad$k =500&200&50000&$>2^{49999}$&0.292615&1000&\nodata\\ 
&2000&500000&$>2^{499999}$&23.995457&1000&\nodata\\
\hline
\phantom{BRUS($QPE$)}$\quad$k =5&300&55125&$>2^{449}$&0.111073&460&\nodata\\ 
BRUS($QPE$)&3000&5411250&$>2^{4499}$&87.830885&4510&\nodata\\ 

\cline{2-7}
\phantom{BRUS($QPE$)}$\quad$k =10&300&110250&$>2^{899}$&0.534205&920&\nodata\\ 
&3000&10822500&$>2^{8999}$&177.056307&9020&\nodata\\
\hline
&7&16&4&0.000043&2&\nodata\\
&49&112&16384&0.000677&20&\nodata\\
Quantum error correction&700&1600&$2^{200}$&0.100990&299&\nodata\\
&14000&32000&$2^{4000}$&54.251738&5999&\nodata\\
\hline
\hline
&16&368&1&0.003037&32&\begin{tabular}{cc}     >300&ECMC~\cite{mei2024equivalence}\\0.03&SQV~\cite{ricciardi2025quantum}\\1.818&QuPRS~\cite{pathsumyfc}
\end{tabular}
\\
\cline{2-7}
Unitary equivalence: QFT&32&1248&1&0.024697&64&\begin{tabular}{cc}
     1.43&ZX~\cite{peham2022equivalence}\\0.04&QCEC~\cite{peham2022equivalence}\\ 
     \end{tabular}
\\
\cline{2-7}
&75&6150&1&0.336248&150&
\begin{tabular}{cc}
    1.23&ZX~\cite{peham2022equivalence}\\>3600&QCEC~\cite{peham2022equivalence}\end{tabular}\\
\cline{2-7}
&500&253500&1&132.308794&1000&\nodata\\
\hline

\hline
\end{tabular}

\begin{center}
    \begin{tabular}{@{}l@{\quad}l@{}}
        Computation times in seconds&$>x$: time out after x seconds\\ 
        * Times reported in the cited articles&\nodata{}: No data available
    \end{tabular}
\end{center} 
\label{table:xps}
\end{table}

\myparagraph{System Configuration} All the following experiments were conducted on a Dell
laptop equipped with an Intel Core i7-12800H CPU (14 cores, 20 threads, up to 
\SI{4.8}{GHz}), \SI{64}{GB} DDR5 RAM, and an NVIDIA T600 Laptop GPU, running Ubuntu 24.04.2
LTS with Linux kernel 6.14.

\smallskip

The \textbf{Quantum Teleportation} experiment is the parameterized
version of our running example protocol, as described in
\cref{sec:walking_through_telep}, where Alice sends an entire quantum register of
varying size to Bob, instead of a single qubit. As reported in
\cref{table:xps}, \hqbricks{} scales up to $10000$  teleported qubits in
$\simeq 30$ seconds (the total circuit requires $3$ quantum wires and $2$
classical wires per teleported qubit, as appears in \cref{QTelcirc}). \fv{We are not aware of any  other experiment of symbolic execution for
teleportation at such a scale. Still, symQV~\cite{bauer2023symqv} and
Tensored Decision Diagrams from~\cite{hong2022equivalence} provide experiments
for the teleportation of a single qubit (using 5 wires). \hqbricks outperforms them
by $6$ and $3$ orders of magnitude, respectively.}

\smallskip
The \textbf{Quantum Phase Estimation (QPE)} circuit \fv{is unitary. However, our
experiments integrates  measuring this circuit and deriving its textbook
specifications in terms of the probability of obtaining a desired classical result
 from this measurement.}  QPE contains a call to an oracle family  $U(k)$. In our
implementation,   we used  a parametric phase encoding circuit family $U_{\P}$.
It is  built by parallel occurrences of $Z_k$ rotations and performs the state
transformation $\ket{j}{} \rightarrow e^{\frac{k *2i\pi j}{2^n}} \ket{j}{}$.
Hence, any basis state $\ket{j}{}$ such that $0\leq j < 2^n$, $\ket{j}{}$  is an
eigenstate of $U(k)$ with corresponding eigenvalue $e^{\frac{k *2i\pi j}{2^n}}$.
Our experiments are led both in the
exact and the approximate case.  In the \emph{exact} case, we
used symbolic inputs, so that the verification holds for any $0\leq j
< 2^n$, just as observed above. To illustrate a concrete case and showcase our scaling, we also performed the computation over 
the RSA-260 secret-key  challenge of the RSA Security LLC lab~\cite{rsalab} (with unknown factor decomposition at present time).

In previous attempt, symQV~\cite{bauer2023symqv}  provided two instances (with respectively $7$ and $11$ wires, using the same oracle as us)\footnote{in~\cite{bauer2023symqv, burgholzer2022handling, govindankutty2026formally}, QPE instances are indexed after the cardinal of the top computing register (register \cre from \cref{fig:qpe-exec}) whereas we index them after the whole number of quantum and classical registers. Therefore, our $7$, $11$, $101$, and $1036$ instances are respectively designated as QPE$-3$ and QPE$-5$ in \cite{bauer2023symqv}, QPE$-50$ in~\cite{burgholzer2022handling}, and QPE$-1024$ in~\cite{govindankutty2026formally}.} of QPE in the approximate case. As shown in \cref{table:xps},  \hqbricks performs better by at least $5$ orders of magnitude. \fv{Another occurrence, also using the same oracle\footnote{More precisely, experiments from~\cite{bauer2023symqv, burgholzer2022handling} also make use of the phase shift oracle family, but it is not implemented from elementary gates in the circuit. Instead, they provide the full oracles as elementary gate constructs, together with their functional interpretation. For sake of similarity, in our instances comparing with~\cite{bauer2023symqv, burgholzer2022handling}, we used the same abstractions.},  comes from QCEC~\cite{burgholzer2022handling} and the case with $50$ qubits and concrete entries. The reported performance is largely outperformed by \hqbricks. In a recent proposition dedicated to the formal verification of QPE~\cite{govindankutty2026formally}, authors exhibits cases, still using the same oracle, with up to $1 024$ qubits in the bottom register (\textit{phase qubits})  and $6$ qubits in the top register (\textit{precision qubits}).  \hqbricks{} runs faster by almost  $5$ orders of magnitude.} 

In addition to the results reported in \cref{table:xps}, we performed a
robustness test for QPE calculation: it consists of iterating $100\ 000$ runs,
with randomly chosen entries of 100 qubits and  a correctness verification
check. The whole test lasted  $44$ minutes with no reported failure.
Note also that, interestingly, our tool performs with similar times  
for the exact and for the approximate case.

\smallskip
The \textbf{Bounded Repeat Until Success} is also parametrized by a unitary
subprocess. We led our experiments on both the \emph{parallel Hadamard circuit} \fv{and instances of the Phase Estimation introduced above}. 
\fv{As an example,  the Phase Estimation case with $300$ wires seeks the best $90$ bits approximation of a  $120$ phase qubits eigenvalue. \hqbricks{} verifies, in 0.534 seconds,  that the process would  succeed 
with probability $> 0.9945$ after $10$ tries.  
While Repeat-Until-Success (also often called \textit{post-selection}) is a fundamental pattern from quantum computing, the present result  is  the first attempt to validate it symbolically}.

\smallskip
\fv{Symbolic execution of \textbf{quantum error correction}~\cite{fang2024symbolic,huang2025efficient} and fault-tolerance (i.e., that errors propagate controllably)~\cite{Verifying_Fault_Chen_2025} was recently addressed. 
These works 
 are
 specialized to the modeling of error correction, featuring various correction
schemes and state injection. }
\fv{ What our experiments highlight, however, is the flexibility of our approach, 
which is able to address this class of problem without any extension. 
For instance, while the codes are different, making the comparison not fully fair, we verified the 3-qubit code with $100$ logical qubits, each one requiring $5$ quantum wires and
$2$ classical wires in $0.101$ seconds, while, e.g., experiments from~\cite{huang2025efficient}
compute a 16-wire-large case featuring the $[[6,1,3]]$ stabilizer code in
\SI{0.252}{s}. These results establish our approach as a promising candidate
for the symbolic verification of error correction at the FTQC scale.}

\smallskip

\textbf{Unitary circuit equivalence: the Quantum Fourier Transform.}
\fv{As a side effect, our implementation  also improve the analysis for the restricted unitary case.
It is a well-known fact that, for unitaries, the equivalence problem (checking whether two circuits $C_1$ and $C_2$ are functionally equivalent) reduces to a case of specifications verification (namely, whether the sequence $C_1 ;C^{\dagger}_2$ behaves as \Skip). To illustrate the performance of \hqbricks{} in this setting, we conducted equivalence checking experiments between the standard version of the Quantum Fourier Transform and variants obtained by shuffling gates that commute and adding Clifford gadgets (combinations of $X, H$ and $Z$ gates) with a neutral net effect. \cref{table:xps} reports  these experiments for input sizes varying from $16$ to $500$ qubits, and compares them  with similar experiments from the state of the art. 
\hqbricks{} performs better than any found comparative experiments by a factor of, at least $\frac{0.04}{0.024697}\simeq 1.6$ (QCEC for the 32 qubits case), and orders of magnitude in most cases.}

\myparagraph{Conclusion}
\hqbricks{} performs very well for a wide range of standard patterns from quantum  computing,
including  communication protocols,
post-selection, gate-to-gate circuit composition, and quantum error correction. 

\disc{The only comparison we found in the literature is from the
symbolic execution symQV tool~\cite{bauer2023symqv}, with 
the single qubit teleportation case and Another significant comparative consideration is to highlight that, as opposed to~\cite{bauer2023symqv}, \hqbricks features cases of industrial cases scale (up to thousands of qubits).
}

\newcommand{\myfootnote}[1]{\color{blu}#1\color{black}}

\section{Related Work}
In the last decade, several efforts have aimed to bring formal methods to quantum
programming~\cite{chareton2021automated, qhlprover,
hietala2020proving,coqq,lewis2024automated}. 
We extend prior work on path-sum analyses~\cite{chareton2021automated,
Vilmart2020SOP,amy2018towards,amy2023complete,deng2024case} to the hybrid case.
This allows tractable semantics and specifications, avoiding heavy
formalized density operators while providing flexibility. Flexibility is both
horizontal --- the solution supports both quantum and classical
computation, and vertical ---  both low-/high-level programming
are handled through generic constructors such as the quantum control and gate-to-gate circuit composition.
\smallskip

Our overall proposal gets inspiration from different ingredients that have been
developed separately in the literature.
First are
path-sums~\cite{amy2018towards,Vilmart2020SOP} (for the unitary parts),
including unbalanced~\cite{amy2023complete,deng2024case} or
parameterized~\cite{chareton2021automated,deng2024case} versions. As an extension
of path-sums, the closest work is certainly~\cite{Vilmart2020SOP}. With regard
to this development, our main contributions are the derivation of hybrid
path-sums from \hqbricks{} quantum programs, the formalization of both quantum
and classical control, and the extension of rewriting and simplification tools
with local reasoning facilities.
Additional inspirations for this work include hybrid state modeling~\cite{qhlcv,
unruh2021quantum} for the non-unitary parts and quantum Hoare
logic~\cite{qhlprover, ying2012floyd, unruh2019quantumghost} for
specification.    
\smallskip

An alternative line of research considers 
the  equivalence checking problem for instantiated
quantum circuits. It benefits 
from a variety of symbolic representations
(automata~\cite{chen2023autoq, chen2023automata},
path-sums~\cite{amy2018towards}, binary decision
graphs~\cite{sistla2023symbolic,10.1145/3651157}, etc.) and/or verification
techniques (symbolic execution~\cite{burgholzer2021qcec, bauer2023symqv}, model
checking~\cite{10.1007/978-3-540-70545-1_51,10.1007/978-3-540-70545-1_51},
assisted proofs~\cite{coqq,qhlprover}, etc). 

\smallskip
\fv{Concerning the analysis  of hybrid classical/quantum computing, several propositions~\cite{burgholzer2021qcec,hong2022equivalence,ricciardi2025quantum} take profit on the \textit{deferred measurement principle}. An important bottleneck of this approach is its intrinsic limitation to hybrid circuits with `unitary behavior' in the sense that discarding non important parts of the memory is forbidden. This severely limiting aspect  was identified and provided partial solutions in~\cite{hong2022equivalence, ricciardi2025quantum}.}

\smallskip
A last line of  research considers fully automatic verification of
simpler properties of quantum programs,  such as  the validity requirements
pointed out in
Section~\ref{sec:hqbricks-syntax}~\cite{rand2019formal,randthesis,garrigue2023typed},
or state abstractions~\cite{yu2021quantum, perdrix2008quantum,
bichsel2023abstraqt}. These works also include state-dependent properties, such
as the management (allocation and freeing) of ancillae, with automated
guaranteed synthesis of uncomputation
subroutines~\cite{silq,paradis2021unqomp,paradis2024reqomp, seidel2024qrisp}.

\section{Conclusion}%
\label{sec:discussion}
In this paper, we introduce the \hps{} formalism,   a  framework for the verification of
hybrid quantum programs that includes a compact symbolic representation, formal reasoning
capabilities, a core assertion language and a symbolic execution mechanism. This framework has
been implemented into a working prototype, and we demonstrated its relevance through a diverse
panel of typical hybrid case studies, ranging from quantum communication to computing and error
correction.

A current limitation of this work is that the control structures are limited to bounded loops.
However, since  \hps{} contain all the necessary information for characterizing their associated multiworld structure in terms of measurable spaces, we plan to generalize their usage to unbounded
loops. Other directions for future work include investigating the (possible extension towards)
completeness of our HPS equational theory, extension to parametrized hybrid quantum programs
and potential integration of SMT solvers.

\section*{Data-Availability Statement}

The reproduction package is available on Zenodo at \url{https://doi.org/10.5281/zenodo.19080601} (for the evaluated version, see\cite{HybridPathSums-artifact}). For reuse, the project is available on GitHub at \url{https://github.com/Qbricks/hqbricks}. \todo{un lien vers version avec supplementary material ?}

\begin{acks}
This work has been partially funded by the French National
 Research Agency (ANR)
 within the
framework of “Plan France 2030”, under the research projects EPIQ ANR-22-
PETQ-0007 and HQI-R\&D ANR-22-PNCQ-0002.
\end{acks}

 \vfill
 \bibliographystyle{ACM-Reference-Format}
 \bibliography{main}
 \newpage
\appendix
  \begin{center}
    \huge
\textbf{Supplementary material for the article}
\\
\vspace{.3cm}
\textbf{``Hybrid Path-Sums for Hybrid Quantum Programs''}
\vspace{.7cm}
\end{center}

\normalsize

  \section{Proofs and Supplementary Theorems}%
\label{app:proofs}

\subsection{Supplementary Material Helpful for the Proofs}
\begin{theorem}[Conservation of probabilities]%
    \label{thm:conservation-prob}
    For any program $p$, \hps{} $\h$, and history $\eta$, we have
    \[
    |\xi(\h)(\eta)|^2 = \sum_{\eta' \geq \eta} |\xi(\semantics{p}(\h))(\eta')|^2
    \]
    In other words, the probability of being in a future of $\eta$ after performing the program $p$ is the same as the probability of having been in $\eta$ before performing $p$. 
\end{theorem}
\begin{proof}

    The proof is by induction on $p$. We elaborate on the case when $p$ is a
    \QProg. In such a case, $\llbracket p \rrbracket$ has the effect of a
    unitary operation on $\h$ possibly parametrized by classical variables (e.g.
    classical/hybrid conditioning) and does not change any classical memory;
    therefore, $\xi(\h)$ and $\xi(\llbracket p \rrbracket (\h))$ have the same
    supports of histories of uniform length $l$, so that if $\xi(\llbracket p
    \rrbracket (\h)(\eta')) \neq 0$ and $\xi(\h)(\eta) \neq 0$ and $\eta' \geq
    \eta$, then, $\eta = \eta'$. On the other hand, we have $\xi(\llbracket p
    \rrbracket (\h))(\eta) = U_\eta \xi(\h)(\eta)$. Therefore, $\sum_{\eta' \geq
    \eta}|\xi(\llbracket p \rrbracket (\h))(\eta')|^2 = |\xi(\llbracket p
    \rrbracket (\h))(\eta')|^2 = |\xi(\h)(\eta)|^2 $. 

    Note: up to this theorem, the set $\U$ of operators had no explicit reason to be
    limited to unitaries; this theorem is what imposes unitarity.

    We will not elaborate fully on all the constructs of programs, but we note
    that the proof relies heavily on the properties of the temporal order
    $\leq$, the $L^2$ norm, the Born rule, and the structure of $\psadd$ with
    respect to $\xi$. Notably, the cases of sequential composition and
    \textbf{For} require the transitivity of $\leq$. Classical assignment does
    not change quantum states, but pushes the histories forward one step in
    time, and the equality in the case of measurement is justified by the
    coherence of the Born rule with the $L^2$ justifies the equality. Finally,
    for \textbf{If}s with classical conditions, we rely heavily on the fact that
    $\kappa$ being classical guarantees that $\xi(\kappa \h)$ and 
    $\xi(\bar \kappa \h)$ have disjoint supports of histories, so that the
    $\psadd$ operation in  $\semantics p (\kappa \h) \psadd \semantics q 
    (\bar \kappa \h)$ corresponds to a disjoint union, and does not allow
    different worlds to interfere.
    \end{proof}
    
The conservation of probabilities is a generalization of unitarity for hybrid
programs, and closely parallels the preservation of trace in the language of
density operators.

\begin{corollary}[Conservation of histories]%
    \label{thm:conservation-histories}
    For any \hps{} $\h$ and any program $p$, if $\xi(\h)$ is supported on $S = \{ \eta | \xi(\h)(\eta) \neq 0\}$, then $\xi(\llbracket p \rrbracket (\h))$ is supported exclusively on futures of $S$; that is, if $\xi(\llbracket p \rrbracket (\h))(\eta) \neq 0$, then $\eta$ is the future of some $\eta' \in S$.

    In particular, if $\xi(\h)$ and $\xi(\h')$ have disjoint support of
    histories of the same length $l$, and $p$ and $q$ are any two programs,
    $\xi(\llbracket p \rrbracket (\h))$ and $\xi(\llbracket q \rrbracket (\h'))$
    also have disjoint supports, and $\h \psadd \h'$ refers purely to the union
    part of the operator $\psadd$. 
\end{corollary}

\begin{definition}[History and time]

    A history $\eta$ of length/age $l \in \mathbb N$ (denoted $|\eta| = l$) is a
    stack of length $l$ of address-indexed tuples of booleans all of the same
    length $l$.  A history represents the values taken by the classical memory
    over $l$ time-steps. As stacks, histories are subject to prefix functions
    $\pi_{l'}$ and suffix functions $\sigma_{l'}$. When $l' \leq |\eta|$, the
    $\pi_{l'}(\eta)$ and $\sigma_{l'}(\eta)$ are the prefix and suffix of $\eta$
    of length $l'$ respectively. When $l' > |\eta|$, $\pi_{l'}(\eta) = \eta
    \histconcat {\sigma_1(\eta)}^{(l'-|\eta|)}$ and $\sigma_{l'}(\eta) =
    {\pi_1(\eta)}^{(l'-|\eta|)} \histconcat \eta$ where $\histconcat$ is
    concatenation and $\eta^n$ is $\bighistconcat_i^n (\eta)$. These three
    operations $|\cdot|$, $\pi$, and
    $\sigma$ can be extended to \hps{} by applying to their symbolic
    histories in their classical output. $f(\pasum{\P}{\N}{\F_{\text{Cl}}
    \F_{\text{Qu}}}) = \pasum{\P}{\N}{f(\F_\text{Cl}) \F_{\text{Qu}}}$ for $f
    \in \{ \pi, \sigma\}$ and $\text{age}(\h) = |\F_{\text{Cl}}|$ which is uniform
    across all paths. We then get three orders on histories/\hps{}. We give them
    for histories, and they are defined exactly the same for \hps{}.
    \[
        \begin{array}{r r c c l}
        \text{Age}:& \text{young / old}: \quad & \eta \leq_a \eta' \quad &\equivdef
        \quad & |\eta| \leq |\eta'| \\
        \text{Time}:& \text{past / future}: \quad & \eta \leq_\pi \eta' \quad
                                          &\equivdef \quad & |\eta| \leq
                                          |\eta'| \land \pi_{|\eta|}(\eta') =
                                          \eta \\
        \text{Memory}:& \text{veiling / awakening}: \quad & \eta \leq_\sigma \eta' \quad
                                                  &\equivdef \quad & |\eta|
                                                  \leq |\eta'| \land
                                                  \sigma_{|\eta|}(\eta') = \eta
        \end{array}
    \]
    
\end{definition}

\begin{proof} (Sketch)

The construction of $\phi$ is by induction on the constructors of the type
\Prog, after normalization of the terms according to the rules of the
simply-typed $\lambda$-calculus introduced. In this proof sketch, we select only
a sample of the more interesting translation cases from \hqbricks{} to
\textsc{qhl-cv} to display.
\begin{itemize}
    \item If $p : \QProg$, $p$ represents a unitary, so we can define $\phi(p) = (\bar q := U)$.
    \item When the set of programs is quotiented over the monoid axioms 
    corresponding to the sequencing operator $;$ and its identity 
    $\textbf{Skip}$ 0-- that is associativity of $;$ and identity of 
    \textbf{Skip} with respect to $;$ 0-- $\phi$ is a morphism of monoids; that 
    is, $\phi(p;q) = 
    \phi(p);\phi(q)$ and $\phi(\textbf{Skip}) = \textbf{Skip}$.
    \item A closed \textbf{For} term has its bounds instantiated; therefore, it
    can be translated directly as a sequencing of its terms: $\phi(\for{x_n^I} i
    j p) = \phi(p[x_n^I:=i] ; \ldots ; p[x_n^I := j])$. In terms of the quotient through the monoid rules above,
    this expresses that $\for{x_n^I} i  j p$ is definable in the monoid as $\prod_{x_n^I = i}^j p$.
\end{itemize}

The reason such quotienting can be meaningfully used is that the interpretation of
programs $\llbracket \cdot \rrbracket$ remains well-defined modulo this quotient.
With these definitions laid down, the proof of $\llbracket \phi(p)
\rrbracket(\phi_{\llbracket\rrbracket}(\h)) =
\phi_{\semantics{}}(\llbracket p \rrbracket(\h))$ is a relatively
automatic induction on $p$. 

\end{proof}

\begin{lemma}[Linearity and consistency of $\xi \circ \llbracket p \rrbracket$]
    We have $\forall \prog, \xi(\bh) = \xi(\bh') \implies
    \xi(\semantics{\prog}(\bh)) = \xi(\semantics{\prog}(\bh')$ and $(\xi \circ \llbracket p
    \rrbracket)(\h \psadd \h') = (\xi \circ \llbracket p \rrbracket) (\h) + (\xi
    \circ \llbracket p \rrbracket) (\h')$. Furthermore, we have $\bh \prequiv
    \bh' \implies \llbracket p \rrbracket (\h) \prequiv \llbracket p \rrbracket
    (\h')$ 
\end{lemma}

\begin{proof}
    Any program $p$ has an associated \hpps{} $\bh_p$ with one free variable of
    type \hps{} such that, $\forall \h, \llbracket p \rrbracket(\h) = \h_p[\h]$.
    This can be seen immediately from the definition of $\llbracket \cdot
    \rrbracket$. The immediate consequence is that $\h \prequiv \h' \implies
    \llbracket p \rrbracket (\h) \prequiv \llbracket p \rrbracket (\h')$.

    Furthermore, in the proof of soundness of the Rewrite rule, we also proved a
    slightly different property that $\llbracket p \rrbracket$ respects equality
    modulo $\xi$: $\xi(\h) = \xi(\h') \implies \xi(\llbracket p \rrbracket (\h)) =
    \xi(\llbracket p \rrbracket (\h'))$. Linearity of $\xi \circ \llbracket p
    \rrbracket$ can be proven by induction on $p$ where, for the case of
    sequencing, we have used that $\llbracket p \rrbracket$ respects equality
    modulo $\xi$ and the inductive hypotheses. We note that linearity is trivial
    for \textbf{Skip} and unitaries.

    \begin{minipage}{0.30\textwidth}
    \begin{align*}
        &\xi(\llbracket p ; q \rrbracket (\h \psadd \h')) \\
        =&\xi(\llbracket q \rrbracket (\llbracket p \rrbracket (\h \psadd \h'))) \\
        =&\xi(\llbracket q \rrbracket (\llbracket p \rrbracket (\h) \psadd \llbracket p \rrbracket (\h')) \\
        =&\xi(\llbracket p ;q \rrbracket (\h) + \llbracket p ;q \rrbracket (\h')) 
    \end{align*}
    \end{minipage} \hfill
    \begin{minipage}{0.55\textwidth}
    \begin{align*}
        &\xi(\llbracket \Ifthene{\kappa}{p}{q} \rrbracket(\h \psadd\h')) \\
        =&\xi(\llbracket p \rrbracket (\kappa\h \psadd \kappa\h')) \psadd \llbracket q \rrbracket (\bar \kappa\h \psadd \bar \kappa\h')) \\
        =&\xi(\llbracket p \rrbracket (\kappa\h \psadd \kappa\h')) ) + \xi (\llbracket q \rrbracket (\bar \kappa\h \psadd \bar \kappa\h')) \\
        =&\xi(\llbracket p \rrbracket (\kappa\h)) + \xi( \llbracket p \rrbracket(\kappa\h')) + \xi (\llbracket q \rrbracket (\bar \kappa\h) + \llbracket q \rrbracket (\bar \kappa\h')) \\
        =&\xi(\llbracket p \rrbracket (\kappa\h) \psadd \llbracket q \rrbracket (\bar \kappa\h) + \xi( \llbracket p \rrbracket(\kappa\h') \psadd \llbracket q \rrbracket (\bar \kappa\h'))
    \end{align*}
    \end{minipage}

    \vskip 1em

Finally, both measurement and classical assignment are covered by the \hps{} form
$\h[c \leftarrow b]$ where $c$ is a classical register. Recall that $\xi(\h[c
\leftarrow b])$ is null on any history that does not have $b$ on the top of the
stack of $c$, so we will only consider the situations where it is.
\begin{align*}
    \xi((\h \psadd\h')[c \leftarrow b]) (\eta, c \leftarrow b)
    = &\xi((\h \psadd\h'))(\eta) \\
    = &\xi(\h)(\eta) + \xi(\h')(\eta) \\
    = &\xi(\h[c \leftarrow b])(\eta, c\leftarrow b) + \xi(\h'[c\leftarrow b])(\eta, c\leftarrow b)
\end{align*}

\end{proof}

\subsection{Relative Soundness of the Language}%
\label{full-proof-soundness-language}

We wish to show that our \hqbricks{} language and its \hps{} representation is
sound with respect to cq-states introduced in~\cite{qhlcv}

A cq-state $\Delta$ over a set of quantum variables $V$ is a function $\Sigma
\to \mathcal{D}(\mathcal{H}_V)$, where $\Sigma$ is the set of classical states
(an assignment of each classical variable to a boolean or integer value) and
$\mathcal{D}(\mathcal{H}_V)$ is the set of partial density operators on (the
Hilbert space generated by) $V$.
Furthermore, to be a cq-state $\Delta$ needs to verify the following conditions:
\begin{itemize}
    \item its support is countable
    \item $\textit{trace}(\Delta) := \sum\limits_{\sigma \in \Delta} \textit{trace}(\Delta(\sigma)) \leq 1$
\end{itemize}
That is for each possible value of the classical states, a (partial) density
operator is given such that the sum of their traces is at most 1.

The language of \hqbricks{} is not the same as the language used in~\cite{qhlcv},
and states in \hqbricks{} are \hpps{} while in~\cite{qhlcv}, they are cq-states. In
order to compare the two, we need to define a translation of languages
$\hqbricks \to \textsc{qhl-cv}$ as well as a translation of states $\hpps \to
(\Sigma \to \mathcal D(\mathcal H_V))$.

\begin{theorem}[Relative soundness of the language]%
\label{theo:language-soundness-app}
\hqbricks{} is sound with respect to \textsc{qhl-cv}: there is a mapping
$\phi_{\llbracket\rrbracket}$ of \hps{} into generalized cq-states, whose
restriction to \hps{} with norm lesser or equal to $1$ lends into cq-states, and a
mapping $\phi$ of closed \hqbricks{} \Prog{} terms into \textsc{qhl-cv} programs
such that \[\llbracket \phi(p) \rrbracket(\phi_{\llbracket\rrbracket}(\h)) =
\phi_{\llbracket\rrbracket}(\llbracket p \rrbracket(\h))\]
\end{theorem}

\begin{remark}
    $\phi$ is not surjective. For example, \hqbricks{} explicitly excludes
    unbounded loops, so that some \textsc{qhl-cv} \textbf{While} programs are
    not images of any \hqbricks{} program. More specifically, the image of $\phi$
    includes bounded while loops via $k := i; \textbf{while } (\kappa \land k
    \leq j) \textbf{ do } p; k := k + 1 \textbf{ done} \equiv p[k:=i] ; \ldots ;
    p[k:=j]$, and excludes general while loops, and the oracles of probabilistic
    classical assignment. However, the language of \hqbricks{} is open to
    extension by arbitrary unitaries. If we add unitaries corresponding to
    probability distributions $p$ such that $U_p : |\vec 0\rangle \mapsto
    \sqrt{p_i}|i\rangle$, probabilistic classical assignment $c :=_\$ p$ is also
    possible through the \hqbricks{} program $\textbf{Reset}(q); U_p(q);
    \textbf{Measure}(q, c)$. In essence, while $\phi$ is not surjective,
    it effectively only misses, intentionally,
    unbounded while loops, so the language of \hqbricks{} can be seen as an
    extension of that of \textsc{qhl-cv} with the restriction of bounded loops.
\end{remark}

\begin{proof} of theorem~\ref{theo:language-soundness}
    For a closed \hpps{} $\bh$, $\phi_{\semantics{}} (\bh) = \rho(\bh)$. For programs, we have a translation of QProg terms into unitaries as follows:
    \begin{align*}
        U(\textbf{Skip}) &= I \\
        U(U(q)) &= I_m \otimes U \otimes I_n \\
        U(p; q) &= U(q) U(p) \\
        U(\Ifthene{\kappa}{p}{q}) &= U(p)  P_{\kappa}  + U(q) P_{\bar \kappa}
    \end{align*}
    where $P_{\kappa}$ is the matrix defined for computational basis states $|\psi\rangle$ by $P_{\kappa} |\psi\rangle = \psi$ if $|\psi\rangle$ satisfies $\kappa$ given the state of the classical registers and $P_{\kappa} |\psi\rangle = 0$ otherwise. For a \textbf{closed} term of type Prog; that is, a program with all parameters specified,
    \begin{alignat*}{2}
        \phi(\textbf{Measure}(q, c)) &= c := \textbf{ meas } \mathcal M[q] \quad && \\
        \phi(c_1 := f(\lceil c_2 \rceil)) &= c_1 := c_2 \\
        \phi(p; q) &= \phi(p) ; \phi(q) \\
        \phi(\for{x^I_n} i j p)& =\\
        \multicolumn{3}{r}{$\phi(p[x_n^I := i]) ; \phi(p[x_n^I := i+1]) ; \cdots ; \phi(p[x_n^I := j])$}& \\  
        \phi(\Ifthene{\kappa}{p}{q}) &= \textbf{If } \kappa \textbf{ then } \phi(p) \textbf{ else } \phi(q) \textbf{ end}  && \text{ when }\kappa : \CBool \\
        \phi(p) &= \bar q *= U(p) && \text{ when } p : QProg \text{ and } \qreg(p) \subseteq q \\
    \end{alignat*}

    $\phi$ respects the monoid axioms for $;$ and \textbf{Skip}. Explicitly, 
    $\phi(\textbf{Skip}) = \textbf{Skip}$ and $\phi(p; q) = \phi(p) ; \phi(q)$
    by definition.

    The proof of the theorem follows by induction on program terms. We give the proof in full detail below.
    
    \begin{minipage}{0.45\textwidth}
    \begin{align*}
        &\llbracket \phi(\textbf{Skip}) \rrbracket(\phi(\bh)) \\
        =& \llbracket \textbf{Skip} \rrbracket(\phi(\bh)) \\   
        =& \phi(\bh) \\   
        =& \phi(\llbracket \textbf{Skip} \rrbracket (\bh))
    \end{align*}
    \end{minipage} \hfill
    \begin{minipage}{0.45\textwidth}
    \begin{align*}
        &\llbracket \phi(p; q) \rrbracket (\phi(\bh)) \\
        =& \llbracket \phi(p); \phi(q) \rrbracket (\phi(\bh)) \\
        =& \llbracket \phi(q) \rrbracket (\llbracket \phi(p) \rrbracket (\phi(\bh))) \\
        =& \phi(\llbracket q \rrbracket(\llbracket p \rrbracket (\bh)) \\
        =& \phi(\llbracket p; q \rrbracket (\bh)) \\
    \end{align*}
    \end{minipage}

    \begin{minipage}{0.45\textwidth}
    \begin{align*}
        &\llbracket \phi(c_1 := f(\lceil c_2 \rceil)) \rrbracket (\phi(\bh))\\
        = &\llbracket c_1 := f(c_2) \rrbracket (\bigoplus \langle c, \sum_{w, c(w) = c} |w\rangle \langle w| \rangle) \\
        = &(\bigoplus \langle c[c(f(c_2))/c_1], \sum_{w, c(w) = c} |w\rangle \langle w| \rangle) \\
        = &\phi(\bh[c_1 \leftarrow \lceil f(c_2) \rceil])
    \end{align*}
    \end{minipage} \hfill
    \begin{minipage}{0.45\textwidth}
    \begin{align*}
        &\llbracket \phi(\textbf{Measure }(q, c)) \rrbracket (\phi(\bh)) \\
        =& \llbracket c := \textbf{ meas } \mathcal M[q] \rrbracket (\bigoplus \langle c, \sum_{w, c(w) = c} |w\rangle \langle w|\rangle)  \\
        =& \sum_i \bigoplus \langle c[i/x], \sum_{w, c(w) = c} M_i |w\rangle \langle w| M_i^\dagger \rangle\\ 
        =& \phi(\bh[q \leftarrow c]) \\
        =& \phi(\llbracket \textbf{Measure }(q, c) \rrbracket(\bh))
    \end{align*}
    \end{minipage}

    If $\kappa : \CBool$,
    \begin{align*}
        &\llbracket \phi(\Ifthene{\kappa}{p}{q}) \rrbracket (\phi(\bh))\\
        =& \llbracket \textbf{If } \kappa \textbf{ then } \phi(p) \textbf{ else } \phi(q) \textbf{ end}) \rrbracket (\phi(\bh)) \\
        =& \llbracket \textbf{If } \kappa \textbf{ then } \phi(p) \textbf{ else } \phi(q) \textbf{ end}) \rrbracket (\bigoplus_{c} \langle c, \sum_{w,c(w) = c} |w\rangle \langle w| \rangle) \\
        =& \llbracket\phi(p) \rrbracket (\bigoplus_{c, c\vDash \kappa} \langle c, \sum_{w,c(w) = c} |w\rangle \langle w| \rangle) +\llbracket\phi(q)\rrbracket (\bigoplus_{c, c\vDash \bar \kappa} \langle c, \sum_{w,c(w) = c} |w\rangle \langle w| \rangle) \\
        =& \llbracket\phi(p) \rrbracket (\phi(\bkappa \bh)) +\llbracket\phi(q)\rrbracket (\phi(\bar \bkappa \bh)) \\
        =& \phi(\llbracket p \rrbracket (\bkappa \bh)) + \phi(\llbracket q\rrbracket (\bar \bkappa \bh)) \\
        =& \phi(\llbracket p \rrbracket (\bkappa \bh) \psadd \llbracket q\rrbracket (\bar \bkappa \bh)) \\
        =& \phi(\llbracket \Ifthene{\kappa}{p}{q} \rrbracket (\bh))
    \end{align*}

    If $p : \QProg$,
    \begin{align*}
        &\llbracket \phi(p) \rrbracket (\phi(\bh))\\
        =& \llbracket U(p) \rrbracket (\bigoplus_{c} \langle c, \sum_{w,c(w) = c} |w\rangle \langle w| \rangle) \\
        =& (\bigoplus_{c} \langle c, \sum_{w,c(w) = c} U(p) |w\rangle \langle w|
        {U(p)}^\dagger \rangle) \\
        =& (\bigoplus_{c} \langle c, \sum_{w,c(w) = c} |U(p) w\rangle \langle U(p) w| \rangle) \\
        =& \phi(\llbracket p \rrbracket (\bh))
    \end{align*}

    \begin{align*}
        \llbracket \phi(\for{x_n^I} i j p) \rrbracket (\phi(\bh))
        &= \llbracket \phi(p[x^I_n := i]) ; \phi(p[x^I_n := i+1] ; \ldots ; \phi([x^I_n := j]) \rrbracket (\phi(\bh)) \\
        &= \llbracket \phi(p[x^I_n := j]) \rrbracket \circ \cdots \circ \llbracket \phi(p[x^I_n := i] \rrbracket (\phi(\bh)) \\
        &= \llbracket \phi(p[x^I_n := j]) \rrbracket \circ \cdots \circ \llbracket p[x^I_n := i+1] \rrbracket (\phi(\llbracket p[x^I_n := i] \rrbracket (\bh))) \\
        &= \cdots \\
        &= \phi(\llbracket p[x^I_n := j] \rrbracket \circ \cdots \circ \llbracket p[x^I_n := i] \rrbracket (\bh))) \\
        &= \phi(\llbracket p[x^I_n := i] ; \ldots ; p[x^I_n := j] \rrbracket (\bh))) \\
        &= \phi(\llbracket \for{x_n^I} i j p \rrbracket (\bh))
    \end{align*}
\end{proof}

\subsection{Equational Theory}%
\label{app:equational-theory}

\begin{theorem}[Soundness of \hps{} reductions]
    For any closed \hpps{} $\bh$ and $\bh'$,
    \begin{itemize}
        \item If $\bh \prequiv \bh'$ then  $ \xi(\bh) = \xi(\bh')$
        \item If $\bh \equiv \bh'$ then $ \tilde\xi(\bh) = \tilde\xi(\bh')$
    \end{itemize}
\end{theorem}

\begin{proof} of Theorem~\ref{theorem:sound-rewrite-gen}

\begin{prooftree}
    \AxiomC{$y \in \bsupport(\h)$}
    \RightLabel{\myrule{Split}}
    \UnaryInfC{$\h \prequiv \h[y:=0] \psadd \h[y:=1]$}
\end{prooftree}
\begin{prooftree}
    \AxiomC{}
    \RightLabel{\myrule{Scalar}}
    \UnaryInfC{$\pasum{\bP}{\bF}{\alpha \bN} \prequiv \alpha \pasum{\bP}{\bF}{\bN}$}
    \DisplayProof\hskip 1em
    \AxiomC{}
    \RightLabel{\myrule{Phase}}
    \UnaryInfC{$\pasum{\theta + \bP}{\bF}{\bN} \prequiv e^{2 \pi i \theta} \pasum{\bP}{\bF}{\bN}$}
\end{prooftree}

The soundness of these rules follows immediately once $\xi(\h)$ is expanded:
\[
    \xi(\h)(\eta) = \sum_{\Vec y_0, \text{history}(\h[\Vec y:= \Vec y_0]) =
    \eta} \bN(\Vec y_0)\cdot e^{2 \pi i \P(\Vec y_0)} |\F_\text{qu}\rangle
\] 
The Scalar and Phase rules express the factorization of constants out of the
sum, and Split expresses the decomposition of the sum as first summing over $y
\in \Vec y$, then over $\Vec y \setminus \{y\}$. 

The properties of $\xi$ being linear and respecting $\otimes$ imply the
soundness of the axioms of \hps{} forming a $\C$-vector space modulo $\prequiv$
and $\otimes$ being commutative, associative, and bilinear with respect to
$\psadd$ modulo $\prequiv$. In fact, these axioms modulo $\prequiv$ are nothing
more than the pullbacks of the respective axioms for $\xi(\h)$. The soundness of
the \myrule{Plus} and \myrule{Times} rules also follow naturally from that.

\mncom{I removed the mention of \myrule{Split-Merge} as I think it doesn't appear in the current version of the paper.}
We note also that the \myrule{Fact-Distr} rule is
correct by definition: the left-hand side of $\prequiv$ is \textbf{defined} as
the right-hand side.

We introduce another rule derived from \myrule{Split}, namely,
\myrule{Split-Many}, which is useful for the remainder of the proof.
\myrule{Split-Many} consists simply of repeating Split to many path variables;
its use lies in the brevity it provides.

\begin{prooftree}
    \AxiomC{$\Vec y \subseteq \bsupport(\h)$}
    \RightLabel{\myrule{Split-Many} --- $\#(\Vec{y})$}
    \UnaryInfC{$\h \prequiv \psadd_{t\in 2^{\#(\Vec y)}} \h[\Vec y := t]$}
\end{prooftree}

The remaining $\prequiv$ rules are proven sound by their construction using the
given rules.  We give the details of these constructions below.

\myrule{Filter}:
\begin{align*}
    \pasum{\bP}{\bF}{\left(\prod \Vec y\right)\bN}_{\support \cup \Vec y}
     &\prequiv_{\myrule{Split-Many}} \bigpsadd_{\Vec y_0 \in 2^{|\Vec y|} \setminus
     \{1\}} (\pasum{\bP}{\bF}{0\bN}_{\support}) \psadd
     \pasum{\bP}{\bF}{\bN}_{\support} \\
     &\prequiv_{\myrule{Scalar}} 0 \cdot \bigpsadd_{\Vec y_0 \in 2^{|\Vec y|}
     \setminus \{1\}} (\pasum{\bP}{\bF}{\bN}_{\support}) \psadd
     \pasum{\bP}{\bF}{\bN}_{\support} \\
     &\prequiv_{\mathbb C\text{-vector space}}  \pasum{\bP}{\bF}{\bN}_{\support}
     \\
\end{align*}

\myrule{Change-Var}:
\begin{align*}
    \h &\equiv_{\myrule{Split-Many}} \bigpsadd_{\Vec y_0 \in 2^{|\Vec y|}} \h[y :=
    y_0] \\
    &\equiv_{\mathbb C\text{-vector space}} \bigpsadd_{\Vec y_0 \in 2^{|\Vec
    y|}} \h[y := \sigma(y_0)] \\
    &= \bigpsadd_{\Vec y_0 \in 2^{|\Vec y|}} \h[y := \sigma(y)][y := y_0] \\
    &\equiv_{\myrule{Split-Many}} \h[y := \sigma(y)] \\
\end{align*}

\myrule{Phase-Bisector}:
\begin{align*}
\pasum{\bP+y \lift Q}{\bF}{\bN} \prequiv_{\myrule{Split-Many}}& \psadd_{\Vec y'_0
= \bsupport} (\pasum{\bP}{\bF}{\bN} \psadd \pasum{\bP + \lift Q}{\bF}{\bN}) \\
\prequiv_{\myrule{Phase}}& \psadd_{\Vec y'_0 = \bsupport} (\pasum{\bP(\Vec
y'_0)}{\bF(\Vec y'_0)}{\bN(\Vec y'_0)} \psadd e^{2 \pi i \lift Q(\Vec y'_0)}
\pasum{\bP(\Vec y'_0)}{\bF(\Vec y'_0)}{\bN(\Vec y'_0)})  \\
\prequiv_{\mathbb C\text{-vector space}}& \psadd_{\Vec y'_0 = \bsupport} (2|\cos \pi \lift Q(\Vec y'_0)| e^{2 \pi i \frac{\lift Q}{2}})
\pasum{\bP(\Vec y'_0)}{\bF(\Vec y'_0)}{\bN(\Vec y'_0)}\\
\prequiv_{\myrule{Phase, Scalar}}& \pasum{\bP + \frac{\lift Q}{2}}{\bF}{\bN}
\end{align*}

\myrule{HH}:
\begin{align*}
    \pasum{\bP + \frac{y_0 \cdot (y_1 \oplus \lift Q)} 2}{\bF}{\bN}
    \prequiv_{\myrule{Phase-Bisector}}& \pasum{\bP + \frac{y_1 +\oplus \lift
    Q} 4}{\bF}{2 |\cos(\pi (y_1 \oplus \lift Q))| \bN} \\
    =& \pasum{\bP + \frac{y_1 \oplus \lift
    Q} 4}{\bF}{2 (1 \oplus y_1 \oplus \lift Q) \bN} \\
    \prequiv_{\myrule{Change-Var}}& \pasum{\bP[y_1 \leftarrow 1 \oplus y_1 \oplus
    \lift Q] + \frac{1 \oplus y_1} 4}{\bF[y_1 \leftarrow 1 \oplus y_1 \oplus \lift
    Q]}{2 y_1 \bN} \\
    \prequiv_{\myrule{Filter}}& \pasum{\bP[y_1 \leftarrow \lift Q]}{\bF[y_1 \leftarrow \lift Q]}{2 \bN}
\end{align*}

\mncom{I commented the \myrule{Split-Ctrl} part here as I think it doesn't appear anywhere else in the current version of the paper.}

Finally, we target \myrule{Phase-equiv}. At this point, we must loosen the
soundness to equality modulo $\tilde \xi$ instead of $\xi$. Given $\P_2$ is
fully determined by the history since it shares no variables with the classical
memory stacks, we have,

\begin{align*}
     \tilde \xi(\pasum{\bP_1 + \bP_2}{\bF}{\bN}_{\bsupport})
    =&\overline{\sum_{y, \text{history}(\bh[y]) = \eta} \bN(y) \cdot e^{2\pi i
    \bP_1(y) + 2 \pi i \bP_2(\eta)} |\F\rangle_{\text{Qu}}} \\
    =&\overline{\sum_{y, \text{history}(\bh[y]) = \eta} \bN(y) \cdot e^{2\pi i
    \bP_1(y)} |\bF\rangle_{\text{Qu}}} \\
    =&\tilde \xi(\pasum{\bP_1}{\bF}{\bN}_{\bsupport}) 
\end{align*}

\begin{lemma}[Correctness of $\otimes$]%
    \label{lemma:otimes-correct}
    Let $\h$ and $\h'$ be \hps{} satisfying the conditions for forming the
    tensor product. Then, $\xi(\h \otimes \h')(\eta \otimes \eta') =
    \xi(\h)(\eta) \otimes \xi(\h')(\eta')$ for $\eta$ and $\eta'$ histories on
    $\bF_{\text{Cl}}$ and $\bF'_{\text{Cl}}$ respectively.
\end{lemma}
\begin{proof}
    Let $\h_1 = \pasum{\bP_1}{\bF_1}{\bN_1}_{\support_1}$ and $\h_2 =
    \pasum{\bP_2}{\bF_2}{\bN_2}_{\support_2}$ be \hps{} satisfying the
    conditions for forming the tensor product; that is, $\support_1 \cap
    \support_2 = \emptyset$ and $\bF_1 \cap \bF_2 = \emptyset$, and let $\eta_1$
    be a history on $\bF_{{1}_{\text{Cl}}}$ and $\eta_2$ be a history on
    $\bF_{{2}_{\text{Cl}}}$.

    \begin{align*}
        \xi(\h_1 \otimes \h_2)(\eta_1 \otimes \eta_2) &= \sum_{{\scriptsize
            \begin{array}{c}
                \Vec y_1 \in
                \support_1, \Vec y_2 \in \support_2\\ \text{history}(\h_1[\Vec y_1]) =
                \eta_1\\ \text{history}(\h_2[\Vec y_2]) = \eta_2
        \end{array}}
} \bN_1(\Vec y_1) \bN_2(\Vec
    y_2) e^{2 \pi i \bP_1(\Vec y_1)} e^{2 \pi i \bP_2(\Vec y_2)}
    |\F_{1_{\text{Qu}}}\rangle \otimes |\F_{2_{\text{Qu}}}\rangle \\ \\
  &= \phantom{\otimes} \sum_{\Vec y_1 \in \support_1, \text{history}(\h_1[\Vec y_1]) =
        \eta_1} \bN_1(\Vec y_1) e^{2 \pi i \bP_1(\Vec y_1)}
        |\F_{1_{\text{Qu}}}\rangle \\
        & \phantom{=}\otimes \sum_{\Vec y_2 \in \support_2, \text{history}(\h_2[\Vec y_2]) = \eta_2} \bN_2(\Vec y_2) e^{2 \pi i \bP_2(\Vec y_2)} |\F_{2_{\text{Qu}}}\rangle \\
        &= \xi(\h_1)(\eta_1) \otimes \xi(\h_2)(\eta_2)
    \end{align*}
\end{proof}

From Lemma~\ref{lemma:otimes-correct} follows the soundness of the rewriting
rules for $\otimes$: associativity, commutativity, and bilinearity with respect
to $\psadd$:
\begin{itemize}
    \item Associativity:
        \begin{align*}
            \xi((\h_1 \otimes \h_2) \otimes \h_3)((\eta_1 \otimes \eta_2)
            \otimes \eta_3) &= \xi(\h_1 \otimes \h_2)(\eta_1 \otimes \eta_2)
            \otimes \xi(\h_3)(\eta_3) \\
            &= \xi(\h_1)(\eta_1) \otimes \xi(\h_2)(\eta_2)
            \otimes \xi(\h_3)(\eta_3) \\
            &= \xi(\h_1)(\eta_1) \otimes \xi(\h_2 \otimes \h_3)(\eta_2 \otimes
            \eta_3) \\
            &= \xi(\h_1 \otimes (\h_2 \otimes \h_3))(\eta_1 \otimes (\eta_2 \otimes
            \eta_3))
        \end{align*}
        Noting that $\o_1 \otimes (\o_2 \otimes \o_3) = (\o_1 \otimes \o_2)
        \otimes \o_3$.
    \item Commutativity:
        \begin{align*}
            \xi(\h_1 \otimes \h_2)(\eta_1 \otimes \eta_2)
            &= \xi(\h_1)(\eta_1) \otimes \xi(\h_2)(\eta_2) \\
            &= \xi(\h_2)(\eta_2) \otimes \xi(\h_1)(\eta_1) \\
            &= \xi(\h_2 \otimes \h_1)(\eta_2 \otimes \eta_1)
        \end{align*}
        since the registers are indexed by addresses, so the order of the tensor
        product does not matter, and $\eta_1 \otimes \eta_2 = \eta_2 \otimes
        \eta_1$.
    \item Bilinearity:
        \begin{align*}
            \xi((\h_1 \psadd \h_2) \otimes \h_3)(\eta_1 \otimes
            \eta_3) 
            &= \xi(\h_1 \psadd \h_2)(\eta_1) \otimes \xi(\h_3)(\eta_3) \\
            &= (\xi(\h_1)(\eta_1) + \xi(\h_2)(\eta_1)) \otimes \xi(\h_3)(\eta_3) \\
            &= \xi(\h_1)(\eta_1) \otimes \xi(\h_3)(\eta_3) + \xi(\h_2)(\eta_1)
            \otimes \xi(\h_3)(\eta_3) \\
            &= \xi(\h_1 \otimes \h_3)(\eta_1 \otimes \eta_3) + \xi(\h_2 \otimes
            \h_3)(\eta_1 \otimes \eta_3)
        \end{align*}
\end{itemize}

\end{proof}

\begin{theorem}[Monotonicity of $\otimes$, and $\sem{P}{}$ with respect to
    $\prinduce$]%
    \label{theo:monotonicity}
    Let $\h \prinduce \h'$, $P$ be a program that exclusively uses registers in
    $\h'$, and $\h''$ be an arbitrary \hps{} satisfying the conditions for
    forming the tensor product with $\h$. Then,
    \begin{enumerate}
        \item $\h \otimes \h'' \prinduce \h' \otimes \h''$ (also implying $\h''
            \otimes \h' \prinduce \h'' \otimes \h'$), and
        \item $\sem{P}{}(\h) \prinduce \sem{P}{}(\h')$.
    \end{enumerate}
\end{theorem}

\begin{proof}
    Let $\h \prinduce \h'$, then, by definition, $\h = \h' \otimes \h_a$
    with $|\h_a| = 1$ and $\text{var}(\h_a) \cup X = \emptyset$.

    \paragraph{(1)}

    Indeed, $\h \otimes \h'' = (\h' \otimes \h_a) \otimes \h'' = \h' \otimes
    (\h_a \otimes \h'')$. Since $\otimes$ is commutative (it is defined using
    addresses for registers), we have $\h' \otimes (\h_a \otimes \h'') = \h'
    \otimes (\h'' \otimes \h_a) = (\h' \otimes \h'' )\otimes \h_a$. Once again,
    $|\h_a| = 1$ and $\text{var}(\h_a) \cup X = \emptyset$, so $(\h' \otimes
    \h'') \otimes \h_a \prinduce \h' \otimes \h''$.

    \paragraph{(3)}

    We prove that by induction on the structure of $P$ with the induction
    hypothesis: $\h = \h' \otimes \h_a \implies \sem{P}{}(\h) = \sem{P}{}(\h')
    \otimes \h_a$. We have the following cases:

    \begin{itemize}
        \item \textbf{If $P = \textbf{Skip}$: } $\sem{P}{}(\h) = \h = \h'
            \otimes \h_a = \sem{P}{}(\h') \otimes \h_a$.

        \item \textbf{If $P = \textbf{Measure}(q, c)$:} Both $q$ and $c$ are
            registers of $\h'$. In $\sem{P}{}(\h) = \h[c \leftarrow \get q] = (\h'
            \otimes \h_a)[c \leftarrow \get q]$, the change occurs in the
            register $c$ of $\h'$, so $\sem{P}{}(\h) = \h'[c \leftarrow \get q]
            \otimes \h_a = \sem{P}{}(\h') \otimes \h_a$.

        \item \textbf{If $P = c_1 := f(\lceil c_2 \rceil)$:} In \hps{}, this
            is identical to measurement as no hard distinction is made between
            quantum and classical registers.

        \item \textbf{If $P = p; q$:} By induction, $\sem{p}{}(\h) =
            \sem{p}{}(\h') \otimes \h_a$. Since $\sem{q}{}$ only acts on the
            registers of $\h'$, it also only acts on the registers of
            $\sem{p}{}(\h')$, so that we have $\sem{q}{}(\sem{p}{}(\h)) =
            \sem{q}{}(\sem{p}{}(\h')) \otimes \h_a$. Thus, $\sem{p; q}{}(\h) =
            \sem{p; q}{}(\h') \otimes \h_a$.

        \item \textbf{If $P = \for{x^I_n} i j p$:} This follows from the
            sequencing case $P = p; q$ and the unfolding of the for loop at
            instantiation: $\sem{P}{} = \sem{p[x_n^I := i] ; \ldots ;
            p[x_n^I := j]}{}$.

        \item \textbf{If $P = U(q)$:} Since $q$ is a register of $\h'$, the
            phase, norm, support, and changes in the values of the registers $q$
            are all contained in $\h'$ and can be factored out the \hps{}
            $\sem{P}{}(\h)$. More concretely, 
            \begin{itemize}
                \item $(\h' \otimes \h_a)[q \leftarrow b] = \h'[q \leftarrow
                    b] \otimes \h_a$.
                \item $(\h' \otimes \h_a) +_p \P = (\h' +_p \P) \otimes \h_a$
                \item $(\h' \otimes \h_a) *_n \N = (\h' *_n \N) \otimes \h_a$
                \item $(\h' \otimes \h_a) \cup_{\bsupport} y = (\h'
                    \cup_{\bsupport} y) \otimes \h_a$ since the registers in
                    $\h_a$ cannot be made to contain $y$ with this operation
                    that only modifies registers in $\h'$.
            \end{itemize}
            
        \item \textbf{If $P = \Ifthene{\kappa}{p}{q}$:}
            By induction, $\sem{p}{}(\h) = \sem{p}{}(\h') \otimes \h_a$ and
            $\sem{q}{}(\h) = \sem{q}{}(\h') \otimes \h_a$. Therefore,
            $\sem{P}{}(\h) = \kappa \sem{p}{}(\h) + \bar \kappa \sem{q}{}(\h)
            = \kappa \sem{p}{}(\h') \otimes \h_a + \bar \kappa \sem{q}{}(\h')$.
            By linearity of $\otimes$, we then have 

            \[
                \sem{P}{}(\h) = (\kappa \sem{p}{}(\h') + \bar \kappa
                \sem{q}{}(\h')) \otimes \h_a
            \]

            Since $\kappa$ can only concern registers in $\h'$, we do also have
            $\sem{P}{}(\h') = \kappa \sem{p}{}(\h') + \bar \kappa
            \sem{q}{}(\h')$, finally giving $\sem{P}{}(\h) = \sem{P}{}(\h')
            \otimes \h_a$.

    \end{itemize}
    
\end{proof}

   \vfill

    \section{Exhaustive Rule Exposition}%
\label{type-system}

This appendix provides a more detailed and comprehensive description of the
rules of typing system, the equational theory, and the evaluation of \hpps{} to
\hps{} while also describing fully formally the grammars for constructing \hps{}
and \hpps{} terms. Each boxed figure in the following elaborates on one aspect
of these descriptions, and is accompanied by short explanations.

Note: we will use Reg instead of QReg and CReg when the rule is available in
both forms. \CBool{} is closed under $\land$, $\oplus$, and $?\textsc{Bool}$,
and $\QProg$ are closed under \textsc{seq} and \textsc{for}. 

\begin{figure}[htb]
    \centering
    \begin{subfigure}{0.33\textwidth}
    \begin{prooftree}
        \AxiomC{$r$ identifier}
        \RightLabel{\myrule{Ident}}
        \UnaryInfC{$\Gamma \vdash r_Q: \QReg$}
    \end{prooftree}    
    \end{subfigure}
    \begin{subfigure}{0.33\textwidth}
    \begin{prooftree}
        \AxiomC{$r$ identifier}
        \RightLabel{\myrule{Ident}}
        \UnaryInfC{$\Gamma \vdash r_C: \CReg$}
    \end{prooftree}
    \end{subfigure}
    \begin{subfigure}{0.33\textwidth}
        \begin{prooftree}
        \AxiomC{}
        \RightLabel{\myrule{Scratch}}
        \UnaryInfC{$\Gamma \vdash c_0 : \CReg$}
        \end{prooftree}%
        \label{subfig:identifier_typing_rules}
    \end{subfigure}
    \caption{Identifiers typing rules}%
    \label{fig:identifier_typing_rules}
    \Description[]{}
\end{figure}

\begin{figure}[htb]
    \centering
    \begin{subfigure}{0.4\textwidth}
        \begin{prooftree}
            \AxiomC{$\Gamma \vdash r : \Reg$}
            \AxiomC{$\Gamma \vdash i : \Int$}
            \RightLabel{$\myrule{r[i:]}$}
            \BinaryInfC{$\Gamma \vdash r[i:] : \Reg$}
        \end{prooftree}
    \end{subfigure}

    \begin{subfigure}{\textwidth}
        \begin{prooftree}
            \AxiomC{$\Gamma \vdash r : \Reg$}
            \AxiomC{$\Gamma \vdash i : \Int$}
            \RightLabel{$\myrule{r[:i]}$}
            \BinaryInfC{$\Gamma \vdash r[:i] : \Reg$}
        \end{prooftree}
    \end{subfigure}
    
    \begin{subfigure}{\textwidth}
        \begin{prooftree}
            \AxiomC{$\Gamma \vdash r : \Reg$}
            \AxiomC{$\Gamma \vdash i : \Int$}
            \AxiomC{$\Gamma \vdash j : \Int$}
            \RightLabel{$\myrule{r[i:j]}$}
            \TrinaryInfC{$\Gamma \vdash r[i:j] : \Reg$}
        \end{prooftree}
    \end{subfigure}
    \caption{Registers typing rules. These rules to be duplicated, once for
    \CReg{} and once for \QReg.}%
    \label{fig:resg_typing_rules}
        \Description[]{}
\end{figure}

\begin{figure}[htb]
        \centering
        \begin{subfigure}{0.33\textwidth}
            \begin{prooftree}
                \AxiomC{}
                \RightLabel{$\myrule{0_b}$}
                \UnaryInfC{$\Gamma \vdash 0_b : \CBool$}
            \end{prooftree}
        \end{subfigure}
        \begin{subfigure}{0.33\textwidth}
            \begin{prooftree}
                \AxiomC{}
                \RightLabel{$\myrule{1_b}$}
                \UnaryInfC{$\Gamma \vdash 1_b : \CBool$}
            \end{prooftree}
        \end{subfigure}
        \begin{subfigure}{0.33\textwidth}
            \begin{prooftree}
                \AxiomC{$i \in \mathbb N$}
                \RightLabel{\myrule{VarBool}}
                \UnaryInfC{$\Gamma \vdash x^b_i : \CBool$}
            \end{prooftree}
        \end{subfigure}

        \begin{subfigure}{\textwidth}
            \begin{prooftree}
                \AxiomC{$\Gamma \vdash r : \CReg$}
                \AxiomC{$\Gamma \vdash i : \Int$}
                \RightLabel{$\lceil r[i] \rceil$\myrule{-classical}}
                \BinaryInfC{$\Gamma \vdash \lceil r[i] \rceil : \CBool$}
            \end{prooftree}
        \end{subfigure}
        \caption{$\CBool~$typing rules. Classical booleans are those that are
        uniform across a given world issued from constants, program parameters,
        or classical registers. $\CBool$ is a subtype of $\Bool$. Technically,
        three more rules are omitted: $\land$, $\oplus$, and $\text{?}$ which are rules
        for \Bool{} (figure~\ref{fig:bool_typing_rules}) and under which
        \CBool{} is closed.}%
        \label{fig:cbool_typing_rules_app}
\end{figure}

\begin{figure}[htb]
        \begin{subfigure}{0.4\textwidth}
            \begin{prooftree}
                \AxiomC{$\Gamma \vdash b : \CBool$}
                \RightLabel{\myrule{CBool}}
                \UnaryInfC{$\Gamma \vdash b : \Bool$}
            \end{prooftree}
        \end{subfigure}
        \begin{subfigure}{0.4\textwidth}
            \begin{prooftree}
                \AxiomC{$\Gamma \vdash r : \QReg$}
                \AxiomC{$\Gamma \vdash i : \Int$}
                \RightLabel{$\myrule{\lceil r[i] \rceil }$}
                \BinaryInfC{$\Gamma \vdash \lceil r[i] \rceil : \Bool$}
            \end{prooftree}
        \end{subfigure}

        \begin{subfigure}{0.4\textwidth}
            \begin{prooftree}
                \AxiomC{$\Gamma \vdash b_1 : \Bool$}
                \AxiomC{$\Gamma \vdash b_2 : \Bool$}
                \RightLabel{$\myrule{\land}$}
                \BinaryInfC{$\Gamma \vdash b_1 \land b_2 : \Bool$}
            \end{prooftree}
        \end{subfigure}
        \begin{subfigure}{0.4\textwidth}
            \begin{prooftree}
                \AxiomC{$\Gamma \vdash b_1 : \Bool$}
                \AxiomC{$\Gamma \vdash b_2 : \Bool$}
                \RightLabel{$\myrule{\oplus}$}
                \BinaryInfC{$\Gamma \vdash b_1 \oplus b_2 : \Bool$}
            \end{prooftree}
        \end{subfigure}

        \begin{subfigure}{\textwidth}
            \begin{prooftree}
                \AxiomC{$\Gamma \vdash b_1 : \Bool$}
                \AxiomC{$\Gamma \vdash b_2 : \Bool$}
                \AxiomC{$\Gamma \vdash b_3 : \Bool$}
                \RightLabel{\myrule{?Bool}}
                \TrinaryInfC{$\Gamma \vdash b_1 ? b_2 : b_3 : \Bool$}
            \end{prooftree}
        \end{subfigure}
    \centering
    \caption{$\Bool~$ typing rules. $b_1 ? b_2 : b_3$ represents the ternary
    operator $\texttt{ if } b_1 \texttt{ then } b_2 \texttt{ else } b_3$. The
    distinction between $\CBool$ and $\Bool$ is the appearance of quantum
    registers.}%
    \label{fig:bool_typing_rules}
    \Description[]{}
\end{figure}

\begin{figure}[htb]
        \centering
        \begin{subfigure}{0.33\textwidth}
            \begin{prooftree}
                \AxiomC{$n \in \mathbb N$}
                \RightLabel{$\myrule{\mathbb N}$}
                \UnaryInfC{$\Gamma \vdash n : \Int$}
            \end{prooftree}
        \end{subfigure}
        \begin{subfigure}{0.33\textwidth}
            \begin{prooftree}
                \AxiomC{$i \in \mathbb N$}
                \RightLabel{\myrule{VarInt}}
                \UnaryInfC{$\Gamma \vdash x^I_i : \Int$}
            \end{prooftree}
        \end{subfigure}
        \begin{subfigure}{0.33\textwidth}
            \begin{prooftree}
                \AxiomC{$\Gamma \vdash i : \Int$}
                \AxiomC{$\Gamma \vdash j : \Int$}
                \RightLabel{$\myrule{+}$}
                \BinaryInfC{$\Gamma \vdash i + j : \Int$}
            \end{prooftree}
        \end{subfigure}

        \begin{subfigure}{\textwidth}
            \begin{prooftree}
                \AxiomC{$\Gamma \vdash i : \Int$}
                \AxiomC{$\Gamma \vdash j : \Int$}
                \RightLabel{$\myrule{*}$}
                \BinaryInfC{$\Gamma \vdash i * j : \Int$}
            \end{prooftree}
        \end{subfigure}

        \begin{subfigure}{\textwidth}
            \begin{prooftree}
                \AxiomC{$\Gamma \vdash i : \Int$}
                \AxiomC{$\Gamma \vdash j : \Int$}
                \RightLabel{\myrule{Exp}}
                \BinaryInfC{$\Gamma \vdash i^ j : \Int$}
            \end{prooftree}
        \end{subfigure}
        \caption{$\Int~$ typing rules}%
        \label{fig:enter-label}
    \Description[]{}
\end{figure}

\begin{figure}[htb]
        \centering
        \begin{subfigure}{0.33\textwidth}
            \begin{prooftree}
                \AxiomC{$\Gamma \vdash p : \QProg$}
                \RightLabel{\myrule{\QProg}}
                \UnaryInfC{$\Gamma \vdash p : \Prog$}
            \end{prooftree}
        \end{subfigure}
        \begin{subfigure}{0.33\textwidth}
            \begin{prooftree}
                \AxiomC{$\Gamma \vdash q : \QReg$}
                \RightLabel{\myrule{Init}}
                \UnaryInfC{$\Gamma \vdash \textbf{Init}(q) : \Prog$}
            \end{prooftree}
        \end{subfigure}
        \begin{subfigure}{\textwidth}
            \begin{prooftree}
                \AxiomC{$\Gamma \vdash q : \QReg$}
                \AxiomC{$\Gamma \vdash c : \CReg$}
                \RightLabel{\myrule{Measure}}
                \BinaryInfC{$\Gamma \vdash \textbf{Measure}(q, c) : \Prog$}
            \end{prooftree}
        \end{subfigure}

        \begin{subfigure}{\textwidth} 
            \begin{prooftree}
                \AxiomC{$\Gamma \vdash c_1 : \CReg$}
                \AxiomC{$\Gamma \vdash c_2 : \CReg$}
                \RightLabel{\myrule{Classical}}
                \BinaryInfC{$\Gamma \vdash c_1 := f(\lceil c_2 \rceil) : \Prog$}
            \end{prooftree}            
        \end{subfigure}

        \begin{subfigure}{\textwidth}
            \begin{prooftree}
                \AxiomC{$\Gamma \vdash p : \Prog$}
                \AxiomC{$\Gamma \vdash q : \Prog$}
                \RightLabel{\myrule{seq}}
                \BinaryInfC{$\Gamma \vdash p;q : \Prog$}
            \end{prooftree}
        \end{subfigure}

        \begin{subfigure}{\textwidth}
            \begin{prooftree}
                \AxiomC{$\Gamma \vdash i : \Int$}
                \AxiomC{$\Gamma \vdash j : \Int$}
                \AxiomC{$\Gamma \vdash p : \Prog$}
                \RightLabel{\myrule{for}}
                \TrinaryInfC{$\Gamma \vdash \textbf{For } x_n^I = i,j \textbf{ do } \{ p\}: \Prog$}
            \end{prooftree}
        \end{subfigure}

        \begin{subfigure}{\textwidth}
            \begin{prooftree}
                \AxiomC{$\Gamma \vdash b : \CBool$}
                \AxiomC{$\Gamma \vdash p : \Prog$}
                \AxiomC{$\Gamma \vdash q : \Prog$}
                \RightLabel{\myrule{C-if}}
                \TrinaryInfC{$\Gamma \vdash \Ifthene{b}{p}{q}: \Prog$}
            \end{prooftree}
        \end{subfigure}
        \caption{$\Prog~$ typing rules}%
        \label{fig:prog_typing_rules_app}
    \Description[]{}
\end{figure}
\begin{figure}
        \centering
        \begin{subfigure}{0.43\textwidth}
            \begin{prooftree}
                \AxiomC{}
                \RightLabel{\textsc{\myrule{Skip}}}
                \UnaryInfC{$\Gamma \vdash \textbf{Skip} : \QProg$}
            \end{prooftree}
        \end{subfigure}
        \begin{subfigure}{0.43\textwidth}
            \begin{prooftree}
                \AxiomC{$U \in \mathcal G$}
                \AxiomC{$\Gamma \vdash q : \QReg$}
                \RightLabel{\textcolor{re}{\myrule{Unitary}}\color{black}}
                \BinaryInfC{$\Gamma \vdash \Apply\ U(q) : \QProg$}
            \end{prooftree}
        \end{subfigure}

        \begin{subfigure}{\textwidth}
            \begin{prooftree}
                \AxiomC{$\Gamma \vdash \qreg (p) \cap\qreg(b) = \emptyset$}
                \AxiomC{$\Gamma \vdash b : \Bool$}
                \AxiomC{$\Gamma \vdash p : \QProg$\; $\Gamma \vdash q : \QProg$}
                \RightLabel{\myrule{Q-if}}
                \TrinaryInfC{$\Gamma \vdash \Ifthene{b}{p}{q} :\QProg$}
            \end{prooftree}
        \end{subfigure}
        \caption{$\QProg~$ typing rules. A \QProg{} is a unitary program.
        Similar to \CBool{}, \QProg{} is a subtype of \Prog{} and is closed
        under sequencing the sequencing operator $\text{;}$ (1 rule omitted).}%
        \label{fig:qprog_typing_rules}
\end{figure}

\begin{figure}[htb]
    \centering

  \begin{center}      
    \AxiomC{}
    \RightLabel{\myrule{Ax}}
    \UnaryInfC{$\Gamma, x : T \vdash x : T$}
\DisplayProof\hskip 1em
    \AxiomC{$\Gamma \vdash \Delta$}
    \RightLabel{\myrule{Weakening}}
    \UnaryInfC{$\Gamma, \Gamma' \vdash \Delta, \Delta'$}
\DisplayProof\vskip 1em
    \AxiomC{$\Gamma, x : T \vdash t : U$}
    \RightLabel{\textcolor{re}{$\lambda $}\color{black}}
    \UnaryInfC{$\Gamma \vdash \lambda x. t : T \to U$}
\DisplayProof\hskip 1em
    \AxiomC{$\Gamma \vdash f : T \to U$}
    \AxiomC{$\Gamma \vdash t : T$}
    \RightLabel{\myrule{@}}
    \BinaryInfC{$\Gamma \vdash f t : U$}
\DisplayProof%
\end{center}

    \caption{Typing rules for the $\lambda-$calculus.}%
    \label{fig:lambda_calculus_and_primitive_gates_typing_rules_main}
\end{figure}

    \begin{figure}[htb]
        \centering
\scalebox{1}{
$
\begin{array}{lcll}
L[R] &:=& \emptyset_{L[R]} \mid R \cdot L[R] &\mathit{LIST\ OF\ R\ SCHEME} \\
\midrule

\creg&:=&  \texttt{ident}^C \mid \creg[\I:] \mid \creg[:\I] \mid
\creg[\I_1:\I_2]& Quantum\ registers\\
\qreg&:=&  \texttt{ident}^Q \mid \qreg[\I:] \mid \qreg[:\I] \mid
\qreg[\I_1:\I_2]& Classical\ registers\\
\B_1, \B_2 &:=&  0 \mid 1 \mid y_\I \mid\get{\cqreg[\I]}  \mid
\B_1\wedge \B_2\mid \B_1 \oplus \B_2  &\mathit{Booleans}\\
\I_1, \I_2&:=&n \mid  \length \cqreg \mid \I_1 +_i\I_2\mid\I_1 *\I_2\mid\I_1^{\I_2} &\mathit{Integers}\\

\support_1, \support_2&:=& \emptyset \mid \{y_\I\} \mid \support_1 \cup \support_2 \mid \support_1\setminus \support_2 
& \textit{\hps support}\\
\P, \P_1, \P_2&:=& 0_{\P}\mid \frac{\B}{2^{\I}}\mid \frac{\I}{2^{\I}}\mid \P_1 +_{\P} \P_2 \mid \I\cdot\P\mid \B \cdot \P\mid -_{\P} \P &\mathit{Dyadics}\\
\An, \An_1, \An_2&:=& 2\pi \frac{\I_1}{\I_2}\mid \arccos \Reals \mid \arcsin
\Reals \mid \I *_a \An\mid \B *_a \An\mid &\mathit{Angles}\\
                                     &&\An_1 +_a \An_2 \mid - \An\\
\Reals, \Reals_1, \Reals_2&:=& \frac{\I_1}{\I_2} \mid \sqrt{\Reals}\ \mid \cos
\An \mid  \sin \An \mid \Reals_1 +_r \Reals_2 \mid  \Reals_1 *_r \Reals_2 \mid
\frac{1}{\Reals} &\mathit{Constructibles}\\
\qmem &:=& L[|\B\rangle_{\qreg[\I]}] & Quantum\ memory \\
\cmem &:=& L[{[\B]}_{\creg[\I]}] & Classical\ memory \\
\stack{\cmem} &:=& L[\cmem] & \cmem \mathit{\ stack} \\
\F&:=& \qmem{} \cdot \stack \cmem &\mathit{Output}\\
\h & := & \pasum{\P}{\F}{\N}_{\support} &\hps\\
\end{array}
$}

\caption{The grammar for \hps{} terms}%
        \label{fig:hps_grammar}
    \end{figure}

\begin{figure}[!h]
\begin{prooftree}
    \AxiomC{}
\RightLabel{(Add-L-distr)}
\UnaryInfC{$(a + b)\bh\prequiv a\bh \psadd b\bh$}
 \DisplayProof\hskip 1.5em
\AxiomC{}
\RightLabel{(Add-R-distr)}
\UnaryInfC{$a(\bh_1\psadd\bh_2)\prequiv a\bh_1 \psadd a\bh_2$}
 \DisplayProof\vskip 1.5em
\AxiomC{}
\RightLabel{(Add-comm)}
\UnaryInfC{$\bh_1 \psadd\bh_2\prequiv \bh_2 \psadd \bh_1$}
 \DisplayProof\hskip 1.5em
\AxiomC{}
\RightLabel{(Add-assoc)}
\UnaryInfC{$\bh_1\psadd(\bh_2\psadd\bh_3)\prequiv (\bh_1\psadd\bh_2)\psadd\bh_3$}
 \DisplayProof\vskip 1.5em
\AxiomC{$\bh_2\prequiv\bh_3$}
\RightLabel{(Local-Add-rew)}
\UnaryInfC{$\bh_1\psadd\bh_2\prequiv \bh_1\psadd\bh_3$}
 \DisplayProof\hskip 1.5em
\AxiomC{}
\RightLabel{(Neutral-add)}
\UnaryInfC{$ \bh \psadd \pasum{x_{\bP}}{x_{\bF}}{0} \prequiv \bh$}
 \DisplayProof\vskip 1.5em
\AxiomC{$\bsupport_1 \cap\bsupport_2 = \emptyset$}
\AxiomC{$\oad_1\cap \oad_2 = \emptyset$}
\RightLabel{(Fact-Distr)}
 \BinaryInfC{$\pasum{\bP_1}{\bF_1}{\bN_1}_{\bsupport_1} \pstens \pasum{\bP_2}{\bF_2}{\bN_2}_{\bsupport_2} \prequiv \pasum{\bP_1+\bP_2}{\bF_1\cup\bF_2}{\bN_1*\bN_2}_{\bsupport_1\cup\bsupport_2}$}
 \DisplayProof\vskip 1.5em
\AxiomC{(Fact-Distr) preconditions }
\RightLabel{(Fact-comm)}
\UnaryInfC{$\bh_1 \pstens\bh_2\prequiv \bh_2 \pstens \bh_1$}
 \DisplayProof\vskip 1.5em
\AxiomC{(Fact-Distr) preconditions}
\RightLabel{(Fact-assoc)}
\UnaryInfC{$\bh_1\pstens(\bh_2\pstens\bh_3)\prequiv (\bh_1\pstens\bh_2)\pstens\bh_3$}
 \DisplayProof\vskip 1.5em
\AxiomC{$\bh\prequiv\bh''$}
\AxiomC{(Fact-Distr) preconditions}
\RightLabel{(Local-Fact-rew)}
\BinaryInfC{$\bh_1\pstens\bh_2\prequiv \bh_1\pstens\bh3$}
\end{prooftree}%
\label{fig:rewriting_rules}

\caption{Additional rules of the equational theory stemming from the axioms of
the algebraic structure of $\psadd$ and $\otimes$}
\end{figure}

    \end{document}